%% file: main.tex
\documentclass[11pt]{article}
\usepackage{yaonan}
\usepackage{array}
\pdfoutput=1

\newcommand{\yj}[1]{{\color{red} {{\bf YJ}: #1}}}
\newcommand{\yfnote}[1]{\footnote{{\color{blue}YF}: {#1}}}
\newcommand{\yf}[1]{{\color{blue} {{\bf YF}: #1}}}
\renewcommand{\Comment}[1]{{\color{OliveGreen} $\triangleright$ {#1}}}

\newcommand{\BayesianMyersonAuction}{\textsf{Bayesian Myerson Auction}}
\newcommand{\BayesianOptimalMechanism}{\textsf{Bayesian Optimal Mechanism}}
\newcommand{\BayesianOptimalSequentialPricing}{\textsf{Bayesian Optimal Sequential Pricing}}

\newcommand{\BayesianOptimalUniformReserve}{\textsf{Bayesian Optimal Uniform Reserve}}
\newcommand{\BayesianOptimalUniformPricing}{\textsf{Bayesian Optimal Uniform Pricing}}

\newcommand{\BayesianMonopolyReserves}{\textsf{Bayesian Monopoly Reserves}}
\newcommand{\BayesianEagerMonopolyReserves}{\textsf{Bayesian Eager Monopoly Reserves}}
\newcommand{\BayesianLazyMonopolyReserves}{\textsf{Bayesian Lazy Monopoly Reserves}}
\newcommand{\VCGAuction}{\textsf{VCG Auction}}
\newcommand{\EmpiricalUniformReserve}{\textsf{Empirical Uniform Reserve}}
\newcommand{\EmpiricalMyersonAuction}{\textsf{Empirical Myerson Auction}}

\newcommand{\BOM}{{\sf BOM}}
\newcommand{\BOSP}{{\sf BOSP}}
\newcommand{\BOUR}{{\sf SPA}_{\sf BOUR}}
\newcommand{\BOUP}{{\sf BOUP}}
\newcommand{\BMR}{{\sf VCG}_{\sf BMR}}
\newcommand{\BEMR}{{\sf VCG}_{\sf BEMR}}
\newcommand{\BLMR}{{\sf VCG}_{\sf BLMR}}

\newcommand{\IdentityPricing}{\textsf{Identity Pricing}}
\newcommand{\SecondPriceAuction}{\textsf{Second Price Auction}}
\newcommand{\oneDuplicateSecondPriceAuction}{\textsf{$1$-Duplicate Second Price Auction}}
\newcommand{\nDuplicateVCGAuction}{\textsf{$n$-Duplicate VCG Auction}}

\newcommand{\IP}{{\sf IP}}
\newcommand{\SPA}{{\sf SPA}}
\newcommand{\SPAone}{{\sf SPA}_{\oplus 1}}
\newcommand{\VCG}{{\sf VCG}}
\newcommand{\VCGn}{{\sf VCG}_{\oplus n}}

\newcommand{\EMA}{\textsf{EMA}}
\newcommand{\EUR}{\textsf{EUR}}
\newcommand{\distspace}{\mathbb{F}}
\newcommand{\regulardistspace}{\distspace_{\textnormal{\tt reg}}}
\newcommand{\quasiregulardistspace}{\distspace_{\textnormal{\tt Q-reg}}}
\newcommand{\mhrdistspace}{\distspace_{\textnormal{\tt MHR}}}
\newcommand{\quasimhrdistspace}{\distspace_{\textnormal{\tt Q-MHR}}}
\newcommand{\generaldistspace}{\distspace_{\textnormal{\tt gen}}}

\input{source/macros}

\input{source/math}

\title{Beyond Regularity: Simple versus Optimal Mechanisms, Revisited}

\author{
Yiding Feng\thanks{Hong Kong University of Science and Technology. Email: {\tt ydfeng@ust.hk}}
\and
Yaonan Jin\thanks{Huawei TCS Lab. Email: {\tt jinyaonan@huawei.com}}
}
\date{}

\begin{document}

\maketitle
\begin{abstract}
A large proportion of the Bayesian mechanism design literature is restricted to the family of {\em regular} distributions $\regulardistspace$ \cite{M81} or the family of {\em monotone hazard rate (MHR)} distributions $\mhrdistspace$ \cite{BMP63}, which overshadows this beautiful and well-developed theory.
We (re-)introduce two generalizations, the family of {\em quasi-regular} distributions $\quasiregulardistspace$ and the family of {\em quasi-MHR} distributions $\quasimhrdistspace$.
All four families together form the following hierarchy:
\[
    \mhrdistspace \subsetneq (\regulardistspace \cap \quasimhrdistspace) \subsetneq \regulardistspace,\ \quasimhrdistspace \subsetneq (\regulardistspace \cup \quasimhrdistspace) \subsetneq \quasiregulardistspace.
\]
The significance of our new families is manifold.
First, their defining conditions are immediate relaxations of the regularity/MHR conditions (i.e., monotonicity of the virtual value functions and/or the hazard rate functions), which reflect economic intuition.
Second, they satisfy natural mathematical properties (about order statistics) that are violated by both original families $\regulardistspace$ and $\mhrdistspace$.
Third but foremost, numerous results \cite{BK96, HR09, CD15, DRY15, HR14, AHNPY19, JLTX20, JLQTX19, FLR19, GHZ19, JLX23, LM24} established before for regular/MHR distributions now can be generalized, with or even without quantitative losses.
\end{abstract}

\thispagestyle{empty}
\newpage

{\hypersetup{linkcolor=black}\tableofcontents}
\thispagestyle{empty}
\newpage
\setcounter{page}{1}

\newpage

\input{source/intro}
% \newpage
\input{source/prelim}

% \newpage
\input{source/structural}

% \newpage
\input{source/simple_mechanism}

\input{source/duplicate}

% \newpage
\input{source/single_sample}
% \newpage
\input{source/sample_complexity}

\input{source/conclusion}

\section*{Acknowledgements}
We are grateful to Zhiyi Huang for invaluable discussions, Hu Fu, Benjamin Golub, Jason Hartline, Pinyan Lu, Rad Niazadeh, and Tim Roughgarden for helpful comments on a preliminary version of this paper, and anonymous reviewers for their dedication to reading this paper.

\begin{flushleft}
\bibliographystyle{alphaurl}
\bibliography{main}
\end{flushleft}

\appendix

\input{source/prophet_inequality}

\end{document}

%% file: source/macros.tex
\usepackage{enumitem}

\usepackage{tikz}
\usetikzlibrary{arrows.meta}
\usepackage{pgfplots}
\pgfplotsset{compat=1.18} 

\newcommand{\feasiblealloc}{\mathcal{X}}

\newcommand{\xhdr}[1]{\vspace{2mm} \noindent{\bf #1}}

\usepackage[suppress]{color-edits}  % use this to suppress the package
\addauthor[Yiding]{yf}{purple}    % yf for Yiding
\newcommand{\event}{\mathcal{E}}

\newcommand{\val}{v}
\newcommand{\vali}{\val_i}
\newcommand{\vals}{\mathbf{\val}}
\newcommand{\prior}{F}
\newcommand{\priori}{\prior_i}
\newcommand{\alloc}{x}
\newcommand{\price}{p}
\newcommand{\virtualval}{\varphi}

\newcommand{\cdf}{F}
\newcommand{\pdf}{\cdf'}
\newcommand{\reals}{\mathbb{R}}
\newcommand{\revcurve}{R}
\newcommand{\ironrevcurve}{\Bar{\revcurve}}
\newcommand{\quant}{q}
\newcommand{\primed}{^\dagger}
\newcommand{\doubleprimed}{^\ddagger}
\newcommand{\ironvirtualval}{\Bar\virtualval}

\newcommand{\Unif}{\text{U}}
\newcommand{\quants}{\boldsymbol{\quant}}
\newcommand{\priors}{\boldsymbol{\prior}}
\newcommand{\prices}{\boldsymbol{\price}}

\newcommand{\reserve}{r}
\newcommand{\order}{\sigma}

\newcommand{\optquant}{\quant^*}
\newcommand{\optreserve}{\reserve^*}
\newcommand{\optorder}{\order^*}
\newcommand{\optprice}{\price^*}
\newcommand{\optprices}{\prices^*}
\newcommand{\optrev}{\revcurve^*}

\newcommand{\Rev}[2][]{\text{\bf Rev}\ifthenelse{\not\equal{}{#1}}{_{#1}}{}\!\left[{\def\givenn{\middle|}#2}\right]}

\newcommand{\sample}{s}
\newcommand{\samples}{\boldsymbol{s}}
\newcommand{\firstsample}{\sample_{(1:n)}}
\newcommand{\firstprior}{\prior_{(1:n)}}

\newcommand{\mech}{{M}}

\newcommand{\naturals}{\mathbb{N}}

\newcommand{\hazardrate}{h}
\newcommand{\cumvirtualval}{\cevirtualval}
\newcommand{\cevirtualval}{\phi_{\tt CE}}
\newcommand{\cehazardrate}{h_{\tt CE}}
\newcommand{\hval}{\bar{v}}

\newcommand{\auxprior}{T}
\newcommand{\auxcdf}{\auxprior}
\newcommand{\auxpdf}{\auxprior'}

\newcommand{\valkn}{\val_{(k:n)}}
\newcommand{\cdfkn}{\cdf_{(k:n)}}

\newcommand{\cumhazardrate}{H}

\newcommand{\auxcumhazardrate}{U}

\newcommand{\auxdpdf}{\auxcdf''}
\newcommand{\auxfunc}{L}

\newcommand{\auxpriors}{\boldsymbol{\auxprior}}
\newcommand{\buyerorder}{\sigma}
\newcommand{\SP}{{\sf SP}}
\newcommand{\UP}{{\sf UP}}

\newcommand{\LambertFunc}{W_{-1}}

\newcommand{\allocs}{\boldsymbol{\alloc}}
\newcommand{\optallocs}{\allocs^*}
\newcommand{\optalloc}{\alloc^*}
\newcommand{\optreserves}{\reserves^*}
\newcommand{\reserves}{\boldsymbol{r}}
\newcommand{\threshold}{t}

\newcommand{\detquant}{\Delta_\quant}

%% file: source/math.tex
\newcommand{\given}{\,\mid\,}

\newcommand{\prob}[2][]{\text{Pr}\ifthenelse{\not\equal{}{#1}}{_{#1}}{}\!\left[{\def\givenn{\middle|}#2}\right]}
\newcommand{\expect}[2][]{\mathbb{E}\ifthenelse{\not\equal{}{#1}}{_{#1}}{}\!\left[{\def\givenn{\middle|}#2}\right]}

\newcommand{\indicator}[2][]{\text{\bf 1}\ifthenelse{\not\equal{}{#1}}{_{#1}}{}\!\left[{\def\givenn{\middle|}#2}\right]}

\newcommand{\tparen}{\big}
\newcommand{\tprob}[2][]{\text{Pr}\ifthenelse{\not\equal{}{#1}}{_{#1}}{}\tparen[{\def\given{\tparen|}#2}\tparen]}
\newcommand{\texpect}[2][]{\mathbb{E}\ifthenelse{\not\equal{}{#1}}{_{#1}}{}\tparen[{\def\given{\tparen|}#2}\tparen]}

\newcommand{\sprob}[2][]{\text{Pr}\ifthenelse{\not\equal{}{#1}}{_{#1}}{}[#2]}
\newcommand{\sexpect}[2][]{\mathbb{E}\ifthenelse{\not\equal{}{#1}}{_{#1}}{}[#2]}

\newcommand{\plus}[1]{{\left( #1 \right)^+}}

%% file: source/intro.tex
\section{Introduction}
\label{sec:intro}

Bayesian mechanism design, a traditional branch of economics and game theory, has also become a core subdomain of theoretical computer science over the last two decades.
Specifically, the TCS community has brought/renewed several crucial principles of this discipline, \yfedit{such as} (i)~an emphasis on approximation guarantees, (ii)~simplicity of mechanism rules, and (iii)~robustness to information assumptions. (See the surveys \cite{HR09survey, CS14survey, R14survey, JLQTX19survey, GHZ19survey} and the textbook \cite{har16} for an overview.)
In the canonical {\em single-parameter revenue maximization} context \cite{M81}, among the most mainstream research topics include:\footnote{\label{footnote:truthful}All mechanisms considered in this paper are {\em truthful} (or more precisely, {\em dominant-strategy incentive-compatible}), so we freely interchange {\em values} and {\em bids}.}
\begin{itemize}
    \item {\bf Approximation by Bayesian Simple Mechanisms:}
    \yfedit{By a line of works initiated by Hartline and Roughgarden \cite{HR09}, Bayesian mechanisms with elementary rules already can be approximately optimal \cite{HR09, CHMS10, CEGMM10a, Y11, A14, DFK16, AHNPY19, CFHOV21, JLTX20, JLQTX19, JLQ19, JJLZ22, GPZ21, BGLPS21, JMZ22, PT22, BC23}.}
    
    For example, {\BayesianOptimalUniformPricing},\footnote{Namely, post a uniform price $\optprice = \optprice(\priors)$ to all buyers, which is optimized based on the Bayesian information $\priors$.} which arguably is the most straightforward and practical mechanism, achieves a tight $\approx 0.3817$-approximation to the optimal mechanism, {\BayesianMyersonAuction} \cite{AHNPY19, JLTX20, JLQTX19}.

    \item {\bf Approximation by Prior-Independent Mechanisms:}
    \yfedit{By a line of works initiated by Bulow and Klemperer \cite{BK96}, prior-independent mechanisms that need not utilize the Bayesian information (possibly after slight modifications of models) already can be approximately optimal \cite{BK96, HR09, DRS12, DRY15, RTY20, FLR19, JLQ19, FILS15, AB20, HJL20}.}
    
    For example, {\SecondPriceAuction} with $(n + 1)$ many i.i.d.\ buyers surpasses {\BayesianMyersonAuction} (that fully leverages the Bayesian information) with $n$ such buyers \cite{BK96}.
    
    \item {\bf Approximation with a Single Sample (or a Few Samples):}
    By a line of works initiated by Dhangwatnota, Roughgarden, and Yan \cite{DRY15}, a single or a few sample(s) from the underlying distributions are sufficient to design approximately optimal mechanisms \cite{DRY15, HMR18, FILS15, BGMM18, RTY20, DZ20, RTY20, ABB22, FHL21}.

    For example, given a single buyer and a single sample from his/her value distribution, {\IdentityPricing}\footnote{Namely, post a price $\price = \sample$ to the buyer, which is precisely the sampled value $\sample \sim \prior$.} achieves a tight $\frac{1}{2}$-approximation to {\BayesianMyersonAuction}; this is the best possible among deterministic mechanisms \cite{DRY15, HMR18}.

    \item {\bf The Sample Complexity of Revenue Maximization:}
    \yfedit{By a line of works initiated by Cole and Roughgarden \cite{CR14}, a number of $\poly(n \cdot \eps^{-1})$ samples are sufficient to learn an empirical $(1 - \eps)$-approximation of {\BayesianMyersonAuction} \cite{MR15, CGM15, MM16, RS16, DHP16,  GN17, S17, GHZ19, HT19, CHMY23, LSTW23, JLX23}.}
\end{itemize}
Despite being a beautiful and well-developed theory, a large proportion of the Bayesian mechanism design literature, including all the mentioned results, requires the value distributions to satisfy the (relatively weaker) {\em regularity} condition \cite{M81} or the (relatively stronger) {\em monotone hazard rate (MHR)} condition \cite{BMP63}.\footnote{In some sense, we have trapped into such cognitive inertia (which our work tries to break): that regularity/MHR are necessary conditions for those results. But, after all, they are just sufficient conditions.}
% (Without distributional assumptions, the above topics cannot have meaningful results.)
\yfedit{Rather, without any distributional assumption, all mentioned results and others would be impossible.}
% Although very standard, both conditions intrinsically violate the ``robustness to information assumptions'' principle.
\yfedit{Although the regularity/MHR conditions are very standard, their appropriate generalizations that still enable meaningful results (if possible) would better befit the ``robustness to information assumptions'' principle.}

On the other hand, irregular (and thus non-MHR) distributions are ordinary in practice.
For example, the {\em equal-revenue distribution}
% \yjnote{Is there any measure of probability distance for the example here?}
-- a paradigmatic regular distribution -- after a slight modification $\cdf_{1}(\val) = (1 - 1 / \val) \cdot \indicator{\val \geq 2}$ would be irregular.\footnote{\label{footnote:quasi-regular/-mhr:example}The quasi-regular (but irregular) distribution $\cdf_{1}(\val) = (1 - 1 / \val) \cdot \indicator{\val \geq 2}$ is the first order statistic of two regular distributions $\cdf_{11}(\val) = 1 - 1 / \val$ for $\val \geq 1$ and $\cdf_{12}(\val) = \indicator{\val \geq 2}$. Likewise, the quasi-MHR (but irregular and thus non-MHR) distribution $\cdf_{2}(\val) = (1 - e^{-\val}) \cdot \indicator{\val \geq 1}$ is the first order statistic of two MHR distributions $\cdf_{21}(\val) = 1 - e^{-\val}$ for $\val \geq 0$ and $\cdf_{22}(\val) = \indicator{\val \geq 1}$.} Namely, a random value $\val_{1} \sim \prior_{1}$ almost follows the equal-revenue distribution, except for a minimum possible value of $2$ (rather than $1$).
Likewise, the {\em exponential distribution} -- a paradigmatic MHR distribution -- after a slight modification $\cdf_{2}(\val) = (1 - e^{-\val}) \cdot \indicator{\val \geq 1}$ also would be irregular (and thus non-MHR).\textsuperscript{\ref{footnote:quasi-regular/-mhr:example}}
Unfortunately, all the mentioned results (and others) that rely on the regularity/MHR conditions cannot accommodate such ordinary distributions.

The above issues overshadow the whole theory of Bayesian mechanism design.
To the rescue, we shall go beyond the regularity/MHR conditions and seek their appropriate relaxations/surrogates. (Unfortunately, very few works have proceeded in this direction.\footnote{\label{footnote:bounded-irregularity}To our knowledge, there are (at most) three exceptions \cite{SS13, HMR18, HR14}.
However, the first two \cite{SS13} and \cite[Section~3.3]{HMR18} studied rather specialized settings and made other (strong) assumptions about their models\ignore{, which restrict their extensibility}.
The last one \cite[Appendix~D]{HR14} just gave a light discussion, and their relaxed/surrogate conditions fit nicely in our framework; we will elaborate on the relation between \cite[Appendix~D]{HR14} and our work in \Cref{sec:intro:contribution}. Besides, \cite{CR14} introduced the $\alpha$-strong regularity \cite{CR14}, which is stronger than regularity but weaker than MHR condition.})
In this way, we will strengthen the applicability of this theory and, conversely, better reveal the essence thereof.
However, before all else, we must stipulate several criteria for the new conditions.
\begin{enumerate}
    \item\label{criterion:generalization} The new conditions shall be \textbf{\em as general as possible} -- at least, they must contain the family of regular/MHR distributions -- but still involve ``a certain degree of regularity'' that \textbf{\em enables generalizations of previous results}, with or even without quantitative losses.
    
    \item\label{criterion:econ} The new conditions must involve \textbf{\em strong economic intuition}, akin to the original ones.\footnote{Recall that the regularity/MHR conditions refer to the monotonicity of the virtual value functions and the hazard rate functions, respectively; see \Cref{sec:intro:contribution,sec:prelim} for more details.}
    
    \item\label{criterion:math} Ideally, the new conditions shall also be \textbf{\em mathematically natural and elegant}.
\end{enumerate}
Seeking such conditions and (nontrivially) generalizing previous results are the themes of our work.

\subsection{Our Contributions}
\label{sec:intro:contribution}

In this subsection, we first review the regularity/MHR conditions and introduce their relaxations/ surrogates, the {\em quasi-regular/-MHR} conditions.
Afterward, we will elaborate on our generalizations of previous results, from the original distribution families to our new ones, for every research topic mentioned at the beginning of the introduction.

\subsection*{Distribution Hierarchy and Structural Results \texorpdfstring{(\Cref{sec:structural})}{}}

The regularity/MHR conditions are specified as the monotonicity of the {\em virtual value functions} and the {\em hazard rate functions}, respectively  \cite{M81, BMP63}; both definitions intrinsically involve strong economic intuition (Criterion~\ref{criterion:econ}).
In the same spirit, we specify the {\em quasi-regularity/-MHR} conditions as immediate relaxations/analogs of such monotonicity, as follows:
% \begin{flushleft}
\begin{itemize}
    \item {\bf Regularity:}
    A {\em regular} distribution $\prior$ has its {\em virtual value function} $\virtualval(\val) \triangleq \val - \frac{1 - \cdf(\val)}{\pdf(\val)}$ being increasing.
    Let $\regulardistspace$ denote this distribution family.

    \item {\bf Monotone Hazard Rate (MHR):} A {\em monotone hazard rate (MHR)} distribution $\prior$ has its {\em hazard rate function} $\hazardrate(\val) \triangleq \frac{\pdf(\val)}{1 - \cdf(\val)}$ being increasing.
    Let $\mhrdistspace$ denote this distribution family.
    
    \item {\bf Quasi Regularity:}
    A {\em quasi-regular} distribution $\prior$ has its {\em conditional expected virtual value function} $\cevirtualval(\val) \triangleq \expect[x \sim \prior]{\virtualval(x) \given x \leq \val} = -\val \cdot \frac{1 - \cdf(\val)}{\cdf(\val)}$ being increasing.
    Let $\quasiregulardistspace$ denote this distribution family.

    \item {\bf Quasi Monotone Hazard Rate (Quasi-MHR):}
    A {\em quasi monotone hazard rate (quasi-MHR)} distribution $\prior$ has its {\em conditional expected hazard rate function} $\cehazardrate(\val) \triangleq \frac{1}{\val}\int_{0}^{\val} \hazardrate(x) \cdot \d x = \frac{-\ln(1 - \cdf(\val))}{\val}$ being increasing.\footnote{The $\cehazardrate(\val)$ may be better called the {\em ``normalized cumulative'' hazard rate function}. Yet we prefer its current name since it incurs no ambiguity and is more consistent with the {\em ``conditional expected'' virtual value function} $\cevirtualval(\val)$.}
    Let $\quasimhrdistspace$ denote this distribution family.
\end{itemize}
% \end{flushleft}
So the quasi-regularity condition relaxes the regularity condition to a large extent, from {\em pointwise monotonicity} of the virtual value function $\virtualval(\val)$ to {\em on-average monotonicity}.
Consequently, a quasi-regular distribution may have (up to countably) many probability masses spreading its support. In contrast, a regular distribution has at most one probability mass (which must be at the support supremum).
The circumstances of the ``MHR to quasi-MHR'' relaxation are similar.

For more economic intuition (Criterion~\ref{criterion:econ}), recall that the regularity (resp.\ MHR) condition is equivalent to the concavity of the {\em revenue curve} (resp.\ the convexity of the {\em cumulative hazard rate function}).
In the same spirit, we will get similar geometric interpretations of the quasi-regularity/ -MHR conditions; see \Cref{sec:prelim,sec:structural} for more details.

All four families together form the following hierarchy (\Cref{fig:distribution-hierarchy}), where the $\generaldistspace$ denotes the family of general (single-dimensional) distributions.
\[
    \mhrdistspace \subsetneq (\regulardistspace \cap \quasimhrdistspace) \subsetneq \regulardistspace,\ \quasimhrdistspace \subsetneq (\regulardistspace \cup \quasimhrdistspace) \subsetneq \quasiregulardistspace \subsetneq \generaldistspace.
\]
Every inclusion here is strict (i.e., $\regulardistspace$ is neither a subset nor a superset of $\quasimhrdistspace$); see \Cref{sec:hierarchy:relation} for more details, e.g., examples in these families and their intersections, unions, and differences.

\tikzset{
    partial ellipse/.style args={#1:#2:#3}{
        insert path={+ (#1:#3) arc (#1:#2:#3)}
    }
}

\def\hierarchyheight{0.7}

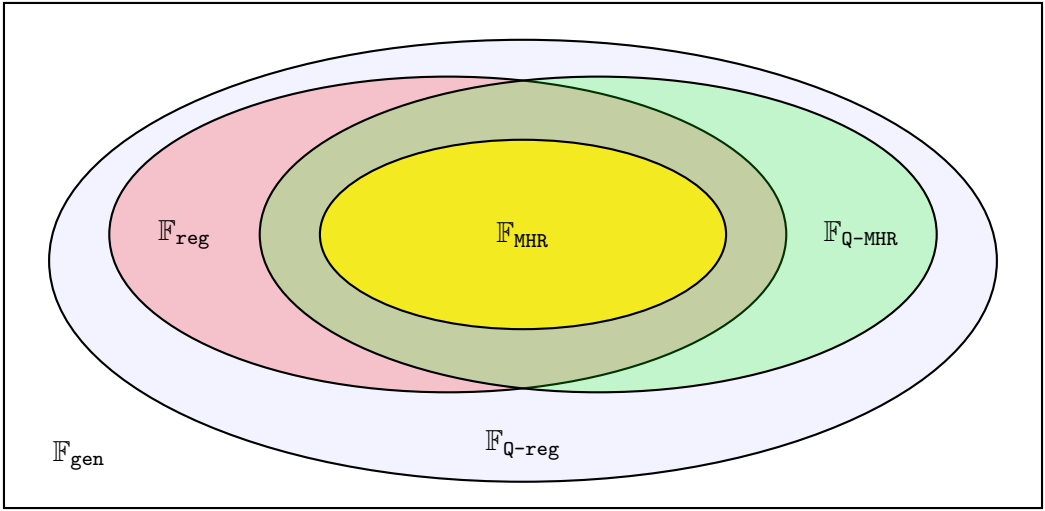
\begin{figure}[t]
    \centering
    \begin{tikzpicture}
        \draw[thick] (-6.9, -5.2 * \hierarchyheight) rectangle ++(13.8, 9.6 * \hierarchyheight);
        \draw[thick, fill = blue, fill opacity = 0.05] (0, -0.5 * \hierarchyheight) ellipse (6.3cm and 4.2cm * \hierarchyheight);
        \draw[thick, fill = red, fill opacity = 0.2] (-1, 0) ellipse (4.5cm and 3cm * \hierarchyheight);
        \draw[thick, fill = green, fill opacity = 0.2] (1, 0) ellipse (4.5cm and 3cm * \hierarchyheight);
        \draw[thick, fill = yellow, fill opacity = 0.8] (0, 0) ellipse (2.7cm and 1.8cm * \hierarchyheight);
        \node at (-6.9 + 1, -5.2 * \hierarchyheight + 1 * \hierarchyheight) {$\generaldistspace$};
        \node at (0, -4 * \hierarchyheight) {$\quasiregulardistspace$};
        \node at (-4.5, 0) {$\regulardistspace$};
        \node at (4.5, 0) {$\quasimhrdistspace$};
        \node at (0, 0) {$\mhrdistspace$};
    \end{tikzpicture}
\caption{\label{fig:distribution-hierarchy}
A Venn diagram of all five families of regular ($\regulardistspace$), MHR ($\mhrdistspace$), quasi-regular ($\quasiregulardistspace$), quasi-MHR ($\quasimhrdistspace$), and general ($\generaldistspace$) distributions, which together form the following hierarchy: $\mhrdistspace \subsetneq (\regulardistspace \cap \quasimhrdistspace) \subsetneq \regulardistspace,\ \quasimhrdistspace \subsetneq (\regulardistspace \cup \quasimhrdistspace) \subsetneq \quasiregulardistspace \subsetneq \generaldistspace$.}
\end{figure}

Next, we examine all four families $\regulardistspace$, $\mhrdistspace$, $\quasiregulardistspace$, and $\quasimhrdistspace$ from the lens of {\em order statistics}.
These notions are the centerpiece of Bayesian mechanism design (and all other Bayesian models).
For example, {\sf First/Second Price Auction(s)} \cite{V61}, along with their generalizations \cite{C71, G73, M81, M00, CSS06, V07, EOS07} and methodologies motivated \cite{CKS16, R15, CH13, ST13, FFGL20, CKKKLLT15, CKST16, JL23poa, JL23pos}, have shaped the whole theory.

% of a distribution family is math-oriented,
% \href{https://en.wikipedia.org/wiki/Closure_(mathematics)#Closure_operator}{[closure]}

For both original families $\regulardistspace$ and $\mhrdistspace$, we prove (\Cref{thm:order:mhr,thm:order:regular}) that every order statistic $\prior_{(k:n)}$, $\forall k \in [n]$, of $n \geq 1$ many {\em i.i.d.}\ regular (resp.\ MHR) distributions $\priors = \{\prior\}^{\otimes n}$ also is regular (resp.\ MHR).
Unfortunately, this structural result fails once the distributions $\priors = \{\prior_{i}\}_{i \in [n]}$ become {\em asymmetric}; see \Cref{footnote:quasi-regular/-mhr:example} for counterexamples.
In contrast, if we consider the new families $\quasiregulardistspace$ and $\quasimhrdistspace$ instead, this structural result {\em does} hold for {\em asymmetric} distributions (\Cref{thm:order:quasi-mhr,thm:order:quasi-regular}).
Hence, by regarding ``i.i.d.\ order statistic'' and ``order statistic'' as two operations, we can examine the {\em closure property} of all four families $\regulardistspace$, $\mhrdistspace$, $\quasiregulardistspace$, and $\quasimhrdistspace$:

\begin{observation}[\Cref{thm:order:mhr,thm:order:regular}]
\label{observation:order:iid}
\begin{flushleft}
Both original families $\regulardistspace$ and $\mhrdistspace$ are {\em closed} under the more restricted ``i.i.d.\ order statistic'' operation but (\Cref{footnote:quasi-regular/-mhr:example}) are {\em not closed} under the more general ``order statistic'' operation.
% Both original families of regular distributions $\regulardistspace$ and MHR distributions $\mhrdistspace$ are {\em closed} under the more restricted ``i.i.d.\ order statistic'' operation. However, both are {\em not closed} under the more general ``order statistic'' operation.
\end{flushleft}
\end{observation}

\begin{observation}[\Cref{thm:order:quasi-mhr,thm:order:quasi-regular}]
\label{observation:order:asymmetric}
\begin{flushleft}
Both new families $\quasiregulardistspace$ and $\quasimhrdistspace$ are {\em closed} under the more general ``order statistic'' operation.
% Both new families of quasi-regular distributions $\quasiregulardistspace$ and quasi-MHR distributions $\quasimhrdistspace$ are {\em closed} under the more general ``order statistic'' operation.
\end{flushleft}
\end{observation}

\noindent
These suggest that both new families $\quasiregulardistspace$ and $\quasimhrdistspace$ might be more natural and elegant, at least mathematically (Criterion~\ref{criterion:math}), than both original families $\regulardistspace$ and $\mhrdistspace$.

In sum, both new families $\quasiregulardistspace$ and $\quasimhrdistspace$ satisfy Criteria~\ref{criterion:econ} and \ref{criterion:math}.
Also, \Cref{observation:order:iid,observation:order:asymmetric} shall be of independent interest; indeed, we have already found many of their applications, which we will discuss soon after (Criterion~\ref{criterion:generalization}).

\xhdr{Comparison with \cite{HR14}.}
In fact, our work is {\em not} the first time in the literature that the quasi-regularity condition has appeared. (The quasi-MHR condition {\em does} appear for the first time. Also, we first name both conditions by ``quasi-regularity/-MHR''.)
The work \cite[Appendix~D]{HR14} introduced the {\em inscribed triangle property} of a distribution;\footnote{As mentioned on Roughgarden's homepage \href{https://timroughgarden.org/chron.html}{[https://timroughgarden.org/chron.html]}, the work \cite{HR14} is a more recent version of \cite{HR08} with a different emphasis and some new results.} it is not difficult to verify (\Cref{prop:q-regular equivalent definition}) that this property coincides with our quasi-regularity condition.

% The previous work \cite[Appendix~D]{HR14},\footnote{As mentioned on Roughgarden's homepage \href{https://timroughgarden.org/chron.html}{[https://timroughgarden.org/chron.html]}, the work \cite{HR14} is a more recent version of \cite{HR08} with a different emphasis and some new results.} based on the {\em revenue curve} (i.e., an equivalent representation of a distribution; see \Cref{sec:prelim} for more details), introduced the {\em inscribed triangle property}.
% It is not difficult to verify that this property is equivalent to our quasi-regularity condition (\Cref{prop:q-regular equivalent definition}).

Disappointingly and surprisingly, the community has omitted the quasi-regularity/-MHR conditions.
The work \cite[Appendix~D]{HR14} itself essentially only gave one specialized Bulow-Klemperer-type result in this regard:
{\SecondPriceAuction} with {\em two} i.i.d.\ quasi-regular buyers $\prior \otimes \prior \in \quasiregulardistspace$ surpasses {\BayesianMyersonAuction} with {\em one} such buyer $\prior \in \quasiregulardistspace$.
(Afterward, no follow-up work has ever investigated this condition.)
% (In the scope of Bayesian mechanism design, \cite[Appendix~D]{HR14} only gave a (specialized) Bulow-Klemperer-type result: {\SecondPriceAuction} extracts more revenue from two i.i.d.\ quasi-regular buyers $\prior \otimes \prior$, than {\BayesianMyersonAuction} from one such buyer $\prior$. Also, no follow-up work has proceeded to study the quasi-regularity condition.)
Nonetheless, as our work will show, the quasi-regular/-MHR conditions have numerous exciting applications (Criterion~\ref{criterion:generalization}).

\subsection*{Approximation by Bayesian Simple Mechanisms \texorpdfstring{(\Cref{sec:simple-mechanism})}{}}

To begin with, we present applications of the quasi-regular/-MHR conditions to the first mentioned topic, Approximation by Bayesian Simple Mechanisms.
For readability, we often consider the {\em single-item} setting first and the broader {\em downward-closed} settings afterward.

\xhdr{The Single-Item Setting.}
The following are among the most canonical single-item mechanisms.
\begin{itemize}
    \item {\BayesianOptimalMechanism} ({\BOM}):
    As the cornerstone of Bayesian mechanism design, Myerson characterized {\BayesianOptimalMechanism} for {\em any} single-parameter setting; thus, that mechanism is often called {\BayesianMyersonAuction}.
    In the single-item setting, the seller first determines buyers' individual {\em ironed virtual value function} $\ironvirtualval_{i}$, $\forall i \in [n]$ based on the Bayesian information $\priors$ and, then, the buyer with the {\em highest nonnegative} ironed virtual value (if any) wins and pays the threshold value for winning.
    (See \Cref{sec:prelim} for more details.)
    
    \item {\BayesianOptimalUniformReserve} ($\BOUR$):
    This is a variant of {\SecondPriceAuction}.
    Given a uniform reserve $\reserve$, if the highest value $\val_{(1:n)}$ is higher than the uniform reserve, i.e., $\val_{(1:n)} \geq \reserve$, then the corresponding buyer wins, paying the maximum of the reserve $\reserve$ and the second highest value $\val_{(2:n)}$; otherwise, the item is unsold.
    The seller would select the reserve $\reserve^{*} = \reserve^{*}(\priors)$ that optimizes the expected revenue based on the Bayesian information $\priors$. (Remarkably, {\BayesianOptimalMechanism} degenerates into {\BayesianOptimalUniformReserve} provided i.i.d.\ regular buyers \cite{M81}.)

    \item {\BayesianOptimalSequentialPricing} ({\BOSP}):
    Given an arrival order $\order \in \Pi_{n}$ and individual prices $\prices = (\price_{i})_{i \in [n]}$, in an online-algorithm style, every arriving buyer $i = \order(1), \order(2), \dots, \order(n)$ makes a take-it-or-leave-it decision at the price $\price_{i}$ if the item remains unsold.
    As before, the seller would Bayesian-optimize the order and the prices $(\optorder, \optprices) = (\order^{*}(\priors), \optprices(\priors))$.
    % \footnote{\yfedit{order-selection vs.\ random-order vs.\ adversarial order}}
    % based on the Bayesian information $\priors$.

    \item {\BayesianOptimalUniformPricing} ({\BOUP}):
    Everything is the same as for {\BOSP}, except that the individual prices must be uniform $\prices = (\price)^{\otimes n}$. (Therefore, different arrival orders $\order \in \Pi_{n}$ may change the winning buyer but not the seller's outcome revenue.)
    Once again, the seller would Bayesian-optimize this price $\optprice = \optprice(\priors)$.
\end{itemize}
All four mechanisms form the hierarchy $\BOUP \leq \BOUR,\ \BOSP \leq \BOM$ regarding revenues, with only one incomparable pair ``$\BOUR$ vs.\ {\BOSP}'' (\Cref{fig:simple-mechanism}).

\begin{figure}[t]
{\centering
\begin{tikzpicture}[scale = 2]
    \node(n1) at (2.25, 0.7) {{\BayesianOptimalMechanism} ({\BOM})};
    \node(n2) at (0, 0) {{\BayesianOptimalUniformReserve} ($\BOUR$)};
    \node(n3) at (4.5, 0) {{\BayesianOptimalSequentialPricing} ({\BOSP})};
    \node(n4) at (2.25, -0.7) {{\BayesianOptimalUniformPricing} ({\BOUP})};
    \draw[thick, <-, >=triangle 45] (n1.west) to node[anchor = -30] {$\star$} (n2.north);
    \draw[thick, <-, >=triangle 45] (n1.east) to (n3.north);
    \draw[thick, <-, >=triangle 45] (n1.south) to node[left] {$\star$} (n4.north);
    \draw[thick, <-, >=triangle 45] (n2.south) to (n4.west);
    \draw[thick, <-, >=triangle 45] (n3.south) to node[anchor = 150] {$\star$} (n4.east);
\end{tikzpicture}
\par}
\caption{\label{fig:simple-mechanism}
A Hasse diagram of the single-item mechanisms (i.e., an arrow ``$\mech_{1} \to \mech_{2}$'' means the latter $\mech_{2}$ surpasses the former $\mech_{1}$), hence the hierarchy $\BOUP \leq \BOUR,\ \BOSP \leq \BOM$, with only one incomparable pair ``$\BOUR$ vs.\ {\BOSP}''.}
\end{figure}
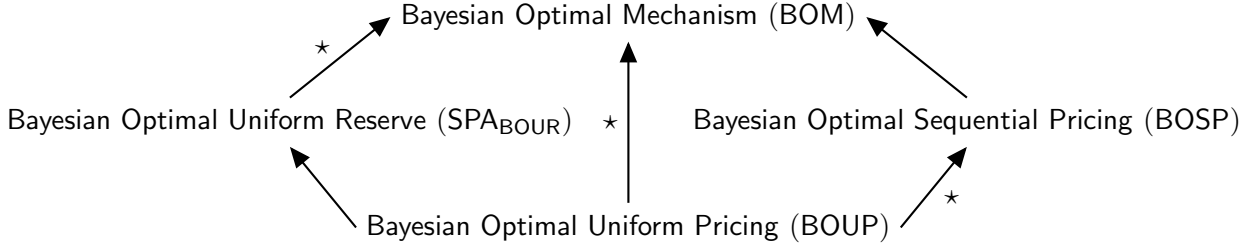

% Every comparable pair (``{\BOUP} vs.\ {\BOSP}'' etc) refers to a meaningful revenue gap ($\calC^{\BOUP}_{\BOSP}$ etc).
% (say ``{\BOUP} vs.\ {\BOM}'')
% (say $\calC_{\BOM}^{\BOUP}$).
Every comparable pair refers to a meaningful revenue gap.
However, two of these revenue gaps,
$\calC_{\BOM}^{\BOSP}$ and $\calC_{\BOUR}^{\BOUP}$,
% {\BayesianOptimalMechanism} vs.\ {\BayesianOptimalSequentialPricing} ($\calC_{\BOM}^{\BOSP}$) and {\BayesianOptimalUniformReserve} vs.\ {\BayesianOptimalUniformPricing} ($\calC_{\BOUR}^{\BOUP}$)
are quite robust to distributional assumptions;\footnote{Both revenue gaps are (absolute) constants even for asymmetric general buyers \cite{PT22, BC23, JLTX20}. Moreover, the former $\calC_{\BOM}^{\BOSP}$ admits an {\em approximation-preserving} reduction from general distributions to regular distributions, whether the buyers are symmetric or not \cite{CFPV19}.
The latter is $\calC_{\BOUR}^{\BOUP} = 1 - 1 / e \approx 0.6321$ for symmetric regular buyers, and is $\calC_{\BOUR}^{\BOUP} = 6 / \pi^{2} \approx 0.6079$ when either ``symmetric'' is relaxed to ``asymmetric'' or ``regular'' is relaxed to ``quasi-regular''/``general'' (or both).}
instead, we will concentrate on the other three revenue gaps $\calC_{\BOSP}^{\BOUP}$, $\calC_{\BOM}^{\BOUP}$, and $\calC_{\BOM}^{\BOUR}$ (marked by the symbol $\star$ in \Cref{fig:simple-mechanism}).
In brief, regarding these three revenue gaps, quasi-regular distributions $\quasiregulardistspace$ will behave more like {\em general} distributions $\generaldistspace$ for {\em symmetric} buyers, namely ``as general as possible'' (Criterion~\ref{criterion:generalization}), but more like {\em regular} distributions $\regulardistspace$ for {\em asymmetric} buyers, namely ``a certain degree of regularity'' that enables generalizations of previous results (Criterion~\ref{criterion:generalization}).

\begin{table}[t]
    \centering
    \begin{tabular}{|l||>{\arraybackslash}p{6.55cm}|>{\arraybackslash}p{6.55cm}|}
        \hline
        \rule{0pt}{13pt} & symmetric buyers & asymmetric buyers \\ [2pt]
        \hline
        \hline
        \rule{0pt}{13pt}$\regulardistspace$ & $\calC_{\BOSP}^{\BOUP} = 1 - \frac{1}{e} \approx 0.6321$ & $\calC_{\BOSP}^{\BOUP} \approx 0.3817$ \\ [2pt]
        \rule{0pt}{13pt} & $\calC_{\BOM}^{\BOUP} = 1 - \frac{1}{e} \approx 0.6321$ & $\calC_{\BOM}^{\BOUP} \approx 0.3817$ \\ [2pt]
        \rule{0pt}{13pt} & $\calC_{\BOM}^{\BOUR} = 1$ & $\calC_{\BOM}^{\BOUR} \in [0.3817, 0.4630]$ \\ [2pt]
        \hline
        \rule{0pt}{13pt}$\quasiregulardistspace$ & \parbox{4.75cm}{$\calC_{\BOSP}^{\BOUP} = \frac{1}{2}$} $(\downarrow)$ & \parbox{4.75cm}{$\calC_{\BOSP}^{\BOUP} \approx 0.3817$} $(\uparrow)$ \\ [2pt]
        \rule{0pt}{13pt} & \parbox{4.75cm}{$\calC_{\BOM}^{\BOUP} = \frac{1}{2}$} $(\downarrow)$ & \parbox{4.75cm}{$\calC_{\BOM}^{\BOUP} \in [0.2770, \underline{0.3817}]$} \underline{$(\uparrow)$} or $(\leftrightarrow)$ \\ [2pt]
        \rule{0pt}{13pt} & \parbox{4.75cm}{$\calC_{\BOM}^{\BOUR} \in [\frac{1}{2}, \underline{0.6822}]$} $(\downarrow)$ or \underline{$(\leftrightarrow)$} & \parbox{4.75cm}{$\calC_{\BOM}^{\BOUR} \in [0.2770, 0.4630]$} \underline{$(\uparrow)$} or $(\leftrightarrow)$ \\ [2pt]
        \hline
        \rule{0pt}{13pt}$\generaldistspace$ & $\calC_{\BOSP}^{\BOUP} = \frac{1}{2}$ & \parbox{4cm}{$\calC_{\BOSP}^{\BOUP} = 0$} (no guarantee) \\ [2pt]
        \rule{0pt}{13pt} & $\calC_{\BOM}^{\BOUP} = \frac{1}{2}$ & \parbox{4cm}{$\calC_{\BOM}^{\BOUP} = 0$} (no guarantee) \\ [2pt]
        \rule{0pt}{13pt} & $\calC_{\BOM}^{\BOUR} = \frac{1}{2}$ & \parbox{4cm}{$\calC_{\BOM}^{\BOUR} = 0$} (no guarantee) \\ [2pt]
        \hline
    \end{tabular}
    \caption{\label{tab:intro:simple-mechanism}
    A summary of three revenue gaps in various single-item settings: \\
    {\BayesianOptimalSequentialPricing} vs.\ {\BayesianOptimalUniformPricing} ($\calC_{\BOSP}^{\BOUP}$); \\
    {\BayesianOptimalMechanism} vs.\ {\BayesianOptimalUniformPricing} ($\calC_{\BOM}^{\BOUP}$); \\
    {\BayesianOptimalMechanism} vs.\ {\BayesianOptimalUniformReserve} ($\calC_{\BOM}^{\BOUR}$). \\
    Symbol $(\uparrow)$: the quasi-regular case $\quasiregulardistspace$ (upward-)collapses into the regular case $\regulardistspace$. \\
    Symbol $(\leftrightarrow)$: the quasi-regular case $\quasiregulardistspace$ separates from both other cases $\regulardistspace$ and $\generaldistspace$. \\
    Symbol $(\downarrow)$: the quasi-regular case $\quasiregulardistspace$ (downward-)collapses into the general case $\generaldistspace$. \\
    Symbol \underline{\quad}: an underlined \underline{bound} is conjectured to be tight; likewise for a \underline{collapse} or a \underline{separation}.}
\end{table}

First, we consider {\em asymmetric} buyers.
Previously, only for regular distributions $\regulardistspace$ the three studied revenue gaps $\calC_{\BOSP}^{\BOUP}$, $\calC_{\BOM}^{\BOUP}$, and $\calC_{\BOM}^{\BOUR}$ are known to be {\em constants} (\Cref{wisdom:BOSP_BOUP,wisdom:BOM_BOUP,wisdom:BOM_BOUR}).\footnote{For general distributions $\generaldistspace$ and $\forall n \geq 1$, they are $\calC_{\BOSP}^{\BOUP} = \calC_{\BOM}^{\BOUP} = \calC_{\BOM}^{\BOUR} = \frac{1}{n} \to 0$ \cite[Section~6]{AHNPY19}.}
Instead, we will show that they all remain constants under the relaxation to quasi-regular distributions $\quasiregulardistspace$ (\Cref{result:BOSP_BOUP,result:BOM_BOUP,result:BOM_BOUR}).
Remarkably, the generalization of $\calC_{\BOSP}^{\BOUP}$ is {\em approximation-preserving} (and we conjecture that the other two are also {\em approximation-preserving}).
More concretely:

\begin{wisdom}[{\cite[Theorem~1]{JLTX20}}]
\label{wisdom:BOSP_BOUP}
\begin{flushleft}
For asymmetric regular buyers, {\BayesianOptimalUniformPricing} achieves a tight $\calC_{\BOSP}^{\BOUP} \approx 0.3817$-approximation to {\BayesianOptimalSequentialPricing}.
% \begin{align*}
%     \calC_{\BOSP}^{\BOUP} = \calC_{\BOM}^{\BOUP} & \eqdef \left(2 + \int_{1}^{+\infty} \left(1 - e^{-\calQ}(x)\right) \cdot \d x\right)^{-1} \approx 0.3817, \\
%     \calQ(x) & \eqdef \ln\left(\frac{x^{2}}{x^{2} - 1}\right) - \frac{1}{2} \sum_{k = 1}^{\infty} \frac{1}{k^{2} \cdot x^{2k}},\ \forall x \geq 1.
% \end{align*}
\end{flushleft}
\end{wisdom}

\begin{result}[\Cref{thm:BOSP_BOUP}]
\label{result:BOSP_BOUP}
\begin{flushleft}
\Cref{wisdom:BOSP_BOUP} admits an {\em approximation-preserving} generalization, from regular buyers to quasi-regular buyers.
% For asymmetric quasi-regular buyer, $\calC_{\BOSP}^{\BOUP} \approx 0.3817$, which indicates a {\em collapse} to the ``asymmetric regular'' case (\Cref{wisdom:BOSP_BOUP}) and a {\em separation} from the ``asymmetric general'' case (i.e., no guarantee).
\end{flushleft}
\end{result}

\begin{wisdom}[{\cite[Theorem~3.1]{AHNPY19} \cite[Theorem~1]{JLTX20} \cite[Theorem~1]{JLQTX19}}]
\label{wisdom:BOM_BOUP}
\begin{flushleft}
For asymmetric regular buyers, {\BayesianOptimalUniformPricing} achieves a tight $\calC_{\BOM}^{\BOUP} \approx 0.3817$- approximation (which is exactly the same as in \Cref{wisdom:BOSP_BOUP}) to {\BayesianOptimalMechanism}.
\end{flushleft}
\end{wisdom}

\begin{result}[\Cref{thm:BOM_BOUP}]
\label{result:BOM_BOUP}
\begin{flushleft}
\Cref{wisdom:BOM_BOUP} admits a generalization, from regular buyers to quasi-regular buyers, except that the counterpart tight bound $\calC_{\BOM}^{\BOUP} \in [0.2770, 0.3817]$ remains open.
% For asymmetric quasi-regular buyer, $\calC_{\BOM}^{\BOUP} \in [0.2770, 0.3817]$, which indicates an {\em approximate collapse} to the ``asymmetric regular'' case (\Cref{wisdom:BOM_BOUP}) and a {\em separation} from the ``asymmetric general'' case (i.e., no guarantee).

\Comment{We conjecture that this generalization actually also is {\em approximation-preserving} $\calC_{\BOM}^{\BOUP} \approx 0.3817$.}

% \Comment{We conjecture that the upper bound is the tight bound $\calC_{\BOM}^{\BOUP} \approx 0.3817$. Suppose so, the {\em approximate collapse} to the ``asymmetric regular'' case (\Cref{wisdom:BOM_BOUP}) will be an {\em exact collapse}.}
\end{flushleft}
\end{result}

\begin{wisdom}[\cite{HR09,AHNPY19,JLTX20,JLQTX19}]
\label{wisdom:BOM_BOUR}
\begin{flushleft}
For asymmetric regular buyers, {\BayesianOptimalUniformReserve} achieves a tight $\calC_{\BOM}^{\BOUR} \in [0.3817, 0.4630]$-approximation (where the lower bound follows from \Cref{wisdom:BOSP_BOUP,wisdom:BOM_BOUP}) to {\BayesianOptimalMechanism}.
% For asymmetric quasi-regular buyers,  {\BayesianOptimalUniformPricing} achieves the same tight approximation to {\BayesianOptimalMechanism} as the asymmetric regular case.
\end{flushleft}
\end{wisdom}

\begin{result}[\Cref{thm:BOM_BOUR}]
\label{result:BOM_BOUR}
\begin{flushleft}
\Cref{wisdom:BOM_BOUR} admits a generalization, from regular buyers to quasi-regular buyers, except that the counterpart tight bound $\calC_{\BOM}^{\BOUP} \in [0.2770, 0.4630]$ remains open.

% For asymmetric quasi-regular buyer, $\calC_{\BOM}^{\BOUR} \in [0.2770, 0.4630]$, which indicates an {\em approximate collapse} to the ``asymmetric regular'' case (\Cref{wisdom:BOM_BOUR}) and a {\em separation} from the ``asymmetric general'' case (i.e., no guarantee).

\Comment{We conjecture that this generalization actually also is {\em approximation-preserving}, although (\Cref{wisdom:BOM_BOUR}) even the tight bound $\calC_{\BOM}^{\BOUP} \in [0.3817, 0.4630]$ for regular buyers remains open.}
\end{flushleft}
\end{result}

\noindent
To conclude, (Criterion~\ref{criterion:generalization}) quasi-regular distributions $\quasiregulardistspace$ involve ``a certain degree of regularity'' that enables generalizations of previous results; see \Cref{tab:intro:simple-mechanism} for a summary.

Next, we consider {\em symmetric} buyers.
For regular distributions $\regulardistspace$, {\BayesianOptimalUniformPricing} achieves a tight $1 - \frac{1}{e} \approx 0.6321$-approximation to every other mechanism \cite{M81,CHMS10,DFK16,har16}.
For general distributions $\generaldistspace$, its tight revenue guarantees become $\frac{6}{\pi^{2}} \approx 0.6079$ against $\BOUR$ \cite{JLTX20, JJLZ22} and $\frac{1}{2}$ against either $\BOSP$ or $\BOM$ \cite{CHMS10, DFK16, har16}; for every given $n \geq 1$, all three bounds share the same worst-case (irregular) instance $\priors = \{\prior\}^{\otimes n}$ for $\cdf(\val) = (1 - \frac{1}{\val})^{\frac{1}{n}}$.
We notice that this distribution $F$ {\em is} quasi-regular $\in \quasiregulardistspace$; (\Cref{prop:BOM_BOUP:iid}) thus regarding the revenue guarantees of {\BayesianOptimalUniformPricing}, the ``i.i.d.\ quasi-regular'' case collapses into the ``i.i.d.\ general'' case, but separates from the ``i.i.d.\ regular'' case.

\iffalse

Next, we consider {\em symmetric} buyers.
For regular distributions $\regulardistspace$, {\BayesianOptimalUniformPricing} achieves a tight $1 - (1 - \frac{1}{n})^{n} \to 1 - \frac{1}{e} \approx 0.6321$-approximation to every other mechanism \cite{M81, CHMS10, DFK16, har16}.
For general distributions $\generaldistspace$, {\BayesianOptimalUniformPricing} achieves a tight $\calC_{\BOUR}^{\BOUP} = \frac{6}{\pi^{2}} + o_{n}(1) \to \frac{6}{\pi^{2}} \approx 0.6079$-approximation to {\BayesianOptimalUniformReserve} \cite{JLTX20, JJLZ22} or a tight $\calC_{\BOSP}^{\BOUP} = \calC_{\BOM}^{\BOUP} = \frac{n}{2n - 1} \to \frac{1}{2}$-approximation to either {\BayesianOptimalSequentialPricing} or {\BayesianOptimalMechanism} \cite{CHMS10, DFK16, har16}; indeed, these three revenue gaps share the same worst-case (irregular) instance $\priors = \{\prior\}^{\otimes n}$ for $\cdf(\val) = (1 - \frac{1}{\val})^{\frac{1}{n}}$.
We observe that this distribution $F$ {\em is} quasi-regular $\in \quasiregulardistspace$ (\Cref{prop:BOM_BOUP:iid}); thus regarding all the three revenue gaps, the ``i.i.d.\ quasi-regular'' case collapses into the ``i.i.d.\ general'' case, but separates from the ``i.i.d.\ regular'' case.

\fi

% \yj{TBC}

The revenue guarantee of {\BayesianOptimalUniformReserve} against {\BOM} has slightly different circumstances.
Both mechanisms are equivalent in the ``i.i.d.\ regular'' case, but are separate in the ``i.i.d.\ general'' case (\Cref{wisdom:BOM_BOUR:iid}).
In contrast, our \Cref{result:BOM_BOUR:iid} suggests that quasi-regular distributions $\quasiregulardistspace$ (are likely to) constitute an {\em intermediate family} between $\regulardistspace$ and $\generaldistspace$.
More concretely:

\iffalse

The revenue gap $\calC_{\BOM}^{\BOUR}$ of {\BayesianOptimalMechanism} and {\BayesianOptimalUniformReserve} has slightly different circumstances.
Myerson's celebrated result shows that both mechanisms are equivalent for regular distributions $\regulardistspace$ but are separate for general distributions $\generaldistspace$ (\Cref{wisdom:BOM_BOUR:iid}).
In contrast, our \Cref{result:BOM_BOUR:iid} suggests that quasi-regular distributions $\quasiregulardistspace$ (are likely to) constitute an {\em intermediate family} between $\regulardistspace$ and $\generaldistspace$.
More concretely:

\fi

% For regular distributions $\regulardistspace$, {\BayesianOptimalUniformPricing} achieves a tight $1 - (1 - \frac{1}{n})^{n} \to 1 - \frac{1}{e} \approx 0.6321$-approximation to either {\BayesianOptimalMechanism} or {\BayesianOptimalSequentialPricing}

% This suggests that our new family $\quasiregulardistspace$ {\em is} ``as general as possible'' or, at least, is much more general than the original family $\regulardistspace$.

\begin{wisdom}[\cite{M81,CHMS10,har16}]
\label{wisdom:BOM_BOUR:iid}
\begin{flushleft}
For i.i.d.\ regular buyers or i.i.d.\ general buyers, {\BayesianOptimalUniformReserve} achieves a tight $\calC_{\BOM}^{\BOUR} = 1$ or $\calC_{\BOM}^{\BOUR} = \frac{1}{2}$-approximation  to {\BayesianOptimalMechanism}, respectively.
\end{flushleft}
\end{wisdom}

\begin{result}[\Cref{prop:BOM_BOUR:iid}]
\label{result:BOM_BOUR:iid}
\begin{flushleft}
For i.i.d.\ quasi-regular buyers, {\BayesianOptimalUniformReserve} achieves a tight $\calC_{\BOM}^{\BOUR} \in [\frac{1}{2}, 0.6822]$-approximation  to {\BayesianOptimalMechanism}.

\Comment{This shows a {\em separation} from the ``i.i.d.\ regular'' case.
%  ($\calC_{\BOM}^{\BOUR} = 1$)
We conjecture that the upper bound is tight $\calC_{\BOM}^{\BOUR} \approx 0.6822$ (or at least bounded away from $\frac{1}{2}$); suppose so, this shows a {\em separation} from the ``i.i.d.\ general'' case.}
% ($\calC_{\BOM}^{\BOUR} = \frac{1}{2}$)
\end{flushleft}
\end{result}

\noindent
To conclude, (Criterion~\ref{criterion:generalization}) quasi-regular distributions $\quasiregulardistspace$ are ``as general as possible'', or at least much more general than the original family $\regulardistspace$; see \Cref{tab:intro:simple-mechanism} for a summary.

% \yfedit{``general enough'' -- we mean, from the lens of worst-case analysis, the bounds for quasi-regular distributions can be as bad as for general distributions.}

\xhdr{The Downward-Closed Settings.}
Next, we move on to the broader downward-closed settings.
Unlike the single-item setting, (the ``downward-closed'' counterparts of) the three simple mechanisms studied before,
$\BOUR$, $\BOSP$, and $\BOUP$,
% {\BayesianOptimalUniformReserve}, {\BayesianOptimalSequentialPricing}, and {\BayesianOptimalUniformPricing},
all cannot achieve constant approximations to {\BayesianOptimalMechanism}, even when the buyers have {\em deterministic values}.\footnote{Super-constant lower-bound instances can be found from \cite[Section~4.1]{BIKK18} (by which $\calC_{\BOM}^{\BOSP} = \Omega(\frac{\log n}{\log \log n})$) and \cite[Example~5.4]{HR09} \cite[Theorem~5]{JJLZ22} (by which $\calC_{\BOM}^{\BOUR} = \Omega(\log n)$ and $\calC_{\BOM}^{\BOUP} = \Omega(\log n)$).}
The previous works \cite{HR09,DRY15} instead examined another simple mechanism called {\BayesianMonopolyReserves}, which relies on the notion of {\em monopoly reserves} $\optreserve_{i}$ (see \Cref{sec:prelim} for their definitions):
\begin{itemize}
    \item\label{mechanism:BMR}
    {\BayesianMonopolyReserves} ($\BMR$):
    This is a variant of {\VCGAuction} \cite{V61, C71, G73} and has two versions, the ``eager'' version and the ``lazy'' version.
    (Remarkably, either version coincides with {\BayesianOptimalMechanism} provided i.i.d.\ regular buyers \cite{M81,DRY15}.)

    \item {\BayesianEagerMonopolyReserves} ($\BEMR$):
    This mechanism takes three steps. (i)~Remove all buyers $i \in [n]$ with $\val_{i} < \optreserve_{i}$. (ii)~Run {\VCGAuction} on the remaining buyers to determine the winners. (iii)~Charge every winner $i$ the larger of $\optreserve_{i}$ and his/her VCG payment in step~(ii).

    \item {\BayesianLazyMonopolyReserves} ($\BLMR$):
    Everything is the same as for $\BEMR$, except for reversing step~(i) and step~(ii).
\end{itemize}
The works \cite{HR09,DRY15} obtained constant revenue guarantees of $\BEMR$ and $\BLMR$ for MHR distributions $\mhrdistspace$ (\Cref{wisdom:BOM_BMR}). (Such constant revenue guarantees are impossible for regular distributions $\regulardistspace$ \cite[Example~3.4]{HR09}.)
Rather, under the relaxation to quasi-MHR distribution $\quasimhrdistspace$, we will show that these revenue guarantees remain constant (\Cref{result:BOM_BMR}).
More concretely:

\begin{table}[t]
    \centering
    \begin{tabular}{|l||l|l|l|}
        \hline
        \rule{0pt}{13pt}$\mhrdistspace$ & $\calC_{\BOM}^{\BEMR} = \frac{1}{2}$ & $\calC_{\VCG}^{\BLMR} = \frac{1}{e} \approx 0.3679$ & $\calC_{\BOM}^{\VCGn} = \frac{1}{3}$ \\ [2pt]
        \hline
        \rule{0pt}{13pt}$\quasimhrdistspace$ & $\calC_{\BOM}^{\BEMR} = \frac{1}{e + 1} \approx 0.2689$ \quad $(\leftrightarrow)$ & $\calC_{\VCG}^{\BLMR} = \frac{1}{e + 1} \approx 0.2689$ \quad $(\leftrightarrow)$ & $\calC_{\BOM}^{\VCGn} = \frac{1}{3}$ \quad $(\uparrow)$ \\ [2pt]
        \hline
        \rule{0pt}{13pt}$\generaldistspace$ & \multicolumn{3}{c|}{$\calC_{\BOM}^{\BEMR} = \calC_{\VCG}^{\BLMR} = \calC_{\BOM}^{\VCGn} = 0$ (no guarantee)} \\ [2pt]
        \hline
    \end{tabular}
    \caption{\label{tab:intro:downward-closed}
    A summary of three revenue gaps in various downward-closed settings: \\
    {\BayesianOptimalMechanism} revenue vs.\ {\BayesianEagerMonopolyReserves} revenue ($\calC_{\BOM}^{\BEMR}$); \\
    {\VCGAuction} welfare vs.\ {\BayesianLazyMonopolyReserves} revenue ($\calC_{\VCG}^{\BLMR}$); \\
    {\BayesianOptimalMechanism} revenue vs.\ {\nDuplicateVCGAuction} revenue ($\calC_{\BOM}^{\VCGn}$). \\
    The symbols $(\uparrow)$, $(\leftrightarrow)$, $(\downarrow)$, and \underline{\quad} have similar meanings as in \Cref{tab:intro:simple-mechanism}.}
\end{table}

\begin{wisdom}[{\cite[Theorem~3.2]{HR09} \cite[Theorem~3.11]{DRY15}}]
\label{wisdom:BOM_BMR}
\begin{flushleft}
In any downward-closed setting with asymmetric MHR buyers, the revenue from {\BayesianEagerMonopolyReserves} achieves a tight $\calC_{\BOM}^{\BEMR} = \frac{1}{2}$-approximation to the {\em optimal revenue} from {\BayesianOptimalMechanism}, while the revenue from {\BayesianLazyMonopolyReserves} achieves a tight $\calC_{\VCG}^{\BLMR} = \frac{1}{e} \approx 0.3679$- approximation to the {\em optimal welfare} from {\VCGAuction} (thus at least a $\calC_{\BOM}^{\BLMR} \geq \frac{1}{e} \approx 0.3679$- approximation to the {\em optimal revenue} from {\BayesianOptimalMechanism}).
\end{flushleft}
\end{wisdom}

\begin{result}[\Cref{thm:BOM_BMR:quasi-mhr}]
\label{result:BOM_BMR}
\begin{flushleft}
In any downward-closed setting with asymmetric quasi-MHR buyers, the revenue from {\BayesianMonopolyReserves} (i.e., either the ``eager'' version or the ``lazy'' version) achieves a tight $\calC_{\BOM}^{\BEMR} = \calC_{\VCG}^{\BLMR} = \frac{1}{e + 1} \approx 0.2689$-approximation to the {\em optimal revenue} from {\BayesianOptimalMechanism} and/or the {\em optimal welfare} from {\VCGAuction}.
\end{flushleft}
\end{result}

\noindent
The crux of \Cref{wisdom:BOM_BMR,result:BOM_BMR} is the specific characteristics possessed by MHR distributions $\mhrdistspace$ (but not by {\em generic} regular distributions $\regulardistspace$). Also, quasi-MHR distributions $\quasimhrdistspace$ inherit many such characteristics, some with quantitative losses and others without (cf.\ \Cref{lem:quasi-mhr:structural results}).

\Cref{tab:intro:downward-closed} summarizes the above discussions (and a few results that we will mention later).

\iffalse

\begin{wisdom}[Folklore]
\begin{flushleft}
For a MHR buyer, the monopoly quantile $\optquant \geq \frac{1}{e} \approx 0.3679$, while the {\em optimal revenue} achieves a tight $\frac{1}{e} \approx 0.3679$-approximation to the {\em optimal welfare}.
% (The exponential distribution $\cdf(\val) = 1 - e^{-\val}$ for $\val \geq 0$ meets both equalities simultaneously.)
\end{flushleft}
\end{wisdom}

\begin{result}[\Cref{lem:quasi-mhr:structural results}]
\begin{flushleft}
For a quasi-MHR buyer, the monopoly quantile $\optquant \geq \frac{1}{e} \approx 0.3679$, while the {\em optimal revenue} achieves a tight $\frac{1}{3}$-approximation to the {\em optimal welfare}.
% (The distribution $\cdf(\val) = 1 - \frac{1}{e \val}$ for $\val \in [\frac{1}{e}, 1]$ and $\cdf(\val) = 1 - e^{-\val}$ for $\val > 1$ meets both equalities simultaneously.)
\end{flushleft}
\end{result}

\fi

\subsection*{Approximation by Prior-Independent Mechanisms \texorpdfstring{(\Cref{sec:duplicate})}{}}

Now we consider the second mentioned topic, Approximation by Prior Independent Mechanisms.
% Again, we study the {\em single-item} setting first and the broader {\em downward-closed} settings afterward.

\begin{table}[t]
    \centering
    \begin{tabular}{|l||>{\arraybackslash}p{6.55cm}|>{\arraybackslash}p{6.55cm}|}
        \hline
        \rule{0pt}{13pt} & symmetric buyers & asymmetric buyers \\ [2pt]
        \hline
        \hline
        \rule{0pt}{13pt}$\regulardistspace$ & \multicolumn{2}{c|}{\multirow{2}{*}{\rule{0pt}{14pt}$\calC_{\BOUR}^{\SPAone} = 1$}} \\ [2pt]
        \cline{1-1}
        \rule{0pt}{13pt}$\quasiregulardistspace$ & \multicolumn{2}{c|}{} \\ [2pt]
        \hline
        \rule{0pt}{13pt}$\generaldistspace$ & \multicolumn{2}{c|}{$\calC_{\BOUR}^{\SPAone} = 0$ (no guarantee)} \\ [2pt]
    %     \hline
    % \end{tabular}
    % \caption{{\oneDuplicateSecondPriceAuction} vs.\ {\BayesianOptimalUniformReserve}}
    % \vspace{.1in}
    % \begin{tabular}{|l||>{\arraybackslash}p{6.5cm}|>{\arraybackslash}p{6.5cm}|}
    %     \hline
    %     \rule{0pt}{13pt} & symmetric buyers & asymmetric buyers \\ [2pt]
        \hline
        \hline
        \rule{0pt}{13pt}$\regulardistspace$ & $\calC_{\BOM}^{\SPAone} = 1$ & $\calC_{\BOM}^{\SPAone} \in [0.3817, 0.6365]$ \\ [2pt]
        \hline
        \rule{0pt}{13pt}$\quasiregulardistspace$ & \parbox{4.75cm}{$\calC_{\BOM}^{\SPAone} \in [0.2770, \underline{0.6822}]$} $(\leftrightarrow)$ & \parbox{4.75cm}{$\calC_{\BOM}^{\SPAone} \in [0.2770, 0.6365]$} \underline{$(\uparrow)$} or $(\leftrightarrow)$ \\ [2pt]
        \hline
        \rule{0pt}{13pt}$\generaldistspace$ & \multicolumn{2}{c|}{$\calC_{\BOM}^{\SPAone} = 0$ (no guarantee)} \\ [2pt]
        \hline
    \end{tabular}
    \caption{\label{tab:intro:duplicate}
    A summary of two revenue gaps in various single-item settings: \\
    {\BayesianOptimalUniformReserve} vs.\ {\oneDuplicateSecondPriceAuction} ($\calC_{\BOUR}^{\SPAone}$); \\
    {\BayesianOptimalMechanism} vs.\ {\oneDuplicateSecondPriceAuction} ($\calC_{\BOM}^{\SPAone}$). \\
    The symbols $(\uparrow)$, $(\leftrightarrow)$, $(\downarrow)$, and \underline{\quad} have similar meanings as in \Cref{tab:intro:simple-mechanism,tab:intro:downward-closed}.}
\end{table}

\xhdr{The Single-Item Setting.}
The celebrated Bulow-Klemperer result \cite{BK96} and its extension by \cite{FLR19} (\Cref{wisdom:BK96,wisdom:FLR19}) studied a variant mechanism called {\oneDuplicateSecondPriceAuction}:
\begin{itemize}
    \item {\oneDuplicateSecondPriceAuction} ($\SPAone$):
    Run {\SecondPriceAuction} on a new $(n + 1)$-buyer scenario $\priors \times \prior_{i'}$; the new buyer $\prior_{i'} = \prior_{i}$ is duplicated from one original buyer $i \in [n]$ that is chosen {\em optimally} for revenue maximization.\footnote{Optimizing the buyer $i \in [n]$ to be duplicated also relies on the prior information $\priors$. To mitigate this assumption, see \cite[Section~1.1]{FLR19} for a detailed discussion.}
    (Moreover, provided with {\em i.i.d.}\ original buyers, there is no difference in choices of the buyer $i \in [n]$ to be duplicated.)
    
    % retains the original buyers and additionally duplicates exactly {\em one} original buyer $i \in [n]$
    % Each bidder and its duplicate have i.i.d.\ value distributions, are interchangeable within the original (downward-closed) set system, and cannot win simultaneously.
\end{itemize}

\begin{wisdom}[{\cite[Theorem~1]{BK96}}]
\label{wisdom:BK96}
\begin{flushleft}
For i.i.d.\ regular buyers, {\oneDuplicateSecondPriceAuction} surpasses {\BayesianOptimalMechanism} (without duplicates), namely $\calC_{\BOM}^{\SPAone} = 1$.

\Comment{Herein, {\BayesianOptimalMechanism} degenerates into {\BayesianOptimalUniformReserve} \cite{M81}.}
\end{flushleft}
\end{wisdom}

\begin{wisdom}[{\cite[Theorem~3.1]{FLR19}}]
\label{wisdom:FLR19}
\begin{flushleft}
For asymmetric regular buyers, {\oneDuplicateSecondPriceAuction} achieves a tight $\calC_{\BOM}^{\SPAone} \in [0.1080, 0.6931]$-approximation to {\BayesianOptimalMechanism}.
\end{flushleft}
\end{wisdom}

Rigorously speaking, there are two viewpoints of the benchmark in the Bulow-Klemperer result (\Cref{wisdom:BK96}), {\BayesianOptimalMechanism} vs.\ {\BayesianOptimalUniformReserve}.
Although they are equivalent for i.i.d.\ regular buyers \cite{M81}, ambiguity arises if we want to distinguish the ``more right'' benchmark.
For a long time, the community has preferred the former benchmark.
However, as our \Cref{result:BOUR_SPAone} suggests, the latter benchmark arguably is ``more right''.

\begin{result}[\Cref{thm:BOUR_SPAone}]
\label{result:BOUR_SPAone}
\begin{flushleft}
For asymmetric quasi-regular buyers, {\oneDuplicateSecondPriceAuction} surpasses {\BayesianOptimalUniformReserve}, namely $\calC_{\BOUR}^{\SPAone} = 1$.
\end{flushleft}
\end{result}

\noindent
So regarding the {\BayesianOptimalUniformReserve} benchmark, the Bulow-Klemperer result enjoys threefold generalizations:
(i)~from symmetric buyers to asymmetric buyers,
(ii)~from regular buyers to quasi-regular buyers, and
(iii)~most importantly, without quantitative losses $\calC_{\BOUR}^{\SPAone} = 1$.
Hence, \Cref{result:BOUR_SPAone} arguably better reveals the essence of the Bulow-Klemperer result.

% Rigorously speaking, there are two interpretations of t

% For i.i.d.\ regular buyers, {\BayesianOptimalMechanism} degenerates into {\BayesianOptimalUniformReserve} \cite{M81}. Thus, 

% Also, the ``more right'' viewpoint of {\BayesianOptimalUniformReserve}

Another benefit of our \Cref{result:BOUR_SPAone} (given \Cref{wisdom:BOM_BOUR,result:BOM_BOUR}) is an immediate improvement of the lower bound in \Cref{wisdom:FLR19}, to either $\gtrsim 0.3817$ for asymmetric regular buyers or $\gtrsim 0.2770$ for asymmetric quasi-regular buyers.
We also improve the upper bound $\leq \ln 2 \approx 0.6931$ to $\lesssim 0.6365$ by modifying the impossibility instance from \cite[Remark~5.7]{FLR19}. More concretely:
% The following \Cref{result:BOM_SPAone} summarizes both improvements.

\begin{result}[\Cref{thm:BOM_SPAone}]
\label{result:BOM_SPAone}
\begin{flushleft}
For asymmetric regular buyers or asymmetric quasi-regular buyers, {\oneDuplicateSecondPriceAuction} achieves a tight $\calC_{\BOM}^{\SPAone} \in [0.3817, 0.6365]$ or $\calC_{\BOM}^{\SPAone} \in [0.2770, 0.6365]$-approximation to {\BayesianOptimalMechanism}, respectively.

\Comment{We conjecture that both tight bounds are the same; suppose so, the ``asymmetric quasi-regular'' case collapses into the ``asymmetric regular'' case.}
% \yfedit{Clearly, an improvement of the lower bound in \Cref{wisdom:BOM_BOUR} or \Cref{result:BOM_BOUR} implies an improvement here.}
\end{flushleft}
\end{result}

\noindent
\Cref{tab:intro:duplicate} concludes the above discussions (and a few unmentioned results); see \Cref{sec:duplicate} for details.
Also, more Bulow-Klemperer-type results can be found from \cite{HR09,DRS12,DRY15,RTY20}.

\xhdr{The Downward-Closed Settings.}
Next, we move on to the broader downward-closed settings.
Now, a single duplicate buyer is insufficient for constant revenue guarantees. In contrast, the work \cite{HR09} studied another variant mechanism called {\nDuplicateVCGAuction}:
\begin{itemize}
    \item {\nDuplicateVCGAuction} ($\VCGn$):
    The new scenario duplicates {\em every} original buyer $i \in [n]$ and runs {\VCGAuction} on all the $2n$ buyers. As before, regarding allocation feasibility, every pair of original/duplicate buyers $i$ and $i'$ are interchangeable but cannot win simultaneously.
\end{itemize}
The work \cite{HR09} established a constant revenue guarantee of $\VCGn$ for MHR distributions $\mhrdistspace$ (\Cref{wisdom:BOM_VCGn}).\footnote{Such constant revenue guarantees are impossible for regular distributions $\regulardistspace$; it is not difficult to construct a counterexample, by combining the ideas behind \cite[Example~3.4 and Example~4.3]{HR09}.}
Instead, we will prove that, under the relaxation to quasi-MHR distribution $\quasimhrdistspace$, this constant revenue guarantee has no quantitative loss (\Cref{result:BOM_VCGn}).
More concretely:

\begin{wisdom}[{\cite[Theorem~4.2]{HR09}}]
\label{wisdom:BOM_VCGn}
% \label{wisdom:BOM_VCGn:regular}
\begin{flushleft}
In any downward-closed setting with asymmetric MHR buyers, {\nDuplicateVCGAuction} achieves a tight $\calC_{\BOM}^{\VCGn} = \frac{1}{3}$-approximation to {\BayesianOptimalMechanism} (without duplicates).
\end{flushleft}
\end{wisdom}

\begin{comment}

\begin{wisdom}[{\cite[Theorems~4.2 \& 4.4]{HR09}}]
\label{wisdom:BOM_VCGn}
% \label{wisdom:BOM_VCGn:regular}
\begin{flushleft}
In a matroid setting with asymmetric regular buyers (resp.\ a downward-closed setting with asymmetric MHR buyers), {\nDuplicateVCGAuction} achieves at least a $\calC_{\BOM}^{\VCGn} \geq \frac{1}{2}$-approximation\footnote{\yfedit{The best known $\frac{3}{4}$}} (resp.\ a tight $\calC_{\BOM}^{\VCGn} = \frac{1}{3}$-approximation) to {\BayesianOptimalMechanism} (without duplicates).
\end{flushleft}
\end{wisdom}

\end{comment}

\begin{result}[\Cref{thm:BOM_VCGn:quasi-mhr}]
\label{result:BOM_VCGn}
% \label{result:BOM_VCGn:quasi-mhr}
\begin{flushleft}
\Cref{wisdom:BOM_VCGn} admits an {\em approximation-preserving} generalization, from MHR buyers to quasi-MHR buyers.
\end{flushleft}
\end{result}

\Cref{tab:intro:downward-closed} summarizes all the above discussions (and our earlier results for the downward-closed settings).
As before, the crux of \Cref{wisdom:BOM_VCGn,result:BOM_VCGn} is the specific characteristics possessed by MHR distributions $\mhrdistspace$ and inherited by quasi-MHR distributions $\quasimhrdistspace$.

\subsection*{Approximation with a Single Sample \texorpdfstring{(\Cref{sec:single-sample})}{}}

Next, let us reexamine the third mentioned topic, Approximation with a Single Sample \cite{DRY15}.
Unlike the first two topics, now we concentrate solely on the {\em single-item} setting.
In particular, we will consider two kinds of {\em sample access}:\footnote{\label{footnote:sample}For comparison, we may regard different kinds of sample access as analogs of different {\em feedback models} in online learning \cite{CN06}, i.e., the ``full sample access'' is an analog of the {\em full feedback} model, while the ``first-order-statistic sample access'' is an analog of the {\em bandit-feedback} model.}
\begin{itemize}
    \item {\bf Full Sample Access:}
    Herein, {\em one sample} is one $n$-dimensional vector $\samples = (\sample_{i})_{i \in [n]} \sim \priors = \{\prior_{i}\}_{i \in [n]}$ drawn from the underlying $n$-dimensional product distribution.

    \item {\bf First-Order-Statistic Sample Access:}
    Herein, {\em one sample} retains the first-order statistic $\sample_{(1:n)} \sim \prior_{(1:n)}$ of the above full sample $\bs$, discarding all other order statistics $(\sample_{(k:n)})_{k \in [2:n]}$.
\end{itemize}
Throughout \Cref{sec:single-sample}, we will concentrate on the second sample access (mainly because it is {\em more robust to information assumptions} than the first sample access) and study ``{\em revenue maximization with a single sample}'' from the lens of a mechanism called {\IdentityPricing}:
\begin{itemize}
    \item\label{mechanism:IP}
    {\IdentityPricing} ({\IP}):
    Everything is the same as for {\BayesianOptimalUniformPricing}, except that the uniform price $\price = \sample_{(1:n)}$ is set as the single first-order-statistic sample $\sample_{(1:n)} \sim \prior_{(1:n)}$.
\end{itemize}
Indeed, the works \cite{DRY15, HMR18} have addressed the {\em single-buyer} case $n = 1$ of our model, showing the next \Cref{wisdom:DRY15-HMR18}.
(In other words, our ``first-order-statistic samples $\sample_{(1:n)}$'' model naturally generalizes the model by \cite{DRY15,HMR18} and their follow-up works.)

\begin{table}[t]
    \centering
    \begin{tabular}{|l||>{\arraybackslash}p{6.55cm}|>{\arraybackslash}p{6.55cm}|}
        \hline
        \rule{0pt}{13pt} & symmetric buyers & asymmetric buyers \\ [2pt]
        \hline
        \hline
        \rule{0pt}{13pt}$\regulardistspace$ & \multicolumn{2}{c|}{\multirow{2}{*}{\rule{0pt}{14pt}$\calC_{\BOUR}^{\IP} = \frac{1}{2}$}} \\ [2pt]
        \cline{1-1}
        \rule{0pt}{13pt}$\quasiregulardistspace$ & \multicolumn{2}{c|}{} \\ [2pt]
        \hline
        \rule{0pt}{13pt}$\generaldistspace$ & \multicolumn{2}{c|}{$\calC_{\BOUR}^{\IP} = 0$ (no guarantee)} \\ [2pt]
    %     \hline
    % \end{tabular}
    % \caption{{\oneDuplicateSecondPriceAuction} vs.\ {\BayesianOptimalUniformReserve}}
    % \vspace{.1in}
    % \begin{tabular}{|l||>{\arraybackslash}p{6.5cm}|>{\arraybackslash}p{6.5cm}|}
    %     \hline
    %     \rule{0pt}{13pt} & symmetric buyers & asymmetric buyers \\ [2pt]
        \hline
        \hline
        \rule{0pt}{13pt}$\regulardistspace$ & $\calC_{\BOM}^{\IP} = \frac{1}{2}$ & $\calC_{\BOM}^{\IP} \in [0.1908, 0.3750]$ \\ [2pt]
        \hline
        \rule{0pt}{13pt}$\quasiregulardistspace$ & \parbox{4.75cm}{$\calC_{\BOM}^{\IP} \in [0.1385, \underline{0.3978}]$} \underline{$(\leftrightarrow)$} & \parbox{4.75cm}{$\calC_{\BOM}^{\IP} \in [0.1385, 0.3750]$} \underline{$(\uparrow)$} or $(\leftrightarrow)$ \\ [2pt]
        \hline
        \rule{0pt}{13pt}$\generaldistspace$ & \multicolumn{2}{c|}{$\calC_{\BOM}^{\IP} = 0$ (no guarantee)} \\ [2pt]
        \hline
    \end{tabular}
    \caption{\label{tab:intro:single-sample}
    A summary of two revenue gaps in various single-item settings: \\
    {\BayesianOptimalUniformReserve} vs.\ {\IdentityPricing} ($\calC_{\BOUR}^{\IP}$); \\
    {\BayesianOptimalMechanism} vs.\ {\IdentityPricing} ($\calC_{\BOM}^{\IP}$). \\
    The symbols $(\uparrow)$, $(\leftrightarrow)$, $(\downarrow)$, and \underline{\quad} have similar meanings as in \Cref{tab:intro:simple-mechanism,tab:intro:downward-closed,tab:intro:duplicate}.}
\end{table}

\begin{wisdom}[{\cite[Lemma~3.6]{DRY15} \cite[Theorem~6.1]{HMR18}}]
\label{wisdom:DRY15-HMR18}
\begin{flushleft}
Given a single sample $\sample_{(1:n)} \equiv \sample$ from a single regular buyer,
% Given a single sample $\sample \sim \prior$ from a single regular buyer $\prior \in \regulardistspace$,
{\IdentityPricing} achieves a tight $\calC_{\BOM}^{\IP} = \frac{1}{2}$-approximation to {\BayesianOptimalMechanism}; this is the best possible among deterministic mechanisms.

\Comment{For a single (general) buyer, {\BayesianOptimalUniformPricing}, {\BayesianOptimalUniformReserve}, {\BayesianOptimalSequentialPricing}, and {\BayesianOptimalMechanism} all are equivalent.}
\end{flushleft}
\end{wisdom}

\noindent
Akin to the Bulow-Klemperer result (\Cref{wisdom:BK96}), ambiguity arises if we would distinguish among the four Bayesian mechanisms -- $\BOUP \leq \BOUR,\ \BOSP \leq \BOM$ -- the ``more right'' benchmark for \Cref{wisdom:DRY15-HMR18}.

Our \Cref{result:single-sample} suggests that $\BOUP$ and $\BOUR$ are the ``more right'' benchmarks; in between, $\BOUR$ is ``even more right'' since it surpasses $\BOUP$.
In such manner, \Cref{wisdom:DRY15-HMR18} enjoys threefold generalizations:
(i)~from symmetric to asymmetric buyers,
(ii)~from regular to quasi-regular buyers, and
(iii)~most importantly, without quantitative losses $\calC_{\BOUR}^{\IP} = \calC_{\BOUP}^{\IP} = \frac{1}{2}$.
Thus, arguably, our \Cref{result:single-sample} better reveals the essence of \Cref{wisdom:DRY15-HMR18}.

\begin{result}[\Cref{thm:BOUR_IP,thm:SPA_IP,thm:BOM_IP:iid:regular,thm:BOM_IP:asymmetric}]
\label{result:single-sample}
\begin{flushleft}
Given a single first-order-statistic sample $\sample_{(1:n)}$ from asymmetric quasi-regular buyers, {\IdentityPricing} achieves a tight $\calC_{\BOUR}^{\IP} = \calC_{\BOUP}^{\IP} = \frac{1}{2}$-approximation to either {\BayesianOptimalUniformReserve} or {\BayesianOptimalUniformPricing}; this is the best possible among deterministic mechanisms.

\Comment{{\IdentityPricing} (given \cite{M81}, \Cref{wisdom:BOM_BOUR}, and \Cref{result:BOM_BOUR}) also achieves constant approximations -- although may not be as good as $\frac{1}{2}$ -- to {\BayesianOptimalMechanism}; see \Cref{tab:intro:single-sample} for a summary.}
\end{flushleft}
\end{result}

\noindent
Our \Cref{result:single-sample} also shows that the first-order-statistic samples $\sample_{(1:n)}$ are the {\em primary information} to enable constant revenue guarantees, while the other order statistics $(\sample_{(k:n)})_{k \in [2:n]}$ are the {\em secondary information} to improve these constants.\footnote{In contrast, based on $(\sample_{(k:n)})_{k \in [2:n]}$ only (but not $\sample_{(1:n)}$), constant revenue guarantees clearly are impossible.}
To our knowledge, this fact was unknown before.

\Cref{tab:intro:single-sample} summarizes above discussions (and a few unmentioned results); see \Cref{sec:single-sample} for details.
Also, for results of similar flavor, cf.\ \cite{FILS15, BGMM18, RTY20, DZ20, RTY20, ABB22, FHL21}.

\subsection*{The Sample Complexity of Revenue Maximization \texorpdfstring{(\Cref{sec:sample-complexity})}{}}

% one sample refers to.\textsuperscript{\ref{footnote:sample}}

Below, we scrutinize the fourth mentioned topic, the Sample Complexity of Revenue Maximization, in the {\em single-item} setting.
Here, we address the more difficult task of {\em PAC-learnability of Bayesian mechanisms} (instead of constant revenue guarantees).
To this end, the first-order-statistic sample access $\sample_{(1:n)} \sim \prior_{(1:n)}$ is insufficient and, thus, must be replaced with the {\em full sample access} $\samples \sim \priors$. The following two mechanisms have received considerable attention in the literature:
\begin{itemize} 
    \item {\EmpiricalMyersonAuction} ({\EMA}):
    Learn a $(1 - \eps)$-approximation to {\BayesianMyersonAuction} (that fully leverages the Bayesian information) using a finite number of full samples $\samples \sim \priors$.
    
    \item {\EmpiricalUniformReserve} ({\EUR}):
    Everything is the same as for {\EMA}, except for replacing the target function with {\BayesianOptimalUniformReserve}.
\end{itemize}
The works \cite{GHZ19,JLX23} (essentially) settled the sample complexity of {\EmpiricalMyersonAuction} and {\EmpiricalUniformReserve}, respectively, for regular distributions $\regulardistspace$ and/or MHR distributions $\mhrdistspace$ (\Cref{wisdom:sample-complexity}).
Our contributions on this front are to extend their learnability results to our new families $\quasiregulardistspace$ and $\quasimhrdistspace$ (see \Cref{tab:intro:sample-complexity} for a summary).
More concretely:

\begin{table}[t]
    \centering
    \begin{tabular}{|l||>{\arraybackslash}p{4.5cm}|>{\arraybackslash}p{4.5cm}|}
        \hline
        \rule{0pt}{13pt} & {\EmpiricalMyersonAuction} & {\EmpiricalUniformReserve} \\ [2pt]
        \hline
        \hline
        \rule{0pt}{13pt}$\mhrdistspace$ & $\Tilde{\Omega}(n \eps^{-3 / 2})$ and $\Tilde{O}(n \eps^{-2})$ & $\Tilde{\Omega}(\eps^{-3 / 2})$ and $\Tilde{O}(\eps^{-2})$ \\ [2pt]
        \hline
        \rule{0pt}{13pt}$\quasimhrdistspace$ & $\Tilde{\Omega}(n \eps^{-3 / 2})$ and $\Tilde{O}(n \eps^{-2})$ & $\Tilde{\Omega}(\eps^{-3 / 2})$ and $\Tilde{O}(\eps^{-2})$ \\ [2pt]
        \hline
        \rule{0pt}{13pt}$\regulardistspace$ & $\Tilde{\Theta}(n \eps^{-3})$ & $\Tilde{\Theta}(\eps^{-3})$ \\ [2pt]
        \hline
        \rule{0pt}{13pt}$\quasiregulardistspace$ & {$\Tilde{\Omega}(n \eps^{-3})$} & $\Tilde{\Theta}(\eps^{-3})$ \\ [2pt]
        \hline
    \end{tabular}
    \caption{\label{tab:intro:sample-complexity}
    A summary of the sample complexity of {\EmpiricalMyersonAuction} ({\EMA}) and {\EmpiricalUniformReserve} ({\EUR}) in various single-item settings. Several tight bounds remain open.\ignore{The tight bounds for MHR and quasi-MHR distributions remain open.}}
\end{table}

\begin{wisdom}[{\cite[Theorem~1]{GHZ19} \cite[Theorem~1]{JLX23}}]
\label{wisdom:sample-complexity}
\begin{flushleft}
For asymmetric MHR buyers or asymmetric regular buyers, a $(1 - \eps)$-approximate {\EmpiricalMyersonAuction} can be learned using $\Tilde{O}(n \eps^{-2})$ or $\Tilde{O}(n \eps^{-3})$ many (full) samples, respectively, while a $(1 - \eps)$-approximate {\EmpiricalUniformReserve} can be learned using $\Tilde{O}(\eps^{-2})$ or $\Tilde{O}(\eps^{-3})$ many samples, respectively.

\Comment{Matching or best known lower bounds can be found from \cite{HMR18,GHZ19,JLX23}.}
\end{flushleft}
\end{wisdom}

\begin{result}[\Cref{thm:sample-complexity:EUR,thm:sample-complexity:EMA:q-MHR}]
\label{result:sample-complexity}
\begin{flushleft}
\Cref{wisdom:sample-complexity} can be generalized, {\em without quantitative losses},
% in the sample complexity,
either from MHR buyers to quasi-MHR buyers for both {\EmpiricalMyersonAuction} and {\EmpiricalUniformReserve},
or from regular buyers to quasi-regular buyers for {\EmpiricalUniformReserve}.

\Comment{We conjecture that {\EmpiricalMyersonAuction} also admits a sample-complexity-preserving generalization from regular buyers to quasi-regular buyers. Specifically, we provide evidence by showing that the analysis framework developed in \cite{GHZ19} for regular buyers can be extended to instances with a single quasi-regular buyer. See \Cref{lem:truncation:quasi-regular} and \Cref{thm:sample-complexity:EMA:q-regular} for details.}
\end{flushleft}
\end{result}

The crux of \Cref{wisdom:sample-complexity} includes ``revenue loss bounds of certain truncation schemes'' and ``tail bounds of order statistics'' for both original families $\regulardistspace$ and $\mhrdistspace$ \cite{CD15, DHP16, GHZ19, JLX23}.
Likewise, \Cref{result:sample-complexity} stems from extensions of such technical ingredients to our new families $\quasiregulardistspace$ and $\quasimhrdistspace$ (see \Cref{sec:sample-complexity:BOM} for details); which shall be of independent interest.

\iffalse

\begin{wisdom}[Informal; \cite{CD15,AHNPY19,JLX23}]
\label{wisdom:tail}
% \label{wisdom:BOM_VCGn:regular}
\begin{flushleft}
After normalization, the tail of a regular distribution $\prior \in \regulardistspace$ is no heavier than the tail of the {\em equal revenue distribution}, while the tail of an MHR distribution $\prior \in \mhrdistspace$ is no heavier than the tail of the {\em exponential distribution}.
\end{flushleft}
\end{wisdom}

\begin{result}[Informal; \Cref{lem:tail:quasi-regular,lem:tail:quasi-mhr}]
\label{result:tail}
\begin{flushleft}
\Cref{wisdom:tail} still holds, either when the regularity condition is relaxed to the quasi-regular condition or when the MHR condition is relaxed to the quasi-MHR condition.
\end{flushleft}
\end{result}

\fi

\subsection{Further Discussions}
\label{sec:intro:discussion}

In this work, we (re-)introduce the families of quasi-regular $\quasiregulardistspace$ and quasi-MHR $\quasimhrdistspace$ distributions, as natural relaxations/surrogates of the families of regular $\regulardistspace$ and MHR $\mhrdistspace$ distributions, and accordingly generalize previous results from \cite{BK96, HR09, CD15, DRY15, HR14, AHNPY19, JLTX20, JLQTX19, FLR19, GHZ19, JLX23}, with or even without quantitative losses.\footnote{In spirit, the relations of distribution families ``$\regulardistspace$ vs.\ $\quasiregulardistspace$ vs.\ $\generaldistspace$'' and ``$\mhrdistspace$ vs.\ $\quasimhrdistspace$ vs.\ $\generaldistspace$'' somehow are similar to the relations of (say) the complexity classes ``$\mathbf{P}$ vs.\ $\mathbf{NP}$-intermediate vs.\ $\mathbf{NP}$-complete''.} Further, several tight bounds in \Cref{tab:intro:simple-mechanism,tab:intro:downward-closed,tab:intro:duplicate,tab:intro:single-sample,tab:intro:sample-complexity} remain unknown, and it is interesting to close each of them.

% (try to) lay the foundation for this research direction.
Due to the space limit, we concentrate on {\em single-parameter revenue maximization} in this work.
However, our conceptual/technical contributions shall also benefit {\em multi-parameter revenue maximization} and other Bayesian models.
Indeed, we have already found several applications.

% new distribution families $\quasiregulardistspace$ and $\quasimhrdistspace$ and the related technical ingredients shall find (many) applications

\newcommand{\BayesianOptimalDeterministicMechanism}{\textsf{Bayesian Optimal Deterministic Mechanism}}
\newcommand{\BayesianOptimalRandomizedMechanism}{\textsf{Bayesian Optimal Randomized Mechanism}}
\newcommand{\BODM}{{\sf BODM}}
\newcommand{\BORM}{{\sf BORM}}

\xhdr{Unit-Demand Single-Buyer Revenue Maximization.}
This is one of the most basic models in {\em multi-parameter revenue maximization} \cite{CHK07,CHMS10,CMS15,CD15,CDPSY18,CDOPSY22,CDW21,JLQTX19,KSMSW19,JL24}, where a seller aims to sell one of $n \geq 2$ items to a {\em unit-demand} buyer.
Implications from our results for this model include the following.

First, for quasi-regular items, {\BayesianOptimalUniformPricing} achieves a tight $\calC_{\BODM}^{\BOUP} \approx 0.3817$-approximation or $\calC_{\BORM}^{\BOUP} \in [0.1385, 0.3817]$-approximation to {\BayesianOptimalDeterministicMechanism} ({\BODM}) or {\BayesianOptimalRandomizedMechanism} ({\BORM}), respectively.\footnote{$\calC_{\BODM}^{\BOUP} \approx 0.3817$ follows from a combination of \Cref{result:BOSP_BOUP}, that $\BOSP \geq \BODM$ (which is easy to verify through a coupling argument, \`{a} la \cite[Theorem~4]{CHMS10}), \cite[Lemma~1 and Lemma~6]{JLTX20}, and \cite[Main Lemma~3]{JLQTX19}. And $\calC_{\BORM}^{\BOUP} \in [0.1385, 0.3817]$ follows from a combination of \Cref{result:BOM_BOUP}, \cite[Lemma~3 and Lemma~4]{CMS15}, and that $\BORM \geq \BODM$ (by definition).}
Previously such constant revenue guarantees were known only for regular items \cite{CD15, AHNPY19, JLQTX19}.

Second, {\BayesianOptimalDeterministicMechanism} admits a $\mathbf{PTAS}$ for quasi-MHR distributions and a $\mathbf{QPTAS}$ for quasi-regular distributions.
Previously such algorithmic results were known only for MHR and regular distributions, respectively \cite{CD15}.\footnote{Roughly speaking, the crux of \cite{CD15}'s algorithms and proofs precisely is ``revenue loss bounds of certain truncation schemes'' and ``tail bounds of order statistics'' for regular/MHR distributions. Thus, the desired $\mathbf{PTAS}$ and $\mathbf{QPTAS}$ follow by replacing those technical ingredients with our counterparts for quasi-regular/-MHR distributions. For more details, the interested reader can refer to \cite[Sections~8 and 9 and Appendices~F, G, and H]{CD15}.}

% {\BayesianOptimalUniformPricing} achieves a tight $\calC_{\BODM}^{\BOUP} \approx 0.3817$-approximation to {\BayesianOptimalDeterministicMechanism} ({\BODM})\footnote{The lower bound follows from \Cref{result:BOSP_BOUP} and \cite[Proposition~1.3]{JL24} (which summarizes \cite[Appendix~B]{CHMS10}), while the upper bound follows from \cite[Lemma~1 and Lemma~6]{JLTX20} and \cite[Main Lemma~3]{JLQTX19}.} and a tight $\calC_{\BORM}^{\BOUP} \in [0.1385, 0.3817]$-approximation to {\BayesianOptimalRandomizedMechanism} ({\BORM}).\footnote{The lower bound follows from \Cref{result:BOM_BOUP} and \cite[Lemma~3 and Lemma~4]{CMS15}, while the upper bound follows as $\calC_{\BORM}^{\BOUP} \leq \calC_{\BODM}^{\BOUP}$.}

\xhdr{Revenue Maximization for Buyers with Private Budgets.} Budgeted utility is one of the most classic non-quasilinear utility models \cite{CG00,PV14,DM17,FHL19}. Compared with quasilinear buyers, the revenue-optimal mechanism for private budgeted buyers is much more complicated, even for single-buyer instances \cite{CG00,DM17}. Prior work \cite[Theorem~5.1]{FHL19} shows that posting a take-it-or-leave-it price is a $0.5$-approximation for single-buyer instances when the buyer has regular valuation distribution. The crux of their analysis is the tail bound of regular distributions. By showing that quasi-regular distributions admit the same tail bound (\Cref{prop:q-regular equivalent definition}), we can obtain an \emph{approximation-preserving} generalization of \cite[Theorem~5.1]{FHL19}, from regular valuation distribution to quasi-regular valuation distribution.

\xhdr{Cost Prophet Inequalities.}
{\em Prophet inequalities} are a central model in {\em decision making under uncertainty}.
The recent work \cite[Theorem~1.2]{LM24} investigated a {\em cost minimization} variant, and one of their main results is that the optimal online algorithm is $2$-competitive against the offline optimum for i.i.d.\ MHR distributions. (An asymptotically matching lower bound holds for the exponential distributions.)
Instead, we will establish in \Cref{sec:prophet} that this result admits an {\em approximation-preserving} generalization, from MHR distributions to quasi-MHR distributions.

\iffalse

\xhdr{Multi-Parameter Revenue Maximization.}
The equally canonical {\em multi-parameter revenue maximization} context, as before, includes the following four research topics:
\begin{itemize}
    \item {\bf Approximation by Bayesian Simple Mechanisms:}
    \cite{CHK07, CHMS10, CMS15, CD15, KW19, HN17, LY13, BILW20, HH15, Y15, RW18, CDW21, CM16, CZ17, Y18, JLQTX19, KSMSW19, DKL20, MS21, DFLSV22, COZ22, CC23, CCW23, JL24}

    \item {\bf Approximation by Prior-Independent Mechanisms:}
    \cite{RTY20, EFFTW17a, FFR18, BW19, CS21, DRWX24}

    \item {\bf Approximation with a Single Sample (or a Few Samples):} \cite{GK16, AKW19}

    \item {\bf The Sample Complexity of Revenue Maximization:}
    \cite{BSV16, MR16, CD17, BSV23, GW21, GHTZ21}
\end{itemize}
The regularity/MHR conditions also play a pivotal role in this context, so we are optimistic that our quasi-regularity/-MHR conditions will also find applications therein.
For example,

\xhdr{Beyond Bayesian Mechanism Design.}
For example, \cite{LM24}

\fi

%% file: source/prelim.tex
\section{Preliminaries}
\label{sec:prelim}

In this work, we study the \emph{single-parameter revenue maximization problem with multiple agents}. This section introduces the base model and key concepts upon which all of our results are based. Additional model components specific to individual results are described separately in the relevant sections.

\xhdr{Notations.} Denote by $\reals_+$ (resp.\ $\reals_-$) the set of all non-negative (resp.\ non-positive) real numbers. For any pair of $b \geq a \geq 0$, define the sets $[a] \triangleq \{1, 2, \dots, a\}$ and $[a:b] \triangleq\{a, a + 1, \dots, b\}$. Denote by $\indicator{\cdot}$ the indicator function. Denote by $\cdf$ and $\pdf$ the cumulative distribution function (CDF) and density distribution function (PDF) for a given distribution $\prior$. For any pair of $n \geq k \geq 1$, denote by $\prior_{(k:n)}$ the $k$-th order statistic distribution from $n$ distributions $\priors = \{\prior_i\}_{i\in[n]}$.\footnote{Following the literature, to simplify the notation, we refer to the $k$-th \emph{largest} (instead of \emph{smallest}) value as the $k$-th order statistic.}

% \subsection{Bayesian mechanism design}
% \label{sec:prelim:distribution}
\xhdr{Bayesian Mechanism Design.}
There is a seller who wants to sell a single item to a set of $n$ buyers. Given allocation $\alloc_i\in[0, 1]$ and payment $\price_i\in\reals_+$, the utility of buyer $i$ with private value $\val_i$ is $\val_i\alloc_i - \price_i$. The private value $\vali$ is drawn independently from buyer $i$'s valuation distribution (aka., prior) $\priori$. In most cases, we assume that each buyer $i$'s valuation distribution $\priori$ has its support $\supp(\priori) = [0, \hval_i]$ for some $\hval_i\in\reals_+$ or $\supp(\priori) = [0, \infty)$.\footnote{This assumption is without loss of generality, since for any distribution $\prior$ with general support $\supp(\prior)$, we can construct another distribution $\auxprior$ that satisfies our assumption and only differs with $\prior$ with zero measure.} The goal of the seller is to maximize the expected revenue, which is the summation of the expected payment from all buyers.

% Before presenting the classic mechanisms for the seller's revenue maximization task, we first introduce some important and useful concepts for studying this problem.

We consider both the \emph{single-item setting} and \emph{downward-closed settings} in this paper. In the single-item setting, an allocation profile $\allocs = \{\alloc_i\}_{i\in[n]}$ is feasible if and only if $\sum_{i\in[n]}\alloc_i \leq 1$. In the downward-closed settings, an allocation profile $\allocs$ is feasible if and only if $\allocs$ can be induced by a convex combination over a pre-specified downward-closed set system $\feasiblealloc\subseteq 2^{[n]}$.

\xhdr{Revenue Curve, Virtual Value, and Hazard Rate.} This part revisits the important concepts for the revenue analysis in the Bayesian mechanism design.

\begin{definition}[Revenue Curve and Ironed Revenue Curve \cite{BR89}]
    For a buyer with valuation distribution $\prior$, her \emph{revenue curve $\revcurve:[0, 1]\rightarrow \reals_+$} is defined as 
     $\revcurve(\quant) \triangleq \quant \cdot \cdf^{-1}(\quant)$ for each $\quant \in[0,1]$,\footnote{In particular, $\cdf^{-1}(\quant) \triangleq \sup \{\val:\cdf(\val) \leq \quant\}$ for every $\quant \in (0, 1]$, and $\revcurve(0) \triangleq \lim_{\quant \rightarrow 0^+} \revcurve(\quant)$.} and her \emph{ironed revenue curve $\ironrevcurve:[0, 1]\rightarrow \reals_+$} is the concave envelope of her revenue curve $\revcurve$.
\end{definition}
Based on the definition, both the revenue curve and ironed revenue curve are mappings from a given quantile $\quant\in[0, 1]$ to expected revenue. Specifically, the revenue curve (resp.\ ironed revenue curve) captures the maximum expected revenue of posting a deterministic (resp.\ randomized) price to sell a single item to a buyer with valuation distribution $\prior$ \cite{M81,BR89}, given ex ante allocation constraint~$\quant$.\footnote{An ex ante constraint $\quant\in[0, 1]$ requires the expected allocation to be equal to $\quant$, where the randomness considers both buyers' values and the mechanism.}

Fixing a single buyer with revenue curve $\revcurve$, we refer to $\optquant \triangleq \argmax_{\quant\in[0,1]}\revcurve(\quant)$ as the \emph{monopoly quantile}, $\optreserve \triangleq \prior^{-1}(\optquant)$ as the \emph{monopoly reserve}, and $\optrev \triangleq \revcurve(\optquant)$ as the \emph{monopoly revenue}.

\begin{definition}[]
    \label{def:virtual value hazard rate}
    For a buyer with prior $\prior$, her \emph{virtual value function $\virtualval:\supp(\prior)\rightarrow\reals$} and \emph{hazard rate function $\hazardrate:\supp(\prior)\rightarrow\reals_+$} are defined as 
    \begin{align*}
        \virtualval(\val) \triangleq 
        \val - \frac{1 - \cdf(\val)}{\pdf(\val)}
        \;\;
        \mbox{and}
        \;\;
        \hazardrate(\val) \triangleq 
        \frac{\pdf(\val)}{1 - \cdf(\val)}
    \end{align*}
    Moreover, her \emph{conditional expected virtual value function} $\cumvirtualval(\val):\supp(\prior)\rightarrow\reals_-$ and \emph{conditional expected hazard rate function} $\cehazardrate(\val):\supp(\prior)\rightarrow\reals_+$ are defined as
    \begin{align*}
        \cumvirtualval(\val) \triangleq 
        \expect[\val\primed\sim\prior]{\virtualval(\val\primed)\given \val\primed \leq \val} = -\val \cdot \frac{1-\cdf(\val)}{\cdf(\val)}
        \;\;
        \mbox{and}
        \;\;
        \cehazardrate(\val)\triangleq \int_{0}^{\val}\hazardrate(\val\primed)\cdot\d\val\primed = \frac{-\ln(1-\cdf(\val))}{\val}
    \end{align*}
\end{definition}

The relationship between the revenue curve and virtual value is well-established \cite{M81, BR89}: for any value $\val \in \supp(\prior)$, the corresponding virtual value $\virtualval(\val)$ is equal to the slope of the revenue curve at the associated quantile $\quant = 1 - \cdf(\val)$, i.e., $\virtualval(\val) = \revcurve'(\quant)$. Similarly, we define the \emph{ironed virtual value} $\ironvirtualval(\val) \triangleq \ironrevcurve'(\quant)$ for every value $\val\in\supp(\prior)$ and associated quantile $\quant = 1 - \cdf(\val)$. 

\begin{proposition}[{Revenue Equivalence \cite{M81}}]
\label{prop:revenue_equivalence}
\begin{flushleft}
Given any mechanism $\calM$, its expected revenue $\calM(\priors)$ for $n$ buyers with distributions $\priors = \{\prior_i\}_{i\in[n]}$ satisfies the following:
\begin{align*}
\calM(\priors)
= 
\sum_{i\in[n]}
\expect[\val_i\sim\prior_i]
{\virtualval_i(\val_i)\cdot\alloc_i(\val_i)}
\leq 
\sum_{i\in[n]}
\expect[\val_i\sim\prior_i]
{\ironvirtualval_i(\val_i)\cdot \alloc_i(\val_i)}
\end{align*}
where $\alloc_i(\val_i)$ is the interim allocation of buyer $i$ with value $\val_i$.\footnote{Namely, $\alloc_i(\val_i)$ is the expected allocation of buyer $i$ when his value is $\val_i$, where the randomness considers all other buyers and the mechanism.}
In particular, the equality in the second step holds if and only if ``each buyer $i \in [n]$ gets the same interim allocation $\alloc_i(\val_i\primed) = \alloc_i(\val_i\doubleprimed)$, for any two values $\val_{i}\primed$ and $\val_i\doubleprimed$ with the same ironed virtual value $\ironvirtualval_i(\val_i\primed) = \ironvirtualval_i(\val_i\doubleprimed)$.''
\end{flushleft}
\end{proposition}

Given the revenue equivalence (\Cref{prop:revenue_equivalence}), the Bayesian optimal revenue can be achieved by the seller can be characterized as {\BayesianOptimalMechanism} defined as follows. 
Other mechanisms specific to individual results are described separately in the relevant sections.

\begin{definition}[{\BayesianOptimalMechanism}]
    In the \emph{\BayesianOptimalMechanism} (aka., \emph{\BayesianMyersonAuction}), let $S(\vals)$ be the feasible subset of buyers with the maximum non-negative ironed virtual surplus, i.e., $S(\vals) = \argmax_{S\in\feasiblealloc}\sum_{i\in S}\ironvirtualval(\val_i)$, then each buyer $i\in S(\vals)$
    % buyer $i\in[n]$ with the highest nonnegative ironed virtual value $\ironvirtualval(\val_i)$ 
    gets an item and pays the threshold value for keep getting allocated. 
    % If no buyer has positive ironed virtual value, the item then is not allocated. If there are multiple buyers with the highest nonnegative ironed virtual value, an arbitrary but fixed tie-breaking rule can be used.
\end{definition}

\xhdr{Distribution Families.}
As mentioned in the introduction, we study the following four important distribution families:

\begin{itemize}
    \item \textbf{{(Regularity)}} A distribution $\prior$ satisfies the \emph{regularity condition} if its virtual value function $\virtualval(\val)$ is weakly increasing in $\val\in\supp(\prior)$. 
    \item \textbf{{(MHR)}} A distribution $\prior$ satisfies the \emph{monotone hazard rate (MHR) condition} if its hazard rate function $\hazardrate(\val)$ is weakly decreasing in $\val\in\supp(\prior)$. 
    \item \textbf{{(Quasi-Regularity)}} A distribution $\prior$ satisfies the \emph{quasi-regularity condition} if its conditional expected virtual value $\cumvirtualval(\val)$ is weakly increasing in $\val\in\supp(\prior)$.
    \item \textbf{{(Quasi-MHR)}} A distribution $\prior$ satisfies the \emph{quasi-MHR condition} if its conditional expected hazard rate $\cehazardrate(\val)$ is weakly decreasing in $\val\in\supp(\prior)$.
\end{itemize}

\begin{comment}
\begin{definition}[Regularity]
\label{def:regualr}

\end{definition}

\begin{definition}[MHR]
\label{def:MHR}
A distribution $\prior$ satisfies the \emph{monotone hazard rate (MHR) condition} if its hazard rate function $\hazardrate(\val)$ is weakly decreasing in $\val\in\supp(\prior)$. 
\end{definition}

\begin{definition}[Quasi-Regularity]
\label{def:quasi-regular}
A distribution $\prior$ satisfies the \emph{quasi-regularity condition} if its conditional expected virtual value $\cumvirtualval(\val)$ is weakly increasing in $\val\in\supp(\prior)$.
\end{definition}

\begin{definition}[Quasi-MHR]
\label{def:quasi-MHR}
A distribution $\prior$ satisfies the \emph{quasi-MHR condition} if its conditional expected hazard rate $\cehazardrate(\val)$ is weakly decreasing in $\val\in\supp(\prior)$.
\end{definition}
\end{comment}

We refer to a distribution (resp.\ a buyer with valuation distribution) satisfying each distributional condition as a {regular, MHR, quasi-regular, or quasi-MHR} distribution (resp.\ buyer), respectively. We also denote by the class of regular distributions, MHR distributions, quasi-regular distributions as $\regulardistspace$, $\mhrdistspace$, $\quasiregulardistspace$, and $\quasimhrdistspace$, respectively. We characterize the distribution hierarchy of all four families in \Cref{prop:hierarchy} and provide a graphical illustration in \Cref{fig:distribution-hierarchy}.
% While both regularity and MHR are common and widely used in the mechanism design literature, the quasi-regularity is original but also crucial for our analysis. It can be verified that every MHR distribution is also regular \cite{har16} and every regular distribution is also quasi-regular.
Both regular and MHR distributions admit the following equivalent definition.

\begin{lemma}[Folklore, see \cite{har16}]
\label{lem:regular equivalent definition}
\label{lem:mhr equivalent definition}
    A distribution $\prior$ is regular if and only if its induced revenue curve $\revcurve$ is concave, i.e., $\revcurve \equiv\ironrevcurve$.
    A distribution $\prior$ is MHR if and only if the cumulative hazard rate function $\cumhazardrate(\val)\triangleq -\ln(1-\cdf(\val))$ is convex.
\end{lemma}

%% file: source/structural.tex
\section{Distribution Hierarchy and Structural Results}
\label{sec:structural}

This section presents structural properties of regular, quasi-regular, MHR, and quasi-MHR distributions that might be of independent interest. Many of them are used for the revenue approximation analysis in the later sections.

\subsection{Basic Structural Results}
\label{subsec:basic structural results}

In this subsection, we present some basic structural properties of quasi-regular and quasi-MHR distributions.

First, we give two equivalent definitions of quasi-regularity in \Cref{prop:q-regular equivalent definition}. In this proposition, Condition~\ref{condition:quasi-regular:revenue curve} is referred to as the ``\emph{inscribed triangle property}'' in \cite[Appendix~D]{HR14}. A graphical illustration of this condition is provided in \Cref{fig:quasi-regular/mhr equivalent definition}. Both Conditions~\ref{condition:quasi-regular:revenue curve} and \ref{condition:quasi-regular:cdf} are tight in the \emph{worst-case} sense, even when moving from quasi-regular distributions $\quasiregulardistspace$ to regular distributions $\regulardistspace$. Specifically, for a fixed anchor quantile $\quant$ (resp.\ value $\val$) and its corresponding $\revcurve(\quant)$ (resp.\ $\cdf(\val)$), there exists a regular distribution (i.e., triangular distribution) $\prior \in \regulardistspace$ such that the inequality in the condition holds with equality for all quantiles $\quant\primed \in [\quant, 1]$ (resp.\ values $\val\primed \in [0, \val)$). As we will see in later sections, many approximation results allow for an approximation-preserving generalization from regular distributions $\regulardistspace$ to quasi-regular distributions $\quasiregulardistspace$, as their core relies on the inequalities in Condition~\ref{condition:quasi-regular:revenue curve} or \ref{condition:quasi-regular:cdf}.

\begin{figure}
    \centering
    \subfloat[Condition~\ref{condition:quasi-regular:revenue curve} in \Cref{prop:q-regular equivalent definition} for quasi-regular distributions $\quasiregulardistspace$.]{
    \input{figure/quasi-regular-geometric-definition}
    }
    ~~~~
    \subfloat[Condition~\ref{condition:quasi-MHR:cumhazardrate} in \Cref{prop:q-MHR equivalent definition} for quasi-MHR distributions $\quasimhrdistspace$.]{
    \input{figure/quasi-mhr-geometric-definition}
    }
    \caption{Graphical illustration of equivalent definitions of quasi-regular distributions $\quasiregulardistspace$ and quasi-MHR distributions $\quasimhrdistspace$.} 
    \label{fig:quasi-regular/mhr equivalent definition}
\end{figure}
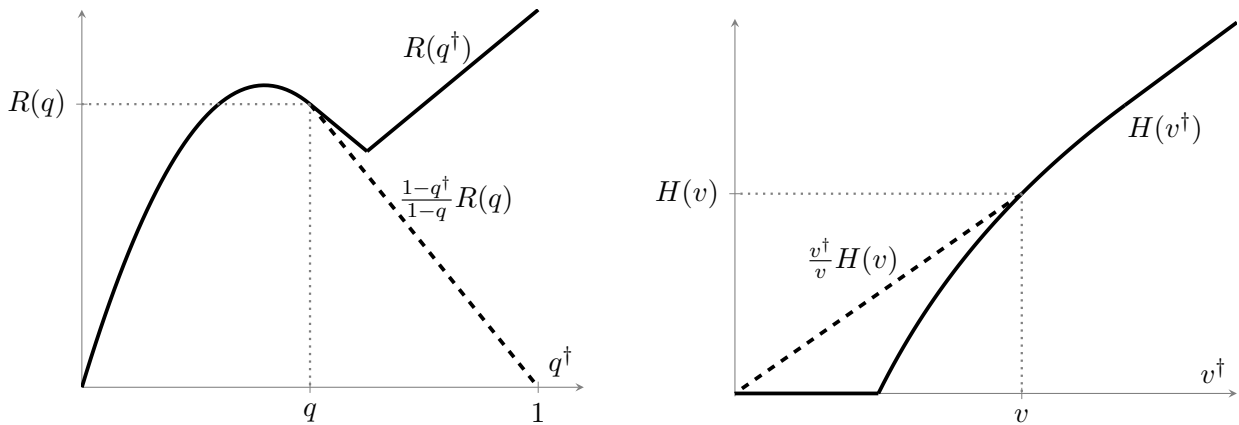

\begin{proposition}[Equivalent Definitions of Quasi-Regularity]
\label{prop:q-regular equivalent definition}
The following three conditions of a given distribution $\prior$ are equivalent:
\begin{enumerate}
    \item \emph{\textbf{(Quasi-Regularity)}} Distribution $\prior$ is quasi-regular, i.e., $\prior\in\quasiregulardistspace$.
    \item  \label{condition:quasi-regular:revenue curve} \emph{\textbf{(Revenue Curve Characterization)}} 
    In revenue curve $\revcurve$ induced by distribution $\prior$, for every anchor quantile $\quant \in[0, 1]$,
    \begin{align*}
    \forall \quant\primed \in [\quant, 1]:&\qquad 
    \revcurve(\quant\primed) \geq
    \frac{1-\quant\primed}{1-\quant}\revcurve(\quant) ~.
    \end{align*}
    % \intertext{
    \item \label{condition:quasi-regular:cdf}
    \emph{\textbf{(Probability Tail Bounds)}}
    For every anchor value $\val\in \supp(\prior)$,
    \begin{align*}
    ~\forall \val\primed \in [0,\val):& \qquad
        \cdf(\val\primed) \leq \frac{\val\primed}{\val\primed + \frac{1 - \cdf(\val)}{\cdf(\val)}\val}~.
    %     \qquad
    % ~\forall \val\primed \in [\val,\infty): \quad
    %     \cdf(\val\primed) \geq \frac{\val\primed}{\val\primed + \frac{1 - \cdf(\val)}{\cdf(\val)}\val}~.
    \end{align*}
\end{enumerate}
\end{proposition}
\begin{proof}
% [Proof of \Cref{prop:q-regular equivalent definition}]
    By definition, distribution $\prior$ is quasi-regular if and only if its conditional expected virtual value function $\cumvirtualval(\val) = \expect[\val\primed\sim\prior]{\virtualval(\val\primed)\given \val\primed \leq \val}$ is weakly increasing in $\val\in\supp(\prior)$. In other words, for every value $\val\in\supp(\prior)$, and every value $\val\primed\in[0, \val)$:
    \begin{align*}
        \frac{1}{\cdf(\val\primed)}\displaystyle\int_{0}^{\val\primed} \virtualval(t) \cdot\d\cdf(t)
        \leq 
        \frac{1}{\cdf(\val)}\displaystyle\int_{0}^{\val} \virtualval(t) \cdot\d\cdf(t)
    \end{align*}
    Let $\quant \triangleq 1 - \cdf(\val)$ and $\quant\primed \triangleq 1 - \cdf(\val\primed)$. Invoking the relation between virtual value $\virtualval$ and revenue curve $\revcurve$, the above inequality is equivalent to 
    \begin{align*}
        \frac{1}{1 - \quant\primed}\displaystyle\int_{\quant\primed}^1 \revcurve'(s)\cdot \d s \leq 
        \frac{1}{1 - \quant}\displaystyle\int_{\quant}^1 \revcurve'(s)\cdot\d s
    \end{align*}
    which is equivalent to 
    \begin{align*}
    \revcurve(\quant\primed) \geq
    \frac{1-\quant\primed}{1-\quant}\revcurve(\quant) 
    \end{align*}
    where we use the assumption that $\revcurve(1) = 0$ (\Cref{sec:prelim}). Thus, quasi-regularity and Condition~\ref{condition:quasi-regular:revenue curve} are equivalent. Moreover, by definition, $\revcurve(\quant) = \val(1 - \cdf(\val))$ and $\revcurve(\quant\primed) = \val\primed(1 - \cdf(\val\primed))$. Hence, the above inequality is equivalent to
    \begin{align*}
        \val\primed(1 - \cdf(\val\primed)) \geq 
        \frac{\cdf(\val\primed)}{\cdf(\val)} \val(1 - \cdf(\val))
    \end{align*}
    After rearranging, it becomes 
    \begin{align*}
        \cdf(\val\primed) \leq \frac{\val\primed}{\val\primed + \frac{1 - \cdf(\val)}{\cdf(\val)}\val}
    \end{align*}
    which is exactly Condition~\ref{condition:quasi-regular:cdf} as desired.
\end{proof}

Similar to \Cref{prop:q-regular equivalent definition} for quasi-regularity, the quasi-MHR condition also admits two equivalent definitions summarized in \Cref{prop:q-MHR equivalent definition} below. In this proposition, a graphical illustration of Condition~\ref{condition:quasi-MHR:cumhazardrate} is provided in \Cref{fig:quasi-regular/mhr equivalent definition}. Both Conditions~\ref{condition:quasi-MHR:cumhazardrate} and \ref{condition:quasi-MHR:cdf} are tight in the \emph{worst-case} sense, even when moving from quasi-MHR distributions $\quasimhrdistspace$ to MHR distributions $\mhrdistspace$. Specifically, for a fixed anchor value $\val$ and its corresponding $\cumhazardrate(\val)$ (resp.\ $\cdf(\val)$), there exists a MHR distribution (i.e., exponential distribution) $\prior \in \mhrdistspace$ such that the inequality in the condition holds with equality for all values $\val\primed \in [0, \val)$. As we will see in later sections, many approximation results allow for an approximation-preserving generalization from MHR distributions $\mhrdistspace$ to quasi-MHR distributions $\quasimhrdistspace$, as their core relies on the inequalities in Condition~\ref{condition:quasi-MHR:cumhazardrate} or \ref{condition:quasi-MHR:cdf}.
\begin{proposition}[Equivalent Definition of Quasi-MHR]
\label{prop:q-MHR equivalent definition}
The following two conditions of a given distribution $\prior$ are equivalent:
\begin{enumerate}
    \item \textbf{\emph{(Quasi-MHR)}} Distribution $\prior$ is quasi-MHR, i.e., $\prior\in\quasimhrdistspace$.
    \item \label{condition:quasi-MHR:cumhazardrate}
    \textbf{\emph{(Cumulative Hazard Rate Characterization)}}
    In cumulative hazard rate function $\cumhazardrate$ induced by distribution $\prior$, for every value $\val \in \supp(\prior)$,
    \begin{align*}
        \forall \val\primed \in [0, \val):\qquad 
        \cumhazardrate(\val\primed) \leq \frac{\val\primed}{\val}\cumhazardrate(\val)~,
    \end{align*}
    \item \label{condition:quasi-MHR:cdf}
    \textbf{\emph{(Probability Tail Bounds)}}
    For every value $\val\in \supp(\prior)$,
    \begin{align*}
    \forall \val\primed \in [0,\val):& \qquad
    \cdf(\val\primed)
    \leq 1 - \left(1 - \cdf(\val)\right)^{\frac{\val\primed}{\val}}~.
    % \\
    % \forall \val\primed \in [\val,\infty):& \qquad
    % \cdf(\val\primed)
    % \geq 1 - \left(1 - \cdf(\val)\right)^{\frac{\val\primed}{\val}}
    \end{align*}
\end{enumerate}
\end{proposition}
\begin{proof}
    By definition, distribution $\prior$ is quasi-regular if and only if its conditional expected hazard rate function $\cehazardrate(\val) = \frac{\cumhazardrate(\val)}{\val}$ is weakly increasing in $\val\in\supp(\prior)$. 
    Rearranging the terms shows the equivalence to Condition~\ref{condition:quasi-MHR:cumhazardrate} immediately. Moreover,
    recall that $\cumhazardrate(\val) = -\ln(1 - \cdf(\val)$. Hence,
    for every value $\val\in\supp(\prior)$, and every value $\val\primed\in[0, \val)$:
    \begin{align*}
        \frac{-\ln(1 - \cdf(\val\primed))}{\val\primed}
        \leq 
        \frac{-\ln(1 - \cdf(\val))}{\val}
    \end{align*}
    After rearranging the term, we obtain 
    \begin{align*}
    \cdf(\val\primed)
    \leq 1 - \left(1 - \cdf(\val)\right)^{\frac{\val\primed}{\val}}
    \end{align*}
    which finishes the proof.
\end{proof}

Equipped with \Cref{prop:q-MHR equivalent definition}, we establish three technical properties for quasi-MHR distributions, summarized in the following lemma. The analogy properties for MHR distributions is known in \cite{HR09,AGM09}. These properties are crucial for analyzing various revenue approximation results in the subsequent sections.

\begin{lemma}
\label{lem:quasi-mhr:structural results}
    For any quasi-MHR distribution $\prior\in\quasimhrdistspace$, the following three properties holds:
    \begin{enumerate}
        \item \label{property:quasi-mhr:monopoly quantile lower bound}\emph{\textbf{(Monopoly Quantile Lower Bound)}} Its monopoly quantile $\optquant$ is at least $\frac{1}{e}$. Moreover, for the exponential distribution (which is also MHR), its monopoly quantile is exactly $\frac{1}{e}$.
        \item \label{property:quasi-mhr:revenue approximate welfare}\emph{\textbf{(Revenue Approximates Welfare via Price Posting)}} For any threshold $\threshold\in\reals_+$, the expected revenue by posting a price weakly higher than $\threshold$ is at least a $\frac{1}{3}$-approximation to the expected welfare of posting price $\threshold$, i.e.,
        \begin{align*}
            \max_{\price \geq \threshold}~\price\cdot (1 - \cdf(\price)) \geq \frac{1}{3}\cdot \expect{\val\cdot \indicator{\val \geq \threshold}}
        \end{align*}
        If threshold $\threshold$ is weakly higher than the monopoly reserve $\optreserve$, the expected revenue by posting price $\threshold$ directly is at least a $\frac{1}{e+1}$-approximation to the expected welfare of posting price $\threshold$, i.e.,
        \begin{align*}
            \threshold \cdot (1 - \cdf(\threshold)) 
            \geq 
            \frac{1}{e + 1} \cdot \expect{\val\cdot \indicator{\val \geq \threshold}}
        \end{align*}
        Moreover, there exists quasi-regular distributions such that the approximation factors of both $\frac{1}{3}$ and $\frac{1}{e + 1}$ are tight.
        \item \label{property:quasi-mhr:revenue approximate welfare duplicating} \emph{\textbf{(Revenue Approximates Welfare via Buyer Duplicating)}} For any threshold $\threshold\in\reals_+$, the expected conditional virtual surplus of the first order statistic $\val_{(1:2)} = \max\{\val_1,\val_2\}$ from two i.i.d.\ samples $\val_1,\val_2$ from distribution $\prior$ is at least a $\frac{1}{3}$-approximation to the expected conditional social welfare of first order statistic $\val_{(1:2)}$, i.e.,
        \begin{align*}
            \expect[\val_{(1:2)}\sim \prior_{(1:2)}]{\virtualval(\val_{(1:2)})\given \val_{(1:2)} \geq \threshold}
            \geq 
            \frac{1}{3}\cdot \expect[\val_{(1:2)}\sim \prior_{(1:2)}]{\val_{(1:2)}\given \val_{(1:2)} \geq \threshold}
        \end{align*}
        Moreover, the approximation factor of $\frac{1}{3}$ is tight with threshold $\threshold = 0$ for the exponential distribution (which is also MHR).
    \end{enumerate}
\end{lemma}
By setting threshold $\threshold = 0$, Property~\ref{property:quasi-mhr:revenue approximate welfare} in \Cref{lem:quasi-mhr:structural results} guarantees that the optimal revenue is a $\frac{1}{3}$-approximation to the expected welfare for quasi-MHR distributions $\quasimhrdistspace$. Note that for MHR distributions $\mhrdistspace$, the tight approximation between the optimal revenue and the expected welfare is $\frac{1}{e}$, whose upper bound is achieved by the exponential distribution. Moreover, due to the concavity of the revenue curve for MHR distributions, if threshold $\threshold \geq \optreserve$, posting price equal to threshold $\threshold$ maximizes the expected revenue among all prices $\price \geq \threshold$, and thus left-hand sides of two inequalities in Property~\ref{property:quasi-mhr:revenue approximate welfare} become equivalent.
In this sense, this property also suggests a strict but bounded separation between MHR distributions $\mhrdistspace$ and quasi-MHR distributions $\quasimhrdistspace$.

\begin{proof}[Proof of \Cref{lem:quasi-mhr:structural results}]
    We prove three properties in the lemma statement sequentially.

    \xhdr{Property~\ref{property:quasi-mhr:monopoly quantile lower bound}- Monopoly Quantile Lower Bound:} We prove this property by introducing an auxiliary distribution $\auxprior$ as follows. Let $\optreserve$ be the optimal reserve for the original quasi-MHR distribution $\prior$. Let auxiliary distribution~$\auxprior$ have cumulative density function $\auxcdf(\val) = 1 - e^{-\cehazardrate(\optreserve)\cdot \val}$ with support $\supp(\auxprior) = [0, \infty)$. 
    Note that
    \begin{align*}
        \auxcdf(\optreserve) = 1 - e^{-\cehazardrate(\optreserve)\cdot \optreserve}
        =
        1 - e^{\frac{\ln(1 - \cdf(\optreserve))}{{\optreserve}}\cdot \optreserve}
        =
        \cdf(\optreserve)
    \end{align*}
    where the first equality holds due to the construction of auxiliary distribution $\auxprior$, and the second equality holds due to the definition of conditional hazard rate function $\cehazardrate(\cdot)$. Similarly,
    for every value $\val\leq\optreserve$,
    \begin{align*}
        \auxcdf(\val) = 1 - e^{-\cehazardrate(\optreserve)\cdot \val}
        \geq 
        1 - e^{-\cehazardrate(\val)\cdot \val}
        =
        1 - e^{\frac{\ln(1 - \cdf(\val))}{{\val}}\cdot \val}
        =
        \cdf(\val)
    \end{align*}
    where the inequality holds since distribution $\prior$ is quasi-MHR and $\val\leq \optreserve$.  Putting two pieces together, we know that for every price $\price \in[0,\optreserve]$, 
    \begin{align*}
        \price(1 - \cdf(\price)) \geq \price(1 - \auxcdf(\price))
    \end{align*}
    and the inequality is binding at price $\price = \optreserve$. Hence, the monopoly reserve (resp.\ monopoly quantile) under auxiliary distribution $\auxprior$ is weakly larger (resp.\ weakly smaller) than the one under distribution $\prior$. Hence, it suffices to verify that the monopoly quantile under auxiliary MHR distribution $\auxprior$ is at least $\frac{1}{e}$, which has been already shown in \cite[Lemma~1]{AGM09}.

    \xhdr{Property~\ref{property:quasi-mhr:revenue approximate welfare}- Revenue Approximates Welfare via Price Posting:}
    Fix an arbitrary quasi-MHR distribution $\prior\in\quasimhrdistspace$. Without loss of generality, we normalize the monopoly revenue of distribution $\prior$ to be 1. 

    \xhdr{Property~\ref{property:quasi-mhr:revenue approximate welfare}.i ($\frac{1}{3}$-Approximation).}
    We first prove the $\frac{1}{3}$-approximation. Note that the left- and right-hand side of the inequality can be interpreted as the conditional optimal expected revenue and conditional expected welfare (conditioning on value $\val \geq \threshold$). By the definition, this conditional distribution is also quasi-MHR. Hence, it suffices to verify the $\frac{1}{3}$-approximation inequality for threshold $\threshold = 0$. Under our assumption that the monopoly revenue is 1, we only need to show the expected welfare $\expect[\val\sim\prior]{\val}$ is at most 3.

    We upper bound the expected welfare $\expect[\val\sim\prior]{\val}$ using the following argument. Since the monopoly revenue is 1, for every value $\val \leq e$, we have
    \begin{align*}
        1 - \cdf(\val) \leq \frac{1}{\val}
    \end{align*}
    Moreover, for every value $\val \geq e$, we have
    \begin{align*}
        1 - \cdf(\val)  = e^{- \cumhazardrate(\val)}
        \leq e^{- \frac{\cumhazardrate(e)}{e}\val} 
        \leq 
        e^{-\frac{1}{e}\val}
    \end{align*}
    where the first inequality holds since distribution $\prior$ is quasi-MHR, and the second distribution holds since $\cumhazardrate(e) = -\ln(1 - \cdf(e)) \geq 1$ as we argued above.

    Let auxiliary distribution $\auxprior$ have support $\supp(\auxprior) = [1, \infty)$, and cumulative density function $\auxcdf(\val) = 1 - \frac{1}{\val}$ for $\val\in [1, e]$, $\auxcdf(\val) = 1 - e^{-\frac{\val}{e}}$ for $\val \in[e, \infty)$. By construction, auxiliary distribution $\auxprior$ first-order stochastically dominants distribution $\prior$. Thus, the expected welfare $\expect[\val\sim\prior]{\val}$ can be upper bounded by $
        \expect[\val\sim\prior]{\val} \leq 
        \expect[\val\sim\auxprior]{\val}
        = 3$,
    where the equality holds due to the construction of auxiliary distribution $\auxprior$. Note that the auxiliary distribution $\auxprior$ is also quasi-MHR and shows the tightness of the $\frac{1}{3}$-approximation. 

    \xhdr{Property~\ref{property:quasi-mhr:revenue approximate welfare}.ii ($\frac{1}{e+1}$-Approximation).}
    Now we assume threshold $\threshold\geq \optreserve$ and prove the $\frac{1}{e + 1}$-approximation. Here we consider two cases separately.

    First, we verify the case where $\cumhazardrate(\threshold) \geq \frac{\threshold}{e}$. In this case, define auxiliary distribution $\auxprior$ with cumulative density function $\auxcdf(\val) = \cdf(\val)$ for $\val \in [0, \threshold)$ and $\auxcdf(\val) = 1 - e^{-\frac{\cumhazardrate(\threshold)}{\threshold}\val}$ for $\val \in [\threshold, \infty)$. As a sanity check, for every value $\val \geq \threshold$, we have
    \begin{align*}
        \cdf(\val) = 1 - e^{-\cumhazardrate(\val)} 
        \geq 1 - e^{-\frac{\cumhazardrate(\threshold)}{\threshold}\val}
        =
        \auxcdf(\val)
    \end{align*}
    where the inequality holds since distribution $\prior$ is quasi-regular. This ensures 
    \begin{align*}
        \expect[\val\sim\prior]{\val\cdot \indicator{\val \geq \threshold}}
        \leq 
        \expect[\val\sim\auxprior]{\val\cdot\indicator{{\val\geq \threshold}}}
    \end{align*}
    and consequently allows us to verify the approximation guarantee as follows.
    \begin{align*}
        &\threshold \cdot (1 - \cdf(\threshold)) 
        - 
        \frac{1}{e + 1} \cdot \expect[\val\sim\prior]{\val\cdot \indicator{\val \geq \threshold}}
        \\
        {}\geq {} &
        \threshold \cdot (1 - \cdf(\threshold)) 
        - 
        \frac{1}{e + 1} \cdot \expect[\val\sim\auxprior]{\val\cdot \indicator{\val \geq \threshold}}
        \\
        {} ={} &
        \threshold\cdot e^{-\cumhazardrate(\threshold)}
        -
        \frac{1}{e + 1} \cdot
        \displaystyle\int_\threshold^\infty \val e^{-\frac{\cumhazardrate(\threshold)}{\threshold}\val}
        \cdot\frac{\cumhazardrate(\threshold)}{\threshold}\val \cdot \d\val
        \\
        {} = {} &
        \threshold\cdot e^{-\cumhazardrate(\threshold)}
        -
        \frac{\threshold}{e + 1} 
        \frac{\cumhazardrate(\threshold) + 1}{\cumhazardrate(\threshold)}
        e^{-\cumhazardrate(\threshold)}
        \\
        \geq {} & 0
    \end{align*}
    where the last inequality holds since $\threshold \geq \optreserve \geq 1$ (due to assumption that monopoly revenue is 1) and thus $\cumhazardrate(\threshold) \geq \frac{\threshold}{e} \geq \frac{1}{e}$ (due to the case assumption).
    
    Next we verify the case where $\cumhazardrate(\threshold) < \frac{\threshold}{e}$. Let $a\in[1, e]$
    % $\triangleq e^{-\LambertFunc(\frac{\cumhazardrate(\threshold)}{\threshold})}$ 
    be the solution of equation $\frac{\ln(a)}{a} = \frac{\cumhazardrate(\threshold)}{\threshold}$. The case assumption that $\cumhazardrate(\threshold) < \frac{\threshold}{e}$ ensures the existence of $a$.
    Since the monopoly revenue is 1, threshold $\threshold\geq \optreserve \geq 1$ and the cumulative hazard rate function $\cumhazardrate(\threshold) \geq \ln(\threshold)$ which further implies threshold $\threshold \leq a\leq e$.
    Define auxiliary distribution $\auxprior$ with support $\supp(\auxprior) = [0, \infty)$ and cumulative density function $\auxcdf(\val) = 1 - e^{-\frac{\ln(a)}{a}\val}$ for $\val \in [0, a]$, $\auxcdf(\val) = 1 - \frac{1}{\val}$ for $\val \in[a, e]$, and $\auxcdf(\val) = 1 - e^{-\frac{\val}{e}}$ for $\val \in[e, \infty)$. As a sanity check, for every value $\val \in [\threshold, a]$, we have 
    \begin{align*}
        \cdf(\val) = 1 - e^{-\cumhazardrate(\val)}
        \overset{(a)}{\geq} 1 - e^{-\frac{\cumhazardrate(\threshold)}{\threshold}\val}
        \overset{(b)}{=}
        1 - e^{-\frac{\ln(a)}{a}\val}
        \overset{}{=} \auxcdf(\val)
    \end{align*}
    where inequality~(a) holds since distribution $\prior$ is quasi-MHR; and equality~(b) holds due to the definition of $a$. Since the monopoly revenue is 1, for every value $\val \in [a, e]$, we have 
    \begin{align*}
        \cdf(\val) \geq 1 - \frac{1}{\val} = \auxcdf(\val)
    \end{align*}
    Finally, for every value $\val\in[e, \infty)$, we have
    \begin{align*}
        \cdf(\val) = 1 - e^{-\cumhazardrate(\val)}
        \overset{(a)}{\geq} 1 - e^{-\frac{\cumhazardrate(e)}{e}\val}
        \overset{(b)}{\geq} 1 - e^{-\frac{\val}{e}} = \auxcdf(\val)
    \end{align*}
    where inequality~(a) holds since distribution $\prior$ is quasi-MHR; and inequality~(b) holds since $\cumhazardrate(e) \geq 1$ as the monopoly revenue is 1. Putting the three pieces together, we know $\cdf(\val) \geq \auxcdf(\val)$ for every value $\val \geq \threshold$ and thus
    \begin{align*}
        \expect[\val \sim\prior]{\val\cdot \indicator{\val  \geq \threshold}} 
        \leq  
        \expect[\val \sim\auxprior]{\val\cdot \indicator{\val \geq \threshold}}
        % \\
        \leq 
        \expect[\val \sim\auxprior]{\val\cdot \indicator{\val \geq 1}}
        =  
        \frac{1 + \ln(a)}{\ln(a)}\cdot e^{-\frac{\ln(a)}{a}}
        -
        \frac{\left(1 - \ln(a)\right)^2}{\ln(a)}
    \end{align*}
    where the second inequality holds since $\threshold \geq \optreserve \geq 1$ and the last equality holds by calculation.

    On the other hand, we claim that
    \begin{align*}
        \threshold\cdot (1 - \cdf(\threshold)) = \threshold\cdot (1 - \auxcdf(\threshold)) \geq 1\cdot (1 - \auxcdf(1)) = e^{-\frac{\ln(a)}{a}}
    \end{align*}
    where the first equality holds since $\cdf(\threshold) = 1 - e^{-{\cumhazardrate(\threshold)}} = 1 - e^{-\frac{\ln(a)}{a}\threshold} = \auxcdf(\threshold)$ which is implied by the definition of $a$ and the fact that $\threshold \leq a$ argued above.
    To see why the inequality holds, note that 
    \begin{align*}
        \frac{\threshold\cdot (1 - \auxcdf(\threshold))}{1\cdot (1 - \auxcdf(1))}
        =
        \threshold\cdot e^{-\frac{\ln(a)}{a}(\threshold-1)}
        =
        \threshold\cdot e^{-\frac{\cumhazardrate(\threshold)}{\threshold}(\threshold-1)}
        \geq 
        \threshold\cdot e^{-\frac{1}{e}(\threshold-1)}
        \geq 1
    \end{align*}
    where the first equality holds due to the construction of auxiliary distribution $\auxprior$ and the fact that $\threshold\leq a$ argued above; the second equality holds due to the definition of $a$; the first inequality holds due to the case assumption that $\cumhazardrate(\threshold) \leq \frac{\threshold}{e}$; and the second inequality holds since $1 \leq \threshold \leq a \leq e$ argued above.

    Combining the lower bound of $ \threshold\cdot (1 - \cdf(\threshold))$ and the upper bound of $\expect[\val \sim\prior]{\val\cdot \indicator{\val  \geq \threshold}} $ derived above, it suffices to verify that for every $a \in [1, e]$,
    \begin{align*}
        e + 1 \geq \frac{1 + \ln(a)}{\ln(a)}
        -
        \frac{\left(1 - \ln(a)\right)^2}{\ln(a)}\cdot e^{\frac{\ln(a)}{a}}
    \end{align*}
    Thus, we finish our argument by noting that the right-hand side is maximized at $e + 1$ with $a = e$ among all $a \in[1, e]$. 

    To see the tightness of the $\frac{1}{e + 1}$-approximation, consider the following quasi-MHR distribution~$\prior$ with support $\supp(\prior) = [1, \infty)$ and cumulative density function $\cdf(\val) = 1 - e^{-\frac{\val}{e}}$ for $\val\in[1,\infty)$. In this distribution, the monopoly reserve $\optreserve = 1$ has a probability mass of $1 - e^{-\frac{1}{e}}$. Consider threshold $\threshold = 1 + \varepsilon$ with sufficiently small $\varepsilon > 0$. By straightforward calculation, we know
    \begin{align*}
        \threshold\cdot (1 - \cdf(\threshold)) = e^{-\frac{1}{e}} - o_\varepsilon(1)
        \;\;
        \mbox{and}
        \;\;
        \expect[\val \sim\prior]{\val\cdot \indicator{\val  \geq \threshold}} = (e + 1) e^{-\frac{1}{e}} - o_\varepsilon(1)
    \end{align*}
    which shows the approximation tightness as desired.

    \xhdr{Property~\ref{property:quasi-mhr:revenue approximate welfare duplicating}- Revenue Approximates Welfare via Buyer Duplicating:}
    The inequality in this property can be re-written as 
    \begin{align*}
        \displaystyle\int_{\threshold}^{\infty} \left(\val - \frac{1 - \cdf(\val)}{\pdf(\val)}\right)\cdot 2\cdf(\val) \cdot \pdf(\val)\cdot \d\val
        \geq 
        \frac{1}{3}\displaystyle\int_{\threshold}^{\infty} \val \cdot 2\cdf(\val) \cdot \pdf(\val)\cdot \d\val
    \end{align*}
    After rearranging and integration by parts, it is equivalent to 
    \begin{align}
    \label{eq:quasi-mhr:revenue approxmation critical inequality}
        \threshold\cdot \left(1 - \left(\cdf(\threshold)\right)^2\right)
        -
        \int_{\threshold}^\infty (1 - \cdf(\val))
        \cdot \left(2\cdf(\val) - 1\right)\cdot \d\val
        \geq 0
    \end{align}
    To verify this inequality, our argument is divided into two cases depending on whether $\cdf(\threshold) \geq \frac{3}{4}$.

    First, suppose $\cdf(\threshold) \geq \frac{3}{4}$. In this case, we construct an auxiliary distribution $\auxprior$ as follows: let auxiliary distribution~$\auxprior$ have cumulative density function $\auxcdf(\val) = 1 - e^{-\cehazardrate(\threshold)\cdot \val}$ with support $\supp(\auxprior) = [0, \infty)$. Note that
    \begin{align*}
        \auxcdf(\threshold) = 1 - e^{-\cehazardrate(\threshold)\cdot \threshold}
        =
        1 - e^{\frac{\ln(1 - \cdf(\threshold))}{{\threshold}}\cdot \threshold}
        =
        \cdf(\threshold)
    \end{align*}
    where the first equality holds due to the construction of auxiliary distribution $\auxprior$, and the second equality holds due to the definition of conditional hazard rate function $\cehazardrate(\cdot)$. Similarly, for every value $\val \geq \threshold$,
    \begin{align*}
        \auxcdf(\val) = 1 - e^{-\cehazardrate(\threshold)\cdot \val}
        \leq 
        1 - e^{-\cehazardrate(\val)\cdot \val}
        =
        1 - e^{\frac{\ln(1 - \cdf(\val))}{{\val}}\cdot \val}
        =
        \cdf(\val)
    \end{align*}
    where the inequality holds since distribution $\prior$ is quasi-MHR and $\val\geq \threshold$. Putting two pieces together, we have
    \begin{align*}
        \threshold\cdot \left(1 - \left(\cdf(\threshold)\right)^2\right)
        &=
        \threshold\cdot \left(1 - \left(\auxcdf(\threshold)\right)^2\right)
        % \;\;\mbox{and}\;\;
        \\
        \int_{\threshold}^\infty (1 - \cdf(\val))
        \cdot \left(2\cdf(\val) - 1\right)\cdot \d\val
        &\leq 
        \int_{\threshold}^\infty (1 - \auxcdf(\val))
        \cdot \left(2\auxcdf(\val) - 1\right)\cdot \d\val
    \end{align*}
    where the second inequality holds due to the fact that $(1 - \cdf(\val))
         \left(2\cdf(\val) - 1\right)\leq (1 - \auxcdf(\val))
         \left(2\auxcdf(\val) - 1\right)$
    for every value $\val \geq t$, which is implied by $\cdf(\val) \geq \auxcdf(\val)$ and $\auxcdf(\val) \geq \auxcdf(\threshold) \geq \frac{3}{4}$.\footnote{Note $\frac{3}{4}$ is the unique maximizer of quadratic function $\auxfunc(x) = (1-x)(2x-1)$.} Hence, it suffices to verify inequality~\eqref{eq:quasi-mhr:revenue approxmation critical inequality} or equivalently Property~\ref{property:quasi-mhr:revenue approximate welfare duplicating} for auxiliary MHR distribution $\auxprior$, which has been already shown in \cite[Lemma 4.1]{HR09}.

    Second, suppose $\cdf(\threshold) < \frac{3}{4}$. Let $\val\primed\triangleq\cdf^{-1}(\frac{3}{4})$ be the value such that $\cdf(\val\primed) = \frac{3}{4}$. As a sanity check, we have $\val\primed > \threshold$. In this case, we construct an auxiliary distribution $\auxprior$ as follows: let auxiliary distribution~$\auxprior$ have cumulative density function $\auxcdf(\val) = 1 - e^{-\cehazardrate(\val\primed)\cdot \val}$ with support $\supp(\auxprior) = [0, \infty)$. Note that
    for every value $\val\leq\val\primed$,
    \begin{align*}
        \auxcdf(\val) = 1 - e^{-\cehazardrate(\val\primed)\cdot \val}
        \geq 
        1 - e^{-\cehazardrate(\val)\cdot \val}
        =
        1 - e^{\frac{\ln(1 - \cdf(\val))}{{\val}}\cdot \val}
        =
        \cdf(\val)
    \end{align*}
    where the inequality holds since distribution $\prior$ is quasi-MHR and $\val\leq \val\primed$. Similarly, for every value $\val \geq \threshold$,
    \begin{align*}
        \auxcdf(\val) = 1 - e^{-\cehazardrate(\threshold)\cdot \val}
        \leq 
        1 - e^{-\cehazardrate(\val)\cdot \val}
        =
        1 - e^{\frac{\ln(1 - \cdf(\val))}{{\val}}\cdot \val}
        =
        \cdf(\val)
    \end{align*}
    where the inequality holds since distribution $\prior$ is quasi-MHR and $\val \geq \val\primed$. Putting two pieces together, we have
    \begin{align*}
        \threshold\cdot \left(1 - \left(\cdf(\threshold)\right)^2\right)
        &\geq
        \threshold\cdot \left(1 - \left(\auxcdf(\threshold)\right)^2\right)
        % \;\;\mbox{and}\;\;
        \\
        \int_{\threshold}^\infty (1 - \cdf(\val))
        \cdot \left(2\cdf(\val) - 1\right)\cdot \d\val
        &\leq 
        \int_{\threshold}^\infty (1 - \auxcdf(\val))
        \cdot \left(2\auxcdf(\val) - 1\right)\cdot \d\val
    \end{align*}
    where the first inequality holds since $\cdf(\threshold) \leq \auxcdf(\threshold)$, and the second inequality holds due to the fact that $(1 - \cdf(\val))
         \left(2\cdf(\val) - 1\right)\leq (1 - \auxcdf(\val))
         \left(2\auxcdf(\val) - 1\right)$
    for every value $\val \geq 0$. Hence, it suffices to verify inequality~\eqref{eq:quasi-mhr:revenue approxmation critical inequality} or equivalently Property~\ref{property:quasi-mhr:revenue approximate welfare duplicating} for auxiliary MHR distribution $\auxprior$, which has been already shown in \cite[Lemma 4.1]{HR09}.

    Finally, to see the tightness of the bound. Consider exponential distribution $\prior$ with cumulative density function $\prior(\val) = 1-e^{-\val}$. Let threshold $\threshold = 0$. By direct computation, we have $\expect[\val_{(1:2)}\sim \prior_{(1:2)}]{\virtualval(\val_{(1:2)})} = \frac{1}{2}
    $
    and $\expect[\val_{(1:2)}\sim \prior_{(1:2)}]{\val_{(1:2)}}=\frac{3}{2}$, which completes the analysis as desired.
\end{proof}

\subsection{Relations among Distribution Families}
\label{sec:hierarchy:relation}

This subsection characterizes the relations among the regularity, quasi-regularity, MHR, and quasi-MHR conditions. A graphical illustration of \Cref{prop:hierarchy} below can be found in \Cref{fig:distribution-hierarchy}. 

\begin{proposition}[Distribution Hierarchy]
\label{prop:hierarchy}
% \begin{flushleft}
The four families $\regulardistspace$, $\quasiregulardistspace$, $\mhrdistspace$, and $\quasimhrdistspace$ of regular, quasi-regular, MHR, and quasi-MHR distributions satisfy the following relations:
\begin{enumerate}[label=(\roman*)]
    \item $\mhrdistspace\subsetneq (\regulardistspace\cap\quasimhrdistspace)$. In particular, consider distribution $\prior$ with support $\supp(\prior)=[1/e, \infty)$ and cumulative density function $\cdf(\val) = 1 - 1/(e\val)$ for value $\val\in[1/e, 1]$ and $\cdf(\val) = 1 - e^{-\val}$ for value $\val\in[1,\infty)$. Distribution $\prior$ is regular and quasi-MHR, but not MHR, i.e., $\prior \in ((\regulardistspace \cap \quasimhrdistspace) \setminus \mhrdistspace)$.
    
    \item $(\regulardistspace\cup\quasimhrdistspace) \subsetneq \quasiregulardistspace$. In particular, consider distribution $\prior$ with support $\supp(\prior) = [2, \infty)$ and cumulative density function $\cdf(\val) = 1 - 1/\val$ for value $\val\in[2, \infty)$. Distribution $\prior$ is quasi-regular but neither quasi-MHR nor regular, i.e., $\prior \in (\quasiregulardistspace \setminus (\regulardistspace \cup \quasimhrdistspace))$.
    
    \item $\quasimhrdistspace\backslash\regulardistspace\not=\emptyset$ and $\regulardistspace\backslash\quasimhrdistspace\not=\emptyset$.
    In particular, distribution $\prior(\val) = 1-e^{-\val}$ with support $\supp(\prior)=[1,\infty)$ is quasi-MHR but not regular, i.e., $\prior \in (\quasimhrdistspace\backslash\regulardistspace)$. Distribution  $\prior(\val) = 1-1/\val$ with support $\supp(\prior)=[1,\infty)$ is regular but not quasi-MHR, i.e., $\prior \in (\regulardistspace \setminus \quasimhrdistspace)$.
\end{enumerate}
% \end{flushleft}
\end{proposition}
\begin{proof}
    All hierarchical relationships can be directly verified using the original definitions of the distribution classes, except $\quasimhrdistspace\subseteq \quasiregulardistspace$. To verify $\quasimhrdistspace\subseteq \quasiregulardistspace$, note that Condition~\ref{condition:quasi-MHR:cdf} in \Cref{prop:q-MHR equivalent definition} for quasi-MHR distributions implies Condition~\ref{condition:quasi-regular:cdf} in \Cref{prop:q-regular equivalent definition} for quasi-regular distributions. This finishes the proof of the proposition statement.
\end{proof}

\begin{comment}
\color{red}

\begin{itemize}
    \item $\mhrdistspace \ni \prior$: $\cdf(\val) = 1 - e^{-\val}$ for $\val \geq 0$.

    \item $((\regulardistspace \cap \quasimhrdistspace) \setminus \mhrdistspace) \ni \prior$: $\cdf(\val) = 1 - 1 / (e \val)$ for $\val \in [1 / e, 1]$ and $\cdf(\val) = 1 - e^{-\val}$ for $\val \geq 1$.

    \item $(\regulardistspace \setminus \quasimhrdistspace) \ni \cdf$: $\cdf(\val) = 1 - 1 / \val$ for $\val \geq 1$.

    \item $(\quasimhrdistspace \setminus \regulardistspace) \ni \cdf$: $F(\val) = (1 - e^{-\val}) \cdot \indicator{\val \geq 1}$

    \item $(\quasiregulardistspace \setminus (\regulardistspace \cup \quasimhrdistspace)) \ni \prior$: $\cdf(\val) = (1 - 1 / \val) \cdot \indicator{\val \geq 2}$.
\end{itemize}

\color{black}

\end{comment}

\subsection{Order Statistics of Symmetric Regular/MHR Distributions}
\label{sec:hierarchy:iid}

In this subsection, we characterize the order statistic distributions generated from symmetric regular or MHR distributions.

\xhdr{Symmetric Regular Distributions.} We first consider regular distributions.
% \label{subsec:order:regular}

\begin{theorem}[Order Statistics of Symmetric Regular Distributions]
\label{thm:order:regular}
\begin{flushleft}
Given $n \geq 1$ many i.i.d.\ regular values $\vali \sim \prior$ for $i \in [n]$, every order statistic $\val_{(k:n)} \sim \prior_{(k:n)}$ for $k \in [n]$ also is regular.
\end{flushleft}
\end{theorem}

At a high level, our argument for  \Cref{thm:order:regular} relies on the following two-step argument: first, we prove the theorem statement for triangular distributions (i.e., a subclass of distributions with closed-form formulation) in \Cref{lem:order:regular:triangle}. For the triangular distributions, we can explicitly characterize the virtual value of its order statistic in closed form. Thus, applying an induction argument enables us to prove \Cref{lem:order:regular:triangle}. In the second step, we present a reduction argument from general distributions to triangular distributions.

\begin{lemma}
\label{lem:order:regular:triangle}
\begin{flushleft}
    Given $a \geq 0$, $b\geq 0$, and $n \geq 1$ many i.i.d.\ regular values $\vali\sim \auxprior$ with cumulative density function $\auxcdf(\val) = 
    \frac{\val - b}{\val - b + a}
    $ and support $\supp(\auxcdf) = [b, \infty)$, every order statistics $\val_{(k:n)} \sim \auxprior_{(k:n)}$ for $k \in [n]$ also is regular.
\end{flushleft}
\end{lemma}
\begin{proof}
    We prove the lemma statement with a two-step argument. In the first step, we prove the statement for the special case of $a \geq 0$, $b = 0$. In the second step, we show how the statement for general $(a, b)$ can be reduced to the special case argued in step 1. In the remainder of the analysis, with slight abuse of notations, we shorthand $\auxcdf_{(k)}$ for $\auxcdf_{(k:n)}$.

    First, we prove the lemma statement for distribution $\auxprior(\val) = \frac{\val}{\val + a}$ with support $\supp(\auxprior) = [1, \infty)$. Let $\virtualval_{(k)}$ be the virtual value function of distribution $\auxprior_{(k)}$. It suffices to show that $\virtualval_{(k)}(\val)$ is weakly increasing, or equivalently, its derivative $\virtualval_{(k)}'(\val)$ is non-negative. Note that 
    \begin{align*}
        \virtualval_{(k)}'(\val) 
        &{}=
        1 - \frac{\left(\auxpdf_{(k)}(\val)\right)^2 - \left(1-\auxcdf_{(k)}(\val)\right)\auxdpdf_{(k)}(\val)}{\left(\auxpdf_{(k)}(\val)\right)^2}
    \end{align*}
    After rearranging, it suffices to show 
    \begin{align*}
        2\left(\auxpdf_{(k)}(\val)\right)^2 + \left(1-\auxcdf_{(k)}(\val)\right)\auxdpdf_{(k)}(\val)
        \geq 0
    \end{align*}
    Invoking the fact that $\auxdpdf_{(k)}(\val) = \auxpdf_{(k)}(\val)(\frac{n-k}{\val} - \frac{n + 1}{\val + a})$ (which holds due to the specific form of $\auxprior(\val) = \frac{\val}{\val + a}$), the above inequality becomes 
    \begin{align*}
        2\auxpdf_{(k)}(\val) \geq \left(\frac{n + 1}{\val + a} - \frac{n-k}{\val}\right)
        \left(1-\auxcdf_{(k)}(\val)\right)
    \end{align*}
    which is satisfied trivially when $\val \leq \frac{n-k}{k + 1}a$. 
    
    To verify the inequality for $\val > \frac{n - k}{k + 1}a$, we define auxiliary function $\auxfunc(\val) \triangleq 2\auxpdf_{(k)}(\val)(\frac{n + 1}{\val + a} - \frac{n-k}{\val})^{-1} - 1+\auxcdf_{(k)}(\val)$. Since $\lim_{\val\rightarrow\infty}\auxfunc(\val) = 0$, it suffices to show that $\auxfunc'(\val) \leq 0$ for all $\val > \frac{n - k}{k + 1}a$. Note that 
    \begin{align*}
        \auxfunc'(\val) = \frac{2\left(\frac{n+1}{(\val + a)^2} - \frac{n - k}{\val^2}\right)}{\left(\frac{n + 1}{\val + a} - \frac{n - k}{\val}\right)^2}\auxpdf_{(k)}(\val) + 
        \frac{2}{\frac{n + 1}{\val + a} - \frac{n - k}{\val}}\auxdpdf_{(k)}(\val) + \auxpdf_{(k)}(\val)
        =
        \frac{2\left(\frac{n+1}{(\val + a)^2} - \frac{n - k}{\val^2}\right)}{\left(\frac{n + 1}{\val + a} - \frac{n - k}{\val}\right)^2}\auxpdf_{(k)}(\val) 
        - \auxpdf_{(k)}(\val)
    \end{align*}
    where the first equality holds by algebra and the second equality holds since $\auxdpdf_{(k)}(\val) = \auxpdf_{(k)}(\val)(\frac{n-k}{\val} - \frac{n + 1}{\val + a})$. It is now sufficient to verify that for every $\val > \frac{n - k}{k + 1}$,
    \begin{align*}
        2\left(\frac{n+1}{(\val + a)^2} - \frac{n - k}{\val^2}\right) \leq 
        \left(\frac{n + 1}{\val + a} - \frac{n - k}{\val}\right)^2
    \end{align*}
    After rearranging, it becomes 
    \begin{align*}
        -(k - 1)(k + 1)\val^2 + 2(k - 1)(n - k)a\val - (n - k)(n - k + 2)a^2 \leq 0
    \end{align*}
    where the left hand side is a quadratic function with the maximum value equal to $\frac{-2(n - k)(n + 1)a^2}{k + 1}\leq 0$ attained at $\val = \frac{n - k}{k + 1}a$. This finished the proof of the lemma statement for the special case of $a \geq 0, b = 0$.

    Now consider arbitrary $a \geq 0, b \geq 0$ and distribution $\auxprior(\val) = \frac{\val - b}{\val - b + a}$ with support $\supp(\auxprior) = [b, \infty)$. It suffices to verify that for every $\val\in[b,\infty)$ and $\val\primed \in [\val, \infty)$,
    \begin{align*}
        \val - \frac{1-\auxcdf_{(k)}(\val)}{\auxpdf_{(k)}(\val)}\leq 
        \val\primed -\frac{1-\auxcdf_{(k)}(\val\primed)}{\auxpdf_{(k)}(\val\primed)}
    \end{align*}
    Let $\hat\auxprior_{(k)}$ be the distribution of $k$-th order statistic induced by distribution $\hat\auxprior(\val) = \frac{\val}{\val+a}$. The above inequality can be rewritten as
    \begin{align*}
        \val - b - \frac{1-\hat\auxcdf_{(k)}(\val - b)}{\hat\auxpdf_{(k)}(\val - b)}\leq 
        \val\primed - b - \frac{1-\hat\auxcdf_{(k)}(\val\primed - b)}{\hat\auxpdf_{(k)}(\val\primed -b)}
    \end{align*}
    whose left-hand side and right-hand side are the virtual value of distribution $\hat\auxprior_{(k)}$ at values $\val - b$ and $\val\primed - b$, respectively. As we shown previously, the virtual value of distribution $\hat\auxprior_{(k)}$ is weakly increasing and thus the proof of the lemma statement has been completed as desired.
\end{proof}

\newcommand{\rightderivative}{\partial_+}
\newcommand{\leftderivative}{\partial_-}
\newcommand{\auxrevcurve}{P}

We also need the following two facts for our analysis.

\begin{fact}[Equivalent Condition for Concavity \cite{JLQTX19}]
\label{fact:concavity equivalent definition}
    Suppose $\revcurve(\quant)$ is a left- and right-differentiable function on $[a, b]$, then $\revcurve(\quant)$ is concave on $[a, b]$ if and only if
    \begin{align*}
        \forall \quant\in(a, b),~\quant\primed\in[\quant, b):\qquad
        \revcurve(\quant\primed) \leq \revcurve(\quant) + \rightderivative\revcurve(\quant)\cdot \left(\quant\primed - \quant\right) 
    \end{align*}
    where $\rightderivative\revcurve(\quant)$ is the right derivative of function $\revcurve$ at $\quant$.
\end{fact}

\begin{fact}[\cite{JLQTX19}]
\label{fact:order statistic dominance}
    If distribution $\auxprior$ first-order stochastically dominants distribution $\prior$, then the corresponding $k$-th order statistic distributions $\auxprior_{(k:n)}$ first-order stochastically dominants distribution $\prior_{(k:n)}$ for every $n\in\naturals$ and $k\in[n]$.
\end{fact}

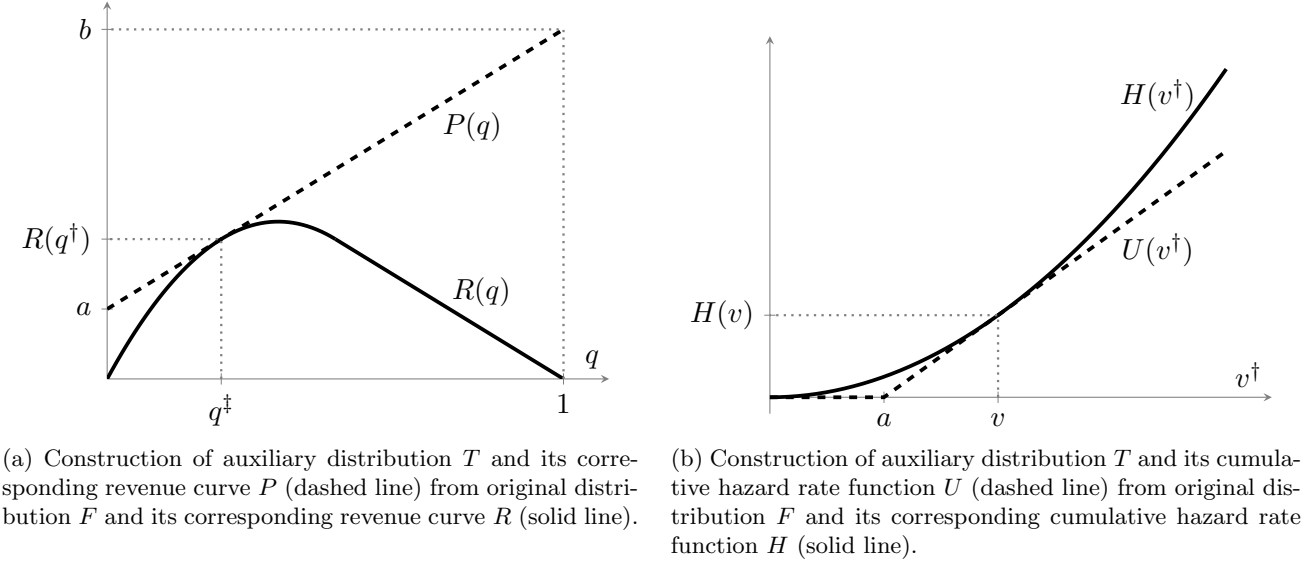
\begin{figure}
    \centering
    \subfloat[Construction of auxiliary distribution $\auxprior$ and its corresponding revenue curve $\auxrevcurve$ (dashed line) from original distribution $\prior$ and its corresponding revenue curve $\revcurve$ (solid line).]{
\input{figure/regular-order-analysis}
\label{fig:regular order analysis}
}
~~~~
    \subfloat[Construction of auxiliary distribution $\auxprior$ and its cumulative hazard rate function $\auxcumhazardrate$ (dashed line) from original distribution $\prior$ and its corresponding cumulative hazard rate function $\cumhazardrate$ (solid line).]{
\input{figure/MHR-order-analysis}
\label{fig:mhr order analysis}
}
    \caption{Graphical illustration of the analysis for \Cref{thm:order:regular} and \Cref{thm:order:mhr}.} 
    \label{fig:regular/mhr order analysis}
\end{figure}

Putting all the pieces together, we are ready to prove \Cref{thm:order:regular}.

\begin{proof}[Proof of \Cref{thm:order:regular}]
Fix an arbitrary $n \geq 1$ and $k \in [n]$. By \Cref{lem:regular equivalent definition}, it suffices to verify the concavity of revenue curve $\revcurve_{(k:n)}$ induced by order statistic distribution $\cdfkn$. With slight abuse of notation, in the remainder of the analysis, we shorthand $\revcurve_{(k:n)}$ and $\cdfkn$ as $\revcurve_{(k)}$ and $\cdf_{(k)}$, respectively.

Invoking \Cref{fact:concavity equivalent definition}, it suffices to verify that for every pair of quantiles $\quant\in(0, 1)$ and $\quant\primed\in[\quant, 1)$,
\begin{align}
\label{eq:order:reqular:revenue curve concavity condition}
    \revcurve_{(k)}(\quant\primed) \leq \revcurve_{(k)}(\quant) + \rightderivative\revcurve_{(k)}(\quant) \cdot 
    \left(\quant\primed - \quant\right)
\end{align}
where $\rightderivative\revcurve_{(k)}(\quant)$ is the right derivative of $\revcurve_{(k)}$ at quantile $\quant$. Let $\val \triangleq \frac{\revcurve_{(k)}(\quant)}{\quant}$, and $\quant\doubleprimed \triangleq 1 - \cdf(\val)$. Define $(a, b)$ as 
\begin{align*}
    a \triangleq \revcurve(\quant\doubleprimed) - \rightderivative\revcurve(\quant\doubleprimed)\cdot \quant\doubleprimed,
    \;\;\mbox{and}\;\;
    b \triangleq \revcurve(\quant\doubleprimed) + \rightderivative\revcurve(\quant\doubleprimed)\cdot \left(1 - \quant\doubleprimed\right)
\end{align*}
Since distribution $\prior$ is regular, revenue curve $\revcurve$ is concave (\Cref{lem:regular equivalent definition}). Invoking \Cref{fact:concavity equivalent definition}, we have $a \geq 0$ and $b \geq 0$. Now consider an auxiliary distribution $\auxprior$ defined as $\auxcdf(\val) = \frac{\val - b}{\val - b + a}$. By construction, auxiliary distribution (i) is regular, (ii) first-order stochastically dominates distribution $\prior$, and (iii) its induced revenue curve $\auxrevcurve$ satisfies $\rightderivative\revcurve(\quant\doubleprimed) = \rightderivative\auxrevcurve(\quant\doubleprimed)$ and $\revcurve(\quant\doubleprimed) = \auxrevcurve(\quant\doubleprimed)$. See \Cref{fig:regular order analysis} for a graphical illustration.

With auxiliary distribution $\auxprior$, order statistic distribution $\auxprior_{(k)}$, and induced revenue curve $\auxrevcurve_{(k)}$, we are ready to verify condition~\eqref{eq:order:reqular:revenue curve concavity condition}:
\begin{align*}
    \revcurve_{(k)}(\quant)
    +
    \rightderivative\revcurve_{(k)}(\quant)\left(\quant\primed - \quant\right)
    \overset{(a)}{=}
    \auxrevcurve_{(k)}(\quant)
    +
    \rightderivative\auxrevcurve_{(k)}(\quant)\left(\quant\primed - \quant\right)
    \overset{(b)}{\geq}
    \auxrevcurve_{(k)}(\quant\primed)
    \overset{(c)}{\geq}
    \revcurve_{(k)}(\quant\primed)
\end{align*}
where equality~(a) holds since $\revcurve(\quant\doubleprimed) = \revcurve(\quant\doubleprimed),\rightderivative\revcurve(\quant\doubleprimed) = \rightderivative\auxrevcurve(\quant\doubleprimed)$ by construction;
inequality~(b) holds since auxiliary order statistic distribution $\auxcdf_{(k)}$ is regular (\Cref{lem:order:regular:triangle}) and \Cref{lem:regular equivalent definition}; and inequality~(c) holds due to \Cref{fact:order statistic dominance} and the construction so that auxiliary distribution $\auxprior$ first-order stochastically dominants distribution $\prior$. This finishes the proof of \Cref{thm:order:regular}.
\end{proof}

\xhdr{Symmetric MHR Distributions.} We next consider MHR distributions.
% \label{subsec:order:mhr}

\begin{theorem}[Order Statistics of Symmetric MHR Distributions]
\label{thm:order:mhr}
\begin{flushleft}
Given $n \geq 1$ many i.i.d.\ MHR values $v_{i} \sim F$ for $i \in [n]$, every order statistic $v_{(k:n)} \sim F_{(k:n)}$ for $k \in [n]$ also is MHR.
% \yf{This theorem is not correct for non-identical MHR distributions. Counterexample can be found at \url{https://www.desmos.com/calculator/1sc6oymwcw}}
\end{flushleft}
\end{theorem}

The proof of \Cref{thm:order:mhr} follows a similar strategy as the one for \Cref{thm:order:regular}. At a high level, our argument relies on the following two-step argument: first, we prove the theorem statement for exponential distributions in \Cref{lem:order:mhr:exponential}. For the exponential distributions, we can explicitly characterize the hazard rate of its order statistic in closed form. Thus, applying an induction argument enables us to prove \Cref{lem:order:mhr:exponential}. In the second step, we present a reduction argument from general distributions to exponential distributions.

\begin{lemma}
\label{lem:order:mhr:exponential}
\begin{flushleft}
    Given $\lambda \geq 0$, $a\geq 0$, and $n \geq 1$ many i.i.d.\ MHR values $\vali\sim \auxprior$ with cumulative density function $\auxcdf(\val) = 
    1 - e^{-\lambda(\val - a)}
    $ and support $\supp(\auxcdf) = [a, \infty)$, every order statistics $\val_{(k:n)} \sim \auxprior_{(k:n)}$ for $k \in [n]$ also is MHR.
\end{flushleft}
\end{lemma}
\begin{proof}
    We prove the lemma statement with a two-step argument. In the first step, we prove the statement for the special case of $\lambda \geq 0$, $a = 0$. In the second step, we show how the statement for general $(\lambda, a)$ can be reduced to the special case argued in step 1. In the remainder of the analysis, with slight abuse of notations, we shorthand $\auxcdf_{(k)}$ for $\auxcdf_{(k:n)}$.

    First, we prove the lemma statement for distribution $\auxprior(\val) = 1-e^{-\lambda \val}$ with support $\supp(\auxprior) = [0, \infty)$. Let $\hazardrate_{(k)}$ be the virtual value function of distribution $\auxprior_{(k)}$. It suffices to show that $\hazardrate_{(k)}(\val)$ is weakly increasing, or equivalently, its derivative $\hazardrate_{(k)}'(\val)$ is non-negative. Note that 
    \begin{align*}
        \hazardrate_{(k)}'(\val) 
        &{}=
        \frac{\left(\auxpdf_{(k)}(\val)\right)^2 + \left(1-\auxcdf_{(k)}(\val)\right)\auxdpdf_{(k)}(\val)}{\left(1 - \auxcdf_{(k)}(\val)\right)^2}
    \end{align*}
    After rearranging, it suffices to show 
    \begin{align*}
        \left(\auxpdf_{(k)}(\val)\right)^2 + \left(1-\auxcdf_{(k)}(\val)\right)\auxdpdf_{(k)}(\val)
        \geq 0
    \end{align*}
    Invoking the fact that $\auxdpdf_{(k)}(\val) = \lambda\auxpdf_{(k)}(\val)\frac{ne^{-\lambda\val}-k}{1 - e^{-\lambda \val}}$ (which holds due to the specific form of $\auxprior(\val) = 1-e^{-\lambda\val}$), the above inequality becomes 
    \begin{align*}
        \auxpdf_{(k)}(\val) \geq \lambda\frac{k-ne^{-\lambda\val}}{1 - e^{-\lambda \val}}
        \left(1-\auxcdf_{(k)}(\val)\right)
    \end{align*}
    which is satisfied trivially when $\val \leq \frac{1}{\lambda}\ln(\frac{n}{k})$. 
    
    To verify the inequality for $\val > \frac{1}{\lambda}\ln(\frac{n}{k})$, we define auxiliary function $\auxfunc(\val) \triangleq \frac{1}{\lambda}\auxpdf_{(k)}(\val)\frac{1 - e^{-\lambda \val}}{k-ne^{-\lambda\val}} - 1+\auxcdf_{(k)}(\val)$. Since $\lim_{\val\rightarrow\infty}\auxfunc(\val) = 0$, it suffices to show that $\auxfunc'(\val) \leq 0$ for all $\val > \frac{1}{\lambda}\ln(\frac{n}{k})$. Note that 
    \begin{align*}
        \auxfunc'(\val) 
        &{} = \frac{1}{\lambda}\frac{1 - e^{-\lambda \val}}{k-ne^{-\lambda\val}}\auxdpdf_{(k)}(\val)
        +\frac{e^{-\lambda\val}\left(k-ne^{-\lambda\val}\right) - \left(1 - e^{-\lambda \val}\right)ne^{-\lambda\val}}{\left(k-ne^{-\lambda\val}\right)^2}\auxpdf_{(k)}(\val)
        +\auxpdf(\val)
        \\
        &{}=
        \frac{e^{-\lambda\val}\left(k-ne^{-\lambda\val}\right) - \left(1 - e^{-\lambda \val}\right)ne^{-\lambda\val}}{\left(k-ne^{-\lambda\val}\right)^2}\auxpdf_{(k)}(\val)
        =
        -\frac{(n-k)e^{-\lambda\val}}{\left(k-ne^{-\lambda\val}\right)^2}\auxpdf_{(k)}(\val)\leq 0
    \end{align*}
    where the first, third equalities holds by algebra and the second equality holds since $\auxdpdf_{(k)}(\val) = \lambda\auxpdf_{(k)}(\val)\frac{ne^{-\lambda\val}-k}{1 - e^{-\lambda \val}}$. This finished the proof of the lemma statement for the special case of $\lambda \geq 0, a = 0$.

    Now consider arbitrary $\lambda \geq 0, a \geq 0$ and distribution $\auxcdf(\val) = 
    1 - e^{-\lambda(\val - a)}
    $ with support $\supp(\auxcdf) = [a, \infty)$. It suffices to verify that for every $\val\in[b,\infty)$ and $\val\primed \in [\val, \infty)$,
    \begin{align*}
        \frac{\auxpdf_{(k)}(\val)}{1-\auxcdf_{(k)}(\val)}\leq 
        \frac{\auxpdf_{(k)}(\val\primed)}{1-\auxcdf_{(k)}(\val\primed)}
    \end{align*}
    Let $\hat\auxprior_{(k)}$ be the distribution of $k$-th order statistic induced by distribution $\hat\auxprior(\val) = 1-e^{-\lambda\val}$. The above inequality can be rewritten as
    \begin{align*}
        \frac{\hat\auxpdf_{(k)}(\val - a)}{1-\hat\auxcdf_{(k)}(\val - a)}\leq 
        \frac{\hat\auxpdf_{(k)}(\val\primed - a)}{1-\hat\auxcdf_{(k)}(\val\primed - a)}
    \end{align*}
    whose left-hand side and right-hand side are the hazard rate of distribution $\hat\auxprior_{(k)}$ at values $\val - a$ and $\val\primed - a$, respectively. As we shown previously, the hazard rate of distribution $\hat\auxprior_{(k)}$ is weakly increasing and thus the proof of the lemma statement has been completed as desired.
\end{proof}

\begin{proof}[Proof of \Cref{thm:order:mhr}]
Fix an arbitrary $n \geq 1$ and $k \in [n]$. By \Cref{lem:mhr equivalent definition}, it suffices to verify the convexity of the cumulative hazard rate function $\cumhazardrate_{(k:n)}(\val) \triangleq-\ln(1-\cdf_{(k,n)}(\val))$ induced by order statistic distribution $\cdfkn$. With slight abuse of notation, in the remainder of the analysis, we shorthand $\cumhazardrate_{(k:n)}$ and $\cdfkn$ as $\cumhazardrate_{(k)}$ and $\cdf_{(k)}$, respectively.

Invoking \Cref{fact:concavity equivalent definition}, it suffices to verify that for every pair of values $\val\in\supp(\cdf_{(k)})$ and $\val\primed\in[\val, \infty)$,
\begin{align}
\label{eq:order:mhr:cumulative hazard rate convex condition}
    \cumhazardrate_{(k)}(\val\primed) \geq \cumhazardrate_{(k)}(\val) + \rightderivative\cumhazardrate_{(k)}(\val) \cdot 
    \left(\val\primed - \val\right)
\end{align}
where $\rightderivative\cumhazardrate_{(k)}(\val)$ is the right derivative of $\cumhazardrate_{(k)}$ at value $\val$. Define $(\lambda, a)$ as 
\begin{align*}
    \lambda \triangleq \rightderivative\cumhazardrate(\val),
    \;\;\mbox{and}\;\;
    a \triangleq \frac{\rightderivative\cumhazardrate(\val)\cdot \val - \cumhazardrate(\val)}{\rightderivative\cumhazardrate(\val)}
\end{align*}
Since distribution $\prior$ is MHR, cumulative hazard rate curve $\cumhazardrate$ is convex (\Cref{lem:mhr equivalent definition}). Invoking \Cref{fact:concavity equivalent definition}, we have $\lambda \geq 0$ and $a \geq 0$. Now consider an auxiliary distribution $\auxprior$ defined as $\auxcdf(\val) = 1-e^{-\lambda(\val - a)}$. By construction, auxiliary distribution $\auxprior$ (i) is MHR, (ii) first-order stochastically dominates distribution $\prior$, and (iii) its induced cumulative hazard rate curve $\auxcumhazardrate$ satisfies $\rightderivative\cumhazardrate(\val) = \rightderivative\auxcumhazardrate(\val)$ and $\cumhazardrate(\val) = \auxcumhazardrate(\val)$. See \Cref{fig:mhr order analysis} for a graphical illustration.

With auxiliary distribution $\auxprior$, order statistic distribution $\auxprior_{(k)}$, and induced cumulative hazard rate curve $\auxcumhazardrate_{(k)}$, we are ready to verify condition~\eqref{eq:order:mhr:cumulative hazard rate convex condition}:
\begin{align*}
    \cumhazardrate_{(k)}(\val)
    +
    \rightderivative\cumhazardrate_{(k)}(\val)\left(\val\primed - \val\right)
    \overset{(a)}{=}
    \auxcumhazardrate_{(k)}(\val)
    +
    \rightderivative\auxcumhazardrate_{(k)}(\val)\left(\val\primed - \val\right)
    \overset{(b)}{\leq}
    \auxcumhazardrate_{(k)}(\val\primed)
    \overset{(c)}{\leq}
    \cumhazardrate_{(k)}(\val\primed)
\end{align*}
where equality~(a) holds since $\rightderivative\cumhazardrate(\val) = \rightderivative\auxcumhazardrate(\val),\cumhazardrate(\val) = \auxcumhazardrate(\val)$ by construction;
inequality~(b) holds since auxiliary order statistic distribution $\auxcdf_{(k)}$ is MHR (\Cref{lem:order:mhr:exponential}) and \Cref{lem:mhr equivalent definition}; and inequality~(c) holds due to \Cref{fact:order statistic dominance} and the construction so that auxiliary distribution $\auxprior$ first-order stochastically dominants distribution $\prior$. This finishes the proof of \Cref{thm:order:mhr}.
\end{proof}

\subsection{Order Statistics of Asymmetric Quasi-Regular/-MHR Distributions}
\label{sec:hierarchy:asymmetric}

In this subsection, we characterize the order statistic distributions generated from asymmetric quasi-regular or quasi-MHR distributions.

\xhdr{Asymmetric Quasi-Regular distributions} We first consider quasi-regular distributions.
% \label{subsec:order:quasi-regular}

\begin{theorem}[Order Statistics of Asymmetric Quasi-Regular Distributions]
\label{thm:order:quasi-regular}
\begin{flushleft}
Given $n \geq 1$ many independent (but possibly asymmetric) quasi-regular values $\val_{i} \sim \priori$ for $i \in [n]$, every order statistic $\valkn \sim \cdfkn$ for $k \in [n]$ also is quasi-regular.
\end{flushleft}
\end{theorem}

\begin{proof}
We prove the theorem statement by an induction argument over $n\in\naturals$. In the base case where $n = 1$, the theorem statement hold trivially. Now consider an arbitrary $n \geq 2$, and $k \in [n]$. By \Cref{prop:q-regular equivalent definition}, it suffices to verify Condition~\ref{condition:quasi-regular:cdf} in the proposition statement, i.e., for every value $\val\in \supp(\cdfkn)$,
\begin{align}
\label{eq:order:quasi-regular critical inequality}
    \forall \val\primed \in [0, \val):\qquad \cdfkn(\val\primed) \leq \frac{\val\primed}{\val\primed + \frac{1 - \cdfkn(\val)}{\cdfkn(\val)}\val}
\end{align}
We prove inequality~\eqref{eq:order:quasi-regular critical inequality} for $k = 1$ and $k \geq 2$ separately. 

\xhdr{Case (a): $k = 1$.}
Recall the fact that 
\begin{align*}
    \cdf_{(1:n)}(\val) = \cdf_{(1:n-1)}(\val) \cdot \cdf_n(\val)
\end{align*}
Invoking the induction hypothesis for $n - 1$ and the assumption in the theorem statement, we know that both distributions $\cdf_{(1:n-1)}$ and $\cdf_n$ are quasi-regular. In the remainder of the analysis of Case~(a), with slight abuse of notation, we shorthand $\cdf_{(k:n)}$, $\cdf_{n}$, $\cdf_{(k:n)}\primed$, $\cdf_{n}\primed$ for $\cdf_{(k:n)}(\val)$, $\cdf_{n}(\val)$, $\cdf_{(k:n)}(\val\primed)$, $\cdf_{n}(\val\primed)$, respectively. Invoking Condition~\ref{condition:quasi-regular:cdf} of \Cref{prop:q-regular equivalent definition}, we obtain
\begin{align*}
    \cdf_{(1:n-1)}\primed
    \leq
    \frac{\val\primed}{\val\primed + \frac{1 - \cdf_{(1:n-1)}}{\cdf_{(1:n-1)}}\val}~,
    ~~~~
    \cdf_{n}\primed
    \leq 
    \frac{\val\primed}{\val\primed + \frac{1 - \cdf_{n}}{\cdf_{n}}\val}
\end{align*}
and thus
\begin{align*}
    \cdf_{(1:n)}\primed \leq \frac{\val\primed}{\val\primed + \frac{1 - \cdf_{(1:n-1)}}{\cdf_{(1:n-1)}}\val} \frac{\val\primed}{\val\primed + \frac{1 - \cdf_{n}}{\cdf_{n}}\val}
\end{align*}
To prove inequality~\eqref{eq:order:quasi-regular critical inequality}, it suffices to verify that for every $\val \in \supp(\cdf_{(1:n)})$ and every $\val\primed\in[0,\val)$:
\begin{align*}
    \frac{\val\primed}{\val\primed + \frac{1 - \cdf_{(1:n-1)}}{\cdf_{(1:n-1)}}\val} \frac{\val\primed}{\val\primed + \frac{1 - \cdf_{n}}{\cdf_{n}}\val} \leq \frac{\val\primed}{\val\primed + \frac{1 - \cdf_{(1:n)}}{\cdf_{(1:n)}}\val}
\end{align*}
After rearranging, it is equivalent to 
\begin{align*}
    \val\primed\left(
    \cdf_{n}\cdf_{(1:n-1)} - 
    \cdf_{n}\cdf_{(1:n)} -
    \cdf_{(1:n-1)}\cdf_{(1:n)} +
    \cdf_{n}\cdf_{(1:n-1)}\cdf_{(1:n)}
    \right)
    \leq 
    \val \left(1 - \cdf_{(1:n-1)}\right)\left(1 - \cdf_{n}\right) \cdf_{(1:n)}
\end{align*}
Invoking the fact that $\cdf_{n}(\val)\cdf_{(1:n-1)}(\val) = \cdf_{(1:n)}(\val)$, we obtain
\begin{align*}
    \val\primed \left(1 - \cdf_{(1:n-1)}\right)\left(1 - \cdf_{n}\right)
    \leq \val \left(1 - \cdf_{(1:n-1)}\right)\left(1 - \cdf_{n}\right)
\end{align*}
which is satisfied since $\val\primed \leq \val$.

\xhdr{Case (b): $k \geq 2$.} The analysis is similar to the previous case.
Recall the fact that 
\begin{align*}
    \cdf_{(k:n)}(\val) = \cdf_{(k:n-1)}(\val) \cdot \cdf_n(\val) + \cdf_{(k-1:n-1)}(\val)(1 - \cdf_n(\val))
\end{align*}
Invoking the induction hypothesis for $n - 1$ and the assumption in the theorem statement, we know that distributions $\cdf_{(k:n-1)}$, $\cdf_{(k - 1: n - 1)}$, and $\cdf_{n}$ are all quasi-regular. In the remainder of the analysis of Case~(b), with slight abuse of notation, we shorthand $\cdf_{(k:n)}$, $\cdf_{n}$, $\cdf_{(k:n)}\primed$, $\cdf_{n}\primed$ for $\cdf_{(k:n)}(\val)$, $\cdf_{n}(\val)$, $\cdf_{(k:n)}(\val\primed)$, $\cdf_{n}(\val\primed)$, respectively. Invoking Condition~\ref{condition:quasi-regular:cdf} of \Cref{prop:q-regular equivalent definition}, we obtain
\begin{align*}
    \cdf_{(k - 1:n-1)}\primed
    \leq
    \frac{\val\primed}{\val\primed + \frac{1 - \cdf_{(k-1:n-1)}}{\cdf_{(k-1:n-1)}}\val}~,
    ~
    \cdf_{(k:n-1)}\primed
    \leq
    \frac{\val\primed}{\val\primed + \frac{1 - \cdf_{(k:n-1)}}{\cdf_{(k:n-1)}}\val}~,
    ~
    \cdf_{n}\primed
    \leq 
    \frac{\val\primed}{\val\primed + \frac{1 - \cdf_{n}}{\cdf_{n}}\val}
\end{align*}
Combining with the fact that $\cdf_{(k:n-1)}(\val) \geq \cdf_{(k-1:n-1)}(\val)$, we obtain
\begin{align*}
    \cdf_{(1:n)}\primed\leq 
    \frac{\val\primed}{\val\primed + \frac{1 - \cdf_{(k:n-1)}}{\cdf_{(k:n-1)}}\val}
    \frac{\val\primed}{\val\primed + \frac{1 - \cdf_{n}}{\cdf_{n}}\val}
    +
    \frac{\val\primed}{\val\primed + \frac{1 - \cdf_{(k-1:n-1)}}{\cdf_{(k-1:n-1)}}\val}
    \frac{\frac{1 - \cdf_{n}}{\cdf_{n}}\val}{\val\primed + \frac{1 - \cdf_{n}}{\cdf_{n}}\val}
\end{align*}
To prove inequality~\eqref{eq:order:quasi-regular critical inequality}, it suffices to verify that for every $\val \in \supp(\cdf_{(1:n)})$ and every $\val\primed\in[0,\val)$:
\begin{align*}
    \frac{\val\primed}{\val\primed + \frac{1 - \cdf_{(k:n-1)}}{\cdf_{(k:n-1)}}\val}
    \frac{\val\primed}{\val\primed + \frac{1 - \cdf_{n}}{\cdf_{n}}\val}
    +
    \frac{\val\primed}{\val\primed + \frac{1 - \cdf_{(k-1:n-1)}}{\cdf_{(k-1:n-1)}}\val}
    \frac{\frac{1 - \cdf_{n}}{\cdf_{n}}\val}{\val\primed + \frac{1 - \cdf_{n}}{\cdf_{n}}\val}
    \leq 
    \frac{\val\primed}{\val\primed + \frac{1 - \cdf_{(k:n)}}{\cdf_{(k:n)}}\val}
\end{align*}
After rearranging, it is equivalent to 
\begin{align*}
    &\left(\val\primed\right)^2\left(\frac{1 - \cdf_{(k:n)}}{\cdf_{(k:n)}}
    -\frac{1 - \cdf_{(k:n-1)}}{\cdf_{(k:n-1)}}\right)
    +
    \val^2
    \frac{1 - \cdf_{(k:n-1)}}{\cdf_{(k:n-1)}}\frac{1 - \cdf_{n}}{\cdf_{n}}
    \left(\frac{1 - \cdf_{(k:n)}}{\cdf_{(k:n)}} - 
    \frac{1 - \cdf_{(k-1:n-1)}}{\cdf_{(k-1:n-1)}}\right)
    \\
    &\hspace{10pt}+
    \val\val\primed 
    \left(
    \frac{1 - \cdf_{(k-1:n-1)}}{\cdf_{(k-1:n-1)}}
    \left(
    \frac{1 - \cdf_{(k:n)}}{\cdf_{(k:n)}}
    -\frac{1 - \cdf_{(k:n-1)}}{\cdf_{(k:n-1)}}
    -\frac{1 - \cdf_{n}}{\cdf_{n}}
    \right)
    +
    \frac{1 - \cdf_{(k:n)}}{\cdf_{(k:n)}}\frac{1 - \cdf_{n}}{\cdf_{n}}
    \right)
    \leq 0
\end{align*}
whose left-hand side is a quadratic function of $\val\primed$. Invoking the fact that $\cdf_{(k:n-1)}(\val) \geq \cdf_{(k:n)}(\val)$, it suffices to verify the inequality at $\val\primed = 0$ and $\val\primed= \val$. Specifically, for $\val\primed = 0$, the above inequality becomes 
\begin{align*}
    \val^2
    \frac{1 - \cdf_{(k:n-1)}}{\cdf_{(k:n-1)}}\frac{1 - \cdf_{n}}{\cdf_{n}}
    \left(\frac{1 - \cdf_{(k:n)}}{\cdf_{(k:n)}} - 
    \frac{1 - \cdf_{(k-1:n-1)}}{\cdf_{(k-1:n-1)}}\right)
    \leq 0
\end{align*}
which is satisfied due to the fact that $\cdf_{(k:n)}(\val) \geq \cdf_{(k-1:n-1)}(\val)$. On the other hand, for $\val\primed = \val$, after rearranging, the above inequality becomes 
\begin{align*}
    \cdf_n\cdf_{(k:n-1)}
    +
    (1-\cdf_n)\cdf_{(k-1:n-1)}
    -
    \cdf_{(k:n)}\leq 0
\end{align*}
whose left-hand side is equal to zero for all values $\val\in\supp(\cdfkn)$.

\smallskip\noindent
Putting all the pieces together, we finish the inductive step for given $n$, and thus \Cref{thm:order:quasi-regular} is proven by induction as desired.
\end{proof}

\xhdr{Asymmetric Quasi-MHR distributions} We next consider quasi-MHR distributions.
% \label{subsec:order:quasi-mhr}

\begin{theorem}[Order Statistics of Asymmetric Quasi-MHR Distributions]
\label{thm:order:quasi-mhr}
\begin{flushleft}
Given $n \geq 1$ many independent (but possibly asymmetric) quasi-MHR values $v_{i} \sim F_{i}$ for $i \in [n]$, every order statistic $v_{(k:n)} \sim F_{(k:n)}$ for $k \in [n]$ also is quasi-MHR.
\end{flushleft}
\end{theorem}
\begin{proof}
We prove the theorem statement by an induction argument over $n\in\naturals$. In the base case where $n = 1$, the theorem statement hold trivially. Now consider an arbitrary $n \geq 2$, and $k \in [n]$. By \Cref{prop:q-MHR equivalent definition}, it suffices to verify Condition~\ref{condition:quasi-MHR:cdf} in the proposition statement, i.e., for every value $\val\in \supp(\cdfkn)$,
\begin{align}
\label{eq:order:quasi-MHR critical inequality}
    \forall \val\primed \in [0, \val):\qquad \cdfkn(\val\primed) \leq 1 - \left(1 - \cdf(\val)\right)^{\frac{\val\primed}{\val}}
\end{align}
We prove inequality~\eqref{eq:order:quasi-MHR critical inequality} for $k = 1$ and $k \geq 2$ separately. 

\xhdr{Case (a): $k = 1$.}
Recall the fact that 
\begin{align*}
    \cdf_{(1:n)}(\val) = \cdf_{(1:n-1)}(\val) \cdot \cdf_n(\val)
\end{align*}
Invoking the induction hypothesis for $n - 1$ and the assumption in the theorem statement, we know that both distributions $\cdf_{(1:n-1)}$ and $\cdf_n$ are quasi-MHR. In the remainder of the analysis of Case~(a), with slight abuse of notation, we shorthand $\cdf_{(k:n)}$, $\cdf_{n}$, $\cdf_{(k:n)}\primed$, $\cdf_{n}\primed$ for $\cdf_{(k:n)}(\val)$, $\cdf_{n}(\val)$, $\cdf_{(k:n)}(\val\primed)$, $\cdf_{n}(\val\primed)$, respectively. Invoking Condition~\ref{condition:quasi-MHR:cdf} of \Cref{prop:q-MHR equivalent definition}, we obtain
\begin{align*}
    \cdf_{(1:n-1)}\primed
    \leq
    1 - \left(1 - \cdf_{(1:n-1)}\right)^{\frac{\val\primed}{\val}}~,
    ~~~~
    \cdf_{n}\primed
    \leq 
    1 - \left(1 - \cdf_{n}\right)^{\frac{\val\primed}{\val}}
\end{align*}
and thus
\begin{align*}
    \cdf_{(1:n)}\primed &\leq 
    \left(1 - \left(1 - \cdf_{(1:n-1)}\right)^{\frac{\val\primed}{\val}}
    \right)
    \left(
    1 - \left(1 - \cdf_{n}\right)^{\frac{\val\primed}{\val}}
    \right)
    \\
    &{}=
    1 - \left(1 - \cdf_{(1:n-1)}\right)^{\frac{\val\primed}{\val}}
    - \left(1 - \cdf_{n}\right)^{\frac{\val\primed}{\val}}
    +
    \left(\left(1 - \cdf_{(1:n-1)}\right)
    \left(1 - \cdf_{n}\right)\right)^{\frac{\val\primed}{\val}}
\end{align*}
To prove inequality~\eqref{eq:order:quasi-MHR critical inequality}, it suffices to verify that for every $\val \in \supp(\cdf_{(1:n)})$ and every $\val\primed\in[0,\val)$:
\begin{align*}
    1 - \left(1 - \cdf_{(1:n-1)}\right)^{\frac{\val\primed}{\val}}
    - \left(1 - \cdf_{n}\right)^{\frac{\val\primed}{\val}}
    +
    \left(\left(1 - \cdf_{(1:n-1)}\right)
    \left(1 - \cdf_{n}\right)\right)^{\frac{\val\primed}{\val}}
    \leq
    1 - \left(1 - \cdf_{(1:n)}\right)^{\frac{\val\primed}{\val}}
\end{align*}
Invoking the fact that  $\cdf_{n}(\val)\cdf_{(1:n-1)}(\val) = \cdf_{(1:n)}(\val)$ and rearranging, we obtain 
\begin{align*}
    \left(\left(1 - \cdf_{(1:n-1)}\right)
    \left(1 - \cdf_{n}\right)\right)^{\frac{\val\primed}{\val}}
    +
    \left(1 -  \cdf_{(1:n-1)}\cdf_n\right)^{\frac{\val\primed}{\val}}
    \leq 
    \left(1 - \cdf_{(1:n-1)}\right)^{\frac{\val\primed}{\val}}
    + \left(1 - \cdf_{n}\right)^{\frac{\val\primed}{\val}}
\end{align*}
which is satisfied since inequality $(AB)^\gamma + (A + B - AB)^\gamma - A^\gamma + B^\gamma \leq 0$ holds for all $A, B,\gamma\in[0, 1]$ and we let $A = 1 - \cdf_{(1:n-1)}$, $B = 1 - \cdf_n$, $\gamma = \frac{\val\primed}{\val}$.

\xhdr{Case (b): $k \geq 2$.} The analysis is similar to the previous case.
Recall the fact that 
\begin{align*}
    \cdf_{(k:n)}(\val) = \cdf_{(k:n-1)}(\val) \cdot \cdf_n(\val) + \cdf_{(k-1:n-1)}(\val)(1 - \cdf_n(\val))
\end{align*}
Invoking the induction hypothesis for $n - 1$ and the assumption in the theorem statement, we know that distributions $\cdf_{(k:n-1)}$, $\cdf_{(k - 1: n - 1)}$, and $\cdf_{n}$ are all quasi-MHR. In the remainder of the analysis of Case~(b), with slight abuse of notation, we shorthand $\cdf_{(k:n)}$, $\cdf_{n}$, $\cdf_{(k:n)}\primed$, $\cdf_{n}\primed$ for $\cdf_{(k:n)}(\val)$, $\cdf_{n}(\val)$, $\cdf_{(k:n)}(\val\primed)$, $\cdf_{n}(\val\primed)$, respectively. Invoking Condition~\ref{condition:quasi-MHR:cdf} of \Cref{prop:q-MHR equivalent definition}, we obtain
\begin{align*}
    \cdf_{(k - 1:n-1)}\primed
    \leq
    1 - \left(1 - \cdf_{(k - 1:n-1)}\right)^{\frac{\val\primed}{\val}}~,
    ~
    \cdf_{(k:n-1)}\primed
    \leq
    1 - \left(1 - \cdf_{(k:n-1)}\right)^{\frac{\val\primed}{\val}}~,
    ~
    \cdf_{n}\primed
    \leq 
    1 - \left(1 - \cdf_{n}\right)^{\frac{\val\primed}{\val}}
\end{align*}
Combining with the fact that $\cdf_{(k:n-1)}(\val) \geq \cdf_{(k-1:n-1)}(\val)$, we obtain
\begin{align*}
    \cdf_{(1:n)}\primed\leq 
    \left(
    1 - \left(1 - \cdf_{(k:n-1)}\right)^{\frac{\val\primed}{\val}}
    \right)
    \left(
    1 - \left(1 - \cdf_{n}\right)^{\frac{\val\primed}{\val}}
    \right)
    +
    \left(
    1 - \left(1 - \cdf_{(k - 1:n-1)}\right)^{\frac{\val\primed}{\val}}
    \right)
    \left(1 - \cdf_{n}\right)^{\frac{\val\primed}{\val}}
\end{align*}
To prove inequality~\eqref{eq:order:quasi-MHR critical inequality}, it suffices to verify that for every $\val \in \supp(\cdf_{(1:n)})$ and every $\val\primed\in[0,\val)$:
\begin{align*}
    \left(
    1 - \left(1 - \cdf_{(k:n-1)}\right)^{\frac{\val\primed}{\val}}
    \right)
    \left(
    1 - \left(1 - \cdf_{n}\right)^{\frac{\val\primed}{\val}}
    \right)
    +
    \left( 
    1 - \left(1 - \cdf_{(k - 1:n-1)}\right)^{\frac{\val\primed}{\val}}
    \right)
    \left(1 - \cdf_{n}\right)^{\frac{\val\primed}{\val}}
    \leq 
    1 - \left(1 - \cdf_{(k:n)}\right)^{\frac{\val\primed}{\val}}
\end{align*}
Invoking the fact that  $\cdf_{(k:n)}(\val) = \cdf_{(k:n-1)}(\val) \cdot \cdf_n(\val) + \cdf_{(k-1:n-1)}(\val)(1 - \cdf_n(\val))$ and rearranging, we obtain 
\begin{align*}
    &\left(
    (1 - \cdf_{(k:n-1)})
    (1 - \cdf_n)
    \right)^{\frac{\val\primed}{\val}}
    +
    \left(1 - \cdf_{(k:n-1)}\cdf_n - \cdf_{(k-1:n-1)} + \cdf_{(k-1:n-1)}\cdf_n\right)^{\frac{\val\primed}{\val}}
    \\ 
    &\hspace{180pt}-
    (1 - \cdf_{(k:n-1)})^{\frac{\val\primed}{\val}}
    -
    \left(
    (1 - \cdf_{(k-1:n-1)})
    (1 - \cdf_n)
    \right)^{\frac{\val\primed}{\val}}
    \leq 0
\end{align*}
Note the left-hand side is decreasing in $\cdf_{(k-1:n-1)}$, it suffices to verify the inequality at $\cdf_{(k-1:n-1)} = 0$, i.e., 
\begin{align*}
    \left(\left(1 - \cdf_{(k:n-1)}\right)
    \left(1 - \cdf_{n}\right)\right)^{\frac{\val\primed}{\val}}
    +
    \left(1 -  \cdf_{(k:n-1)}\cdf_n\right)^{\frac{\val\primed}{\val}}
    \leq 
    \left(1 - \cdf_{(k:n-1)}\right)^{\frac{\val\primed}{\val}}
    + \left(1 - \cdf_{n}\right)^{\frac{\val\primed}{\val}}
\end{align*}
which is satisfied since inequality $(AB)^\gamma + (A + B - AB)^\gamma - A^\gamma + B^\gamma \leq 0$ holds for all $A, B,\gamma\in[0, 1]$ and we let $A = 1 - \cdf_{(k:n-1)}$, $B = 1 - \cdf_n$, $\gamma = \frac{\val\primed}{\val}$.

\smallskip\noindent
Putting all the pieces together, we finish the inductive step for given $n$, and thus \Cref{thm:order:quasi-mhr} is proven by induction as desired.
\end{proof}

%% file: figure/quasi-regular-geometric-definition.tex
\begin{tikzpicture}[scale=1, transform shape]
\begin{axis}[
axis line style=gray,
axis lines=middle,
xlabel = $\quant\primed$,
xtick={0,  0.5, 1},
ytick={0, 0.75},
xticklabels={0, $\quant$, 1},
yticklabels={0, $\revcurve(\quant)$},
xmin=0,xmax=1.1,ymin=-0.0,ymax=1,
width=0.5\textwidth,
height=0.4\textwidth,
samples=50]

\addplot[domain=0:0.5, black!100!white, line width=0.5mm] (x, {4 * x - 5*x*x});

\addplot[domain=0.5:0.625, black!100!white, line width=0.5mm] (x, {-x + 1.25});

\addplot[domain=0.625:1, black!100!white, line width=0.5mm] (x, {x});

\addplot[dashed, line width=0.5mm] (0.5, 0.75) -- (1, 0);

\addplot[dotted, gray, line width=0.3mm] (0.5, 0.75) -- (0.5, 0);

\addplot[dotted, gray, line width=0.3mm] (0.5, 0.75) -- (0, 0.75);

\draw (0.82, 0.5) node {$\frac{1-\quant\primed}{1 - \quant}\revcurve(\quant)$};

\draw (0.78, 0.9) node {$\revcurve(\quant\primed)$};

\end{axis}

\end{tikzpicture}

%% file: figure/quasi-mhr-geometric-definition.tex
\begin{tikzpicture}[scale=1, transform shape]
\begin{axis}[
axis line style=gray,
axis lines=middle,
xlabel = $\val\primed$,
xtick={0, 2},
ytick={0, 0.69315},
xticklabels={0, $\val$},
yticklabels={0, $\cumhazardrate(\val)$},
xmin=0,xmax=3.5,ymin=-0.01,ymax=1.3,
width=0.5\textwidth,
height=0.4\textwidth,
samples=50]

\addplot[domain=0:1, black!100!white, line width=0.5mm] (x, {0});

\addplot[domain=1:2.7182818, black!100!white, line width=0.5mm] (x, {ln(x)});

\addplot[domain=2.7182818:3.5, black!100!white, line width=0.5mm] (x, {x/exp(1)});

\addplot[dashed, line width=0.5mm] (2, 0.69315) -- (0, 0);

\addplot[dotted, gray, line width=0.3mm] (2, 0.69315) -- (2, 0);

\addplot[dotted, gray, line width=0.3mm] (2, 0.69315) -- (0, 0.69315);

\draw (0.82, 0.47) node {$\frac{\val\primed}{\val}\cumhazardrate(\val)$};

\draw (3, 0.92) node {$\cumhazardrate(\val\primed)$};

\end{axis}

\end{tikzpicture}

%% file: figure/regular-order-analysis.tex
\begin{tikzpicture}[scale=1, transform shape]
\begin{axis}[
axis line style=gray,
axis lines=middle,
xlabel = $\quant$,
xtick={0,  0.25, 1},
ytick={0, 0.25, 0.5, 1.25},
xticklabels={0, $\quant\doubleprimed$, 1},
yticklabels={0, $a$, $\revcurve(\quant\primed)$, $b$},
xmin=0,xmax=1.1,ymin=-0.0,ymax=1.35,
width=0.5\textwidth,
height=0.4\textwidth,
samples=50]

\addplot[domain=0:0.5, black!100!white, line width=0.5mm] (x, {3 * x - 4*x*x});
\addplot[line width=0.5mm] (0.5, 0.5) -- (1, 0);

\addplot[line width=0.5mm, dashed] (0, 0.25) -- (1, 1.25);

\addplot[dotted, gray, line width=0.3mm] (0.25, 0) -- (0.25, 0.5) -- (0, 0.5);

\addplot[dotted, gray, line width=0.3mm] (1, 0) -- (1, 1.25) -- (0, 1.25);

\draw (0.82, 0.31) node {$\revcurve(\quant)$};

\draw (0.8, 0.9) node {$\auxrevcurve(\quant)$};

\end{axis}

\end{tikzpicture}

%% file: figure/MHR-order-analysis.tex
\begin{tikzpicture}[scale=1, transform shape]
\begin{axis}[
axis line style=gray,
axis lines=middle,
xlabel = $\val\primed$,
xtick={0,  0.25, 0.5},
ytick={0, 0.25},
xticklabels={0, $a$, $\val$},
yticklabels={0, $\cumhazardrate(\val)$},
xmin=0,xmax=1.1,ymin=-0.05,ymax=1.1,
width=0.5\textwidth,
height=0.4\textwidth,
samples=50]

\addplot[domain=0:1, black!100!white, line width=0.5mm] (x, {x * x});

\addplot[line width=0.5mm, dashed] (0, 0) -- (0.25, 0) -- (1, 0.75);

\addplot[dotted, gray, line width=0.3mm] (0.5, 0) -- (0.5, 0.25) -- (0, 0.25);

\draw (0.85, 0.45) node {$\auxcumhazardrate(\val\primed)$};

\draw (0.85, 0.92) node {$\cumhazardrate(\val\primed)$};

\end{axis}

\end{tikzpicture}

%% file: source/simple_mechanism.tex
\section{Revenue Approximation by Bayesian Simple Mechanisms}
\label{sec:simple-mechanism}

% \afterpage{
\input{table/table_simple_mechanism}
% \clearpage}

In this section, we revisit the celebrated literature about revenue approximations by Bayesian simple mechanisms. In \Cref{sec:simple-mechanism:single-item}, we study the revenue approximation of {\BayesianOptimalUniformPricing} and {\BayesianOptimalUniformReserve} for quasi-regular buyers in the single-item setting. In \Cref{sec:simple-mechanism:downward-closed}, we study the revenue approximation of {\BayesianMonopolyReserves} for quasi-MHR buyers in the downward-closed settings.

\subsection{The Single-Item Setting}
\label{sec:simple-mechanism:single-item}

In this subsection, we establish the revenue approximation (aka., revenue gap) between simple and non-discrimination mechanisms -- {\BayesianOptimalUniformReserve} and {\BayesianOptimalUniformPricing} -- with more complicated, discriminatory but revenue-(approximately)-optimal {\BayesianOptimalMechanism} ({\BayesianOptimalSequentialPricing}) in the single-item setting. We impose our focus on quasi-regular distributions for both symmetric and asymmetric buyers. All results in this subsection are summarized in \Cref{tab:simple-mechanism:BOUP,tab:simple-mechanism:BOUR}.

\xhdr{Asymmetric Buyers.}
We first consider asymmetric buyers with independent (but not necessarily identical) quasi-regular distributions. Recall that for asymmetric buyers, {\BayesianOptimalUniformReserve} and {\BayesianOptimalUniformPricing} admit no bounded revenue approximation under general distributions, while constant revenue approximations exist under regular distributions. For quasi-regular distributions, we establish constant revenue approximations (\Cref{thm:BOSP_BOUP,thm:BOM_BOUR,thm:BOM_BOUP}). Our results suggest that in term of approximation guarantees, the quasi-regularity is closer to regularity when buyers are asymmetric. 

To compare {\BayesianOptimalSequentialPricing} and {\BayesianOptimalUniformPricing},  \Cref{thm:BOSP_BOUP} gives a (nontrivial) {\em approximation-preserving} generalization of \cite[Theorem~1]{JLTX20} from regular buyers to quasi-regular buyers.

\begin{theorem}[{\BOSP} vs.\ {\BOUP}]
\label{thm:BOSP_BOUP}
\begin{flushleft}
% Given $n \geq 1$ many independent (but possibly asymmetric) quasi-regular or regular buyers,
For asymmetric quasi-regular buyers, {\BayesianOptimalUniformPricing} achieves a tight $\calC_{\BOSP}^{\BOUP} \approx 0.3817$-approximation to {\BayesianOptimalSequentialPricing}.\footnote{\label{footnote:upper_bound:regular}The upper bound holds even for asymmetric regular buyers.}
\begin{align*}
    \calC_{\BOSP}^{\BOUP} & \eqdef \left(2 + \int_{1}^{+\infty} \left(1 - e^{-\calQ}(x)\right) \cdot \d x\right)^{-1} \approx 0.3817, \\
    \calQ(x) & \eqdef \ln\left(\frac{x^{2}}{x^{2} - 1}\right) - \frac{1}{2} \sum_{k = 1}^{\infty} \frac{1}{k^{2} \cdot x^{2k}},\ \forall x \geq 1.
\end{align*}
\end{flushleft}
\end{theorem}

We note that the upper bound of $\calC_{\BOSP}^{\BOUP}$ is achieved by asymmetric regular buyers. In fact, our analysis utilizes a reduction (\Cref{lem:BOSP_BOUP:reduction_to_triangle}) from instances with general quasi-regular distributions to instances with triangular distributions, which are regular.

\begin{lemma}
\label{lem:BOSP_BOUP:reduction_to_triangle}
Fix any sequential posted pricing mechanism $\SP^{(\sigma,\prices)}$ that approaches buyers with deterministic price profile $\prices$ under order $\buyerorder$, and any uniform pricing mechanism $\UP^{(\tilde\price)}$ with deterministic uniform price $\tilde\price$.
Given $n \geq 1$ many independent (but possibly asymmetric) buyers with quasi-regular distributions $\priors = \{\prior_i\}_{i\in[n]}$, there exists $n$ buyers with regular (i.e., triangular) distributions $\auxpriors = \{\auxprior_i\}_{i\in[n]}$ such that
\begin{enumerate}
    \item the expected revenues of sequential posted pricing mechanism $\SP^{(\sigma,\prices)}$ are the same for two groups of buyers, i.e., $\SP^{(\sigma,\prices)}(\auxpriors) = \SP^{(\sigma,\prices)}(\priors)$; and 
    \item the revenue of uniform pricing mechanism $\UP^{(\tilde\price)}$ is smaller for buyers with regular distributions $\auxpriors$, i.e., $\UP^{(\tilde\price)}(\auxpriors) \leq \UP^{(\tilde\price)}(\priors)$.
\end{enumerate}
\end{lemma}
\begin{proof}
For each buyer $i\in[n]$, construct distribution $\auxprior_i$ with support $\supp(\auxprior_i) = [0, \price_i]$ and cumulative density function $\auxcdf_i(\cdot)$ as  
\begin{align*}
    \auxcdf_i(\val) \triangleq \left\{
    \begin{array}{ll}
    \frac{\val}{\val + \frac{1-\cdf_i(\price_i)}{\cdf_i(\price_i)}\price_i}     & ~\text{if $\val \in [0, \price_i)$} \\
    1     & ~\text{if $\val = \price_i$}
    \end{array}
    \right.
\end{align*}
By construction, distribution $\auxprior_i$ is regular and $\prob[\val\sim\prior_i]{\val \geq \price_i} = \prob[\val\sim\auxprior_i]{\val \geq \price_i}$. Hence, the expected revenue of sequential posted pricing mechanism $\SP^{(\sigma,\prices)}$ are the same for two groups of buyers.

Next we compare the expected revenue of uniform pricing mechanism $\UP^{(\tilde\price)}$ for two groups of buyers. Since distribution $\prior_i$ is quasi-regular, invoking condition 3 in \Cref{prop:q-regular equivalent definition}, we have that for every value $\val \in [0, \price_i)$:
\begin{align*}
    \cdf_i(\val) \leq \frac{\val}{\val + \frac{1-\cdf_i(\price_i)}{\cdf_i(\price_i)}\price_i} = \auxcdf_i(\val)
\end{align*}
For every value $\val \in [\price_i, \infty)$, we also have 
\begin{align*}
    \cdf_i(\val) \leq 1 = \auxcdf_i(\val)
\end{align*}
Thus, $\prior_i$ is first order stochastic dominates distribution $\auxprior_i$. 
Consequently, the expected revenue of uniform pricing mechanism $\UP^{(\tilde\price)}$ is weakly smaller for buyers with regular distributions $\auxpriors$.
\end{proof}

Now we are ready to proof \Cref{thm:BOSP_BOUP}.
\begin{proof}[Proof of \Cref{thm:BOSP_BOUP}]
Fix an arbitrary instance with asymmetric quasi-regular buyers. Recall that any randomized sequential posted pricing mechanism (and its expected revenue) can be viewed as a convex combination of deterministic sequential posted pricing mechanism (and their expected revenue). Hence, it is without loss of generality to assume {\BayesianOptimalSequentialPricing} is deterministic. Similarly, it is without loss of generality to assume {\BayesianOptimalUniformPricing} is deterministic. Invoking \Cref{lem:BOSP_BOUP:reduction_to_triangle}, the approximation guarantee for this group of quasi-regular buyers is lower bounded by the approximation guarantee for a group of regular buyers. Therefore, we conclude the proof by invoking \cite[Theorem~1]{JLTX20} for asymmetric regular buyers.
\end{proof}

We have the following two straightforward implications (\Cref{thm:BOM_BOUP,thm:BOM_BOUR}) of \Cref{thm:BOSP_BOUP}.

\begin{theorem}[{\BOM} vs.\ {\BOUP}]
% [{\BayesianOptimalMechanism} vs.\ {\BayesianOptimalUniformPricing}]
\label{thm:BOM_BOUP}
\begin{flushleft}
% Given $n \geq 1$ many independent (but possibly asymmetric) quasi-regular or regular buyers,
For asymmetric quasi-regular buyers, {\BayesianOptimalUniformPricing} achieves a tight $\calC_{\BOM}^{\BOUP} \in [0.2770, 0.3817]$-approximation to {\BayesianOptimalMechanism}.\textsuperscript{\textnormal{\ref{footnote:upper_bound:regular}}}
\end{flushleft}
\end{theorem}

\begin{proof}
The state-of-the-art revenue approximation guarantee of {\BayesianOptimalSequentialPricing} against {\BayesianOptimalMechanism} is $\calC_{\BOM}^{\BOSP} \gtrsim 0.7258$ \cite{PT22,BC23}. Thus the lower-bound part follows as $\calC_{\BOM}^{\BOUP} \geq \calC_{\BOSP}^{\BOUP} \cdot \calC_{\BOM}^{\BOSP} \gtrsim 0.2770$. The upper-bound part is an implication of \Cref{thm:BOSP_BOUP}, namely $\calC_{\BOM}^{\BOUP} \leq \calC_{\BOSP}^{\BOUP} \approx 0.3817$.
\end{proof}

\begin{theorem}[{\BOM} vs.\ $\BOUR$]
% [{\BayesianOptimalMechanism} vs.\ {\BayesianOptimalUniformReserve}]
\label{thm:BOM_BOUR}
\begin{flushleft}
% Given $n \geq 1$ many independent (but possibly asymmetric) quasi-regular or regular buyers,
For asymmetric quasi-regular buyers, {\BayesianOptimalUniformReserve} achieves a tight $\calC_{\BOM}^{\BOUR} \in [0.2770, 0.4630]$-approximation to {\BayesianOptimalMechanism}.\textsuperscript{\textnormal{\ref{footnote:upper_bound:regular}}}
\end{flushleft}
\end{theorem}

\begin{proof}
The lower-bound part is an implication of \Cref{thm:BOM_BOUP}, namely $\calC_{\BOM}^{\BOUR} \geq \calC_{\BOM}^{\BOUP} \gtrsim 0.2770$. The upper-bound part $\calC_{\BOM}^{\BOUR} \lesssim 0.4630$ follows from \cite[Theorem~3]{JLTX20} which studies the upper bound for regular buyers.
\end{proof}

\xhdr{Symmetric Buyers.} We now consider symmetric buyers. Recall that under general distributions, both {\BayesianOptimalUniformReserve} and {\BayesianOptimalUniformPricing} admit constant revenue approximations; and those revenue approximations can be further improved under regular distributions. For quasi-regular distributions, we establish revenue approximation upper bounds (\Cref{prop:BOM_BOUP:iid,prop:BOM_BOUR:iid}) which are strictly worse than tight bounds for regular distributions. Our results suggest that in term of absolute approximation guarantees, quasi-regular distributions are closer to general (aka., irregular) distributions when buyers are symmetric.

We first present the revenue approximation for {\BayesianOptimalUniformReserve}. Its lower bound part follows from \cite{CHMS10,DFK16,har16} which study the same approximation guarantee for general distributions, and the upper bound utilizes \Cref{exp:BOM_BOUR:iid}.

\begin{theorem}
\label{prop:BOM_BOUR:iid}
    For symmetric quasi-regular buyers, {\BayesianOptimalUniformReserve} achieves a tight $\calC_{\BOM}^{\BOUR} \in [0.5, 1 + (\LambertFunc(-\frac{1}{e^2}))^{-1}]$-approximation to {\BayesianOptimalMechanism}, where the upper bound of $1 + (\LambertFunc(-\frac{1}{e^2}))^{-1}\approx 0.6822$.\footnote{In mathematics, the Lambert W function, also called the omega function or product logarithm, is the converse relation of the function $y(x) = x\cdot e^x$. The non-principal branch $x = \LambertFunc(y)$ is the inverse relation when $x \leq -1$.}
\end{theorem}

\begin{example}[Symmetric Quasi-Regular Instances for $\BOM$ vs.\ $\BOUR$]
\label{exp:BOM_BOUR:iid}
Fix any $a \geq 0$.
There are $n$ symmetric quasi-regular buyers with independent and identical valuation distributions
$\priors = \{\prior\}^{\otimes n}$, where distribution $\prior$ has cumulative density function $\cdf(\val) = \frac{\val}{\val + \frac{1}{n}}$ and support $\supp(\prior) = [a, \infty)$.\footnote{Namely, distribution $\prior$ has probability mass $\frac{a}{a + {1}/{n}}$ at value $a$.} 
In this instance, $\frac{\BOUR}{\BOM} = \frac{ae^{-{1}/{a}} + 1}{a +1} - o_n(1)$, which is minimized at $1 + (\LambertFunc(-\frac{1}{e^2}))^{-1}\approx 0.6822$ by setting $a  = -(\LambertFunc(-\frac{1}{e^2}) + 2)^{-1} \approx 0.8725$ and letting $n$ approach infinity.
\ignore{\yfnote{\url{https://www.desmos.com/calculator/tpwzdk0zqi}. The ratio is also the solution to the equation $\frac{x}{1 - x} + \ln(1 - x) = 1$.}}
\end{example}

\begin{proof}[Proof of \Cref{prop:BOM_BOUR:iid}]
    The lower bound part follows from \cite{CHMS10,DFK16,har16}. 
    
    For the upper bound part, we analyze \Cref{exp:BOM_BOUR:iid}. The quasi-regularity of the constructed distribution can be easily checked by algebra. Also see a graphical illustration in \Cref{fig:BOM_BOUR:iid}. Below we verify the approximation ratio stated in this example.

    We first compute the expected revenue of {\BayesianOptimalMechanism}. Let $\auxprior_{(n:n)}$ be the $n$-th order statistic of $n$ i.i.d.\ uniform distribution with support $[0, 1]$. As a sanity check, the quantile of largest value $\val_{(1:n)} \sim \prior$ drawn i.i.d.\ from distribution $\prior$ is exactly $\quant_{(n:n)}\sim \auxprior_{(n:n)}$. Let $\revcurve$ and $\ironrevcurve$ be the revenue curve and ironed revenue curve induced by distribution $\prior$. Invoking \Cref{prop:revenue_equivalence}, the expected revenue of {\BayesianOptimalMechanism} can be expressed as 
    \begin{align*}
        \BOM &{}= \expect
        {\ironrevcurve'(\quant_{(n:n)})}
        % \\
        % &{}
        =
        \expect{\ironrevcurve'(0)\cdot \indicator{\quant_{(0:0)} = 0}}
        +
        \expect{\ironrevcurve'(\quant_{(n:n)})\cdot \indicator{\quant_{(0:0)}>0}}
        % \\
        % &{}
        % \overset{(*)}{=}
        =
        1 + a - \frac{1}{n}
    \end{align*}
    where all expectations are taken over $\quant_{(n:n)}\sim \auxprior_{(n:n)}$. The last equality holds by the facts that $\expect{\ironrevcurve'(0)\cdot \indicator{\quant_{(0:0)}=0}} = 1$ (implied by $\revcurve(0) = \frac{1}{n}$) and $\ironrevcurve'(\quant) = a - \frac{1}{n}$ for every $\quant\in(0, 1)$.

    We next compute the expected revenue of {\BayesianOptimalUniformReserve}. Combining \Cref{prop:revenue_equivalence} and the structure of revenue curve $\revcurve$ (i.e., virtual value $\virtualval(\val)$ is negative for every value $\val \in(a, \infty)$), it suffices to consider {\SecondPriceAuction} with uniform reserve $\reserve = \infty$\footnote{Formally, we consider a finite uniform reserve $\reserve$ and then let it approach infinity.} or uniform reserve $\reserve = a$. Namely, we have
    \begin{align*}
        \BOUR = \max\{\SPA_{\reserve=\infty},\SPA_{\reserve=a}\}
    \end{align*}
    where $\BOUR$, $\SPA_{\reserve=\infty}$, and $\SPA_{\reserve=a}$ are the expected revenue of {\BayesianOptimalMechanism}, {\SecondPriceAuction} with uniform reserve $\reserve = \infty$, and {\SecondPriceAuction} with uniform reserve $\reserve = a$, respectively. Following the argument similar to the one for analyzing $\BOM$, we know $\SPA_{\reserve=\infty} = 1$. To compute $\SPA_{\reserve=a}$, note that for $n\geq 2$,
    \begin{align*}
        \SPA_{\reserve = a} = \expect[\val_{(2:n)}\sim \prior_{(2:n)}]{\val_{(2:n)}}
        =
        \displaystyle\int_{0}^{\infty} 1 - \cdf_{(2:n)}(\val)\cdot\d\val
        =
        a - \frac{(na + 1)^n(na-n+1)-(na)^n(na + 1)}{n(na + 1)^n}
    \end{align*}
    where $\val_{(2:n)}$, $\cdf_{(2:n)}$ are the second order statistic of $n$ i.i.d.\ samples from distribution $\prior$ and its cumulative density function, respectively. The last equality holds by algebra. Moreover, we can verify by algebra that $\SPA_{\reserve = a} \geq 1 - o_n(1)$ for every $a \geq 0$. Putting all the pieces together, we obtain
    \begin{align*}
        \frac{\BOUR}{\BOM} = 
        \frac{    a - \frac{(na + 1)^n(na-n+1)-(na)^n(na + 1)}{n(na + 1)^n}}{1 + a - \frac{1}{n}}
        =
        \frac{ae^{-\frac{1}{a}} + 1}{a +1} - o_n(1)
    \end{align*}
    By considering the first-order condition of $\frac{ae^{-\frac{1}{a}} + 1}{a + 1}$, we know that the ratio is minimized at $1 + (\LambertFunc(-\frac{1}{e^2}))^{-1}\approx 0.6822$ by setting $a  = -(\LambertFunc(-\frac{1}{e^2}) + 2)^{-1} \approx 0.8725$.
\end{proof}

\begin{figure}[ht]
    \centering
    \subfloat[\Cref{exp:BOM_BOUR:iid} with $n = 5$ and $a = -(\LambertFunc(-\frac{1}{e^2}) + 2)^{-1} \approx 0.8725$. The dashed gray line verifies the quasi-regularity.]{
    \input{figure/BOM_BOUR_iid_example}
\label{fig:BOM_BOUR:iid}
}
    \subfloat[\Cref{exp:BOM_BOUP:iid} with $n = 5$.]{
    \input{figure/BOM_BOUP_iid_example}
\label{fig:BOM_BOUP:iid}
}
\\
    \subfloat[\Cref{exp:single_sample:iid:MA} with $n = 5$, $a =  0.6016$ and $\varepsilon = 0.01$. The dashed gray line verifies the quasi-regularity.]{
    \input{figure/BOM_IP_iid_example}
\label{fig:BOM_IP:iid}
}
    \caption{Revenue curves of symmetric quasi-regular buyers in \Cref{exp:BOM_BOUP:iid,exp:BOM_BOUR:iid,exp:single_sample:iid:MA}. The quasi-regularity can be verified by checking Condition~\ref{condition:quasi-regular:revenue curve} in \Cref{prop:q-regular equivalent definition}.}
    \label{fig:BOUR and BOUP:iid}
\end{figure}
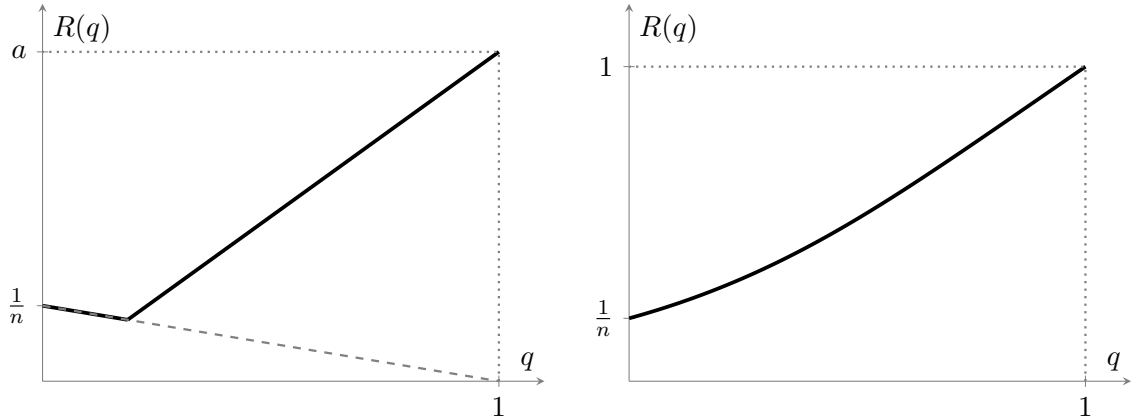
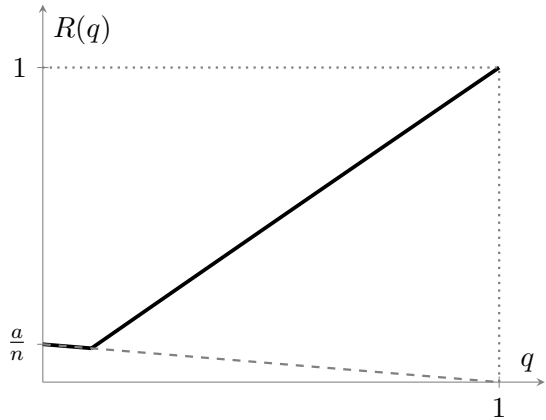

We next present the revenue approximation for {\BayesianOptimalUniformPricing}. It is mainly adopted from \cite{CHMS10,DFK16,har16} which study the same approximation guarantee for general distributions, and the upper bound part utilizes \Cref{exp:BOM_BOUP:iid}.
\begin{theorem}[adopted from \cite{CHMS10,DFK16,har16,JLTX20}]
\label{prop:BOM_BOUP:iid}
    For symmetric quasi-regular buyers, {\BayesianOptimalUniformPricing} achieves a tight $\calC_{\BOM}^{\BOUP} = \frac{1}{2}$, $\calC_{\BOSP}^{\BOUP} = \frac{1}{2}$, $\calC_{\BOUR}^{\BOUP} = \frac{6}{\pi^2}\approx 0.6079$-approximation to {\BayesianOptimalMechanism}, {\BayesianOptimalSequentialPricing}, and {\BayesianOptimalUniformReserve}, respectively.
\end{theorem}

\begin{example}[Symmetric Quasi-Regular Instance for $\BOM$, $\BOSP$, $\BOUR$ vs.\ $\BOUP$]
\label{exp:BOM_BOUP:iid}
There are $n$ symmetric quasi-regular buyers with independent and identical valuation distributions
$\priors = \{\prior\}^{\otimes n}$, where distribution $\prior$ has cumulative density function $\cdf(\val) = (\frac{\val-1}{\val})^{\frac{1}{n}}$ and support $\supp(\prior) = [1, \infty)$. In this instance, $\frac{\BOUP}{\BOM} = \frac{\BOUP}{\BOSP} = \frac{1}{2} - o_n(1)$, and $\frac{\BOUP}{\BOUR} = \frac{6}{\pi^2} - o_n(1) \approx 0.6079$.
\end{example}

\begin{proof}[Proof of \Cref{prop:BOM_BOUP:iid}]
    The lower bound part follows from \cite{CHMS10,DFK16,har16}. For the upper bound part, we analyze \Cref{exp:BOM_BOUP:iid}. The quasi-regularity of the constructed distribution can be easily checked by algebra. Also see a graphical illustration in \Cref{fig:BOM_BOUP:iid}. Next we verify the approximation ratio stated in the theorem. Using the same argument as in \Cref{prop:BOM_BOUR:iid}, we know that the expected revenue of {\BayesianOptimalMechanism} is $2 - \frac{1}{n}$. Moreover, this optimal revenue can be achieved by {\BayesianOptimalSequentialPricing}. For {\BayesianOptimalUniformReserve} and {\BayesianOptimalUniformPricing}, \cite[Section 4.2]{JLTX20} shows that the expected revenues are 1 and $\frac{\pi^2}{6}$, respectively. Combining these results completes the upper bound analysis.
\end{proof}

\xhdr{Tightness of \Cref{thm:BOM_BOUR,thm:BOM_BOUP,prop:BOM_BOUR:iid}.} We conclude this subsection by discussing our conjectures about the tight approximation bounds, and leave them as future directions. 

As we observed in \Cref{thm:BOSP_BOUP}, the revenue approximations of {\BayesianOptimalUniformPricing} against {\BayesianOptimalSequentialPricing} are the same for asymmetric quasi-regular buyers and regular buyers. We conjecture that there is also no separation between quasi-regularity and regularity in terms of the revenue approximation of {\BayesianOptimalUniformPricing} against {\BayesianOptimalMechanism} for asymmetric buyers. In particular, the tight analysis in \cite{JLQTX19} might be extended from regular buyers to quasi-regular buyers.
\begin{conjecture}[$\BOM$ vs.\ $\BOUP$]
\label{conj:BOM_BOUP}
For asymmetric quasi-regular buyers, the tight approximation ratio between {\BayesianOptimalUniformPricing} and {\BayesianOptimalMechanism} is the same as the tight approximation ratio ${\calC_{\BOM}^{\BOUP}}\approx 0.3817$ for asymmetric regular buyers \cite[Theorem~1]{JLQTX19}.
\end{conjecture}

Regarding {\BayesianOptimalUniformReserve}, we conjecture that there is also no separation between quasi-regularity and regularity in terms of the revenue approximation of {\BayesianOptimalUniformReserve} against {\BayesianOptimalMechanism} for asymmetric buyers. To prove this conjecture, one possible approach is to establish a reduction framework (similar to \Cref{lem:BOSP_BOUP:reduction_to_triangle}) from quasi-regular buyers to regular buyers. However, as we observed in \Cref{exp:BOM_BOUR:iid}, for symmetric quasi-regular buyers, there is a strict but bounded separation between quasi-regularity and regularity, and we conjecture that \Cref{exp:BOM_BOUR:iid} is indeed the worst-case instance.

\begin{conjecture}[$\BOM$ vs.\ $\BOUR$]
\label{conj:BOM_BOUR}
For asymmetric quasi-regular buyers, the tight approximation ratio between {\BayesianOptimalUniformPricing} and {\BayesianOptimalMechanism} is the same as the tight approximation ratio $\calC_{\BOM}^{\BOUR}\in[0.3817,0.4630]$ for asymmetric regular buyers \cite{HR09,JLTX20,JLQTX19survey}. 

For symmetric quasi-regular buyers, the tight approximation ratio between {\BayesianOptimalUniformPricing} and {\BayesianOptimalMechanism} is $\calC_{\BOM}^{\BOUR}=1 + (\LambertFunc(-\frac{1}{e^2}))^{-1}\approx 0.6822$.
\end{conjecture}

\begin{comment}

\subsection{The Matroid Settings}
\label{sec:simple-mechanism:matroid}

In this subsection, we consider the matroid settings, which generalizes the single-item setting studied in \Cref{sec:simple-mechanism:single-item}.
The following \Cref{thm:BOM_BMR:quasi-regular} gives a (nontrivial) {\em approximation-preserving} generalization of \cite[Theorem~3.7]{HR09} from regular buyers to quasi-regular buyers.

\begin{theorem}[{\BOM} vs.\ $\BMR$]
% [{\BayesianOptimalMechanism} vs.\ {\BayesianMonopolyReserves}]
\label{thm:BOM_BMR:quasi-regular}
\begin{flushleft}
% In the matroid setting with independent (but possibly asymmetric) quasi-regular buyers,
For asymmetric quasi-regular buyers in a matroid setting, {\BayesianMonopolyReserves} achieves a tight $\calC_{\BOM}^{\BMR} = \frac{1}{2}$-approximation to {\BayesianOptimalMechanism}.\textsuperscript{\textnormal{\ref{footnote:upper_bound:regular}}}
\end{flushleft}
\end{theorem}

\begin{proof}
\green{proved}

This finishes the proof of \Cref{thm:BOM_BMR:quasi-regular}.
\end{proof}

\begin{remark}[Eager vs.\ Lazy Monopoly Reserves]
\label{rem:BOM_BMR:quasi-regular}
The exact simple mechanism studied in \Cref{thm:BOM_BMR:quasi-regular} is {\BayesianEagerMonopolyReserves}, while there is another ``mirror'' mechanism, {\BayesianLazyMonopolyReserves} (\yj{???}).

\cite[Theorem~3.17]{DRY15}

\yj{{\BayesianLazyMonopolyReserves} also achieves a tight $\calC_{\BOM}^{\BMR} = \frac{1}{2}$-approximation to {\BayesianOptimalMechanism}. We omit the detailed arguments for ease of presentation.}
\end{remark}

\end{comment}

\subsection{The Downward-Closed Settings}
\label{sec:simple-mechanism:downward-closed}

% Downward-Closed QMHR VCG with Monopoly Reserves

In this subsection, we consider the downward-closed settings, which generalize the single-item settings studied in \Cref{sec:simple-mechanism:single-item}.
The following \Cref{thm:BOM_BMR:quasi-mhr} gives a (nontrivial) generalization of \cite[Theorem~3.2]{HR09} and \cite[Theorem~3.11]{DRY15} from MHR buyers to quasi-MHR buyers.

\ignore{The generalization of \cite{HR09} is nontrivial. For example, their Lemma 3.1 does not hold for QMHR distributions.}

\begin{theorem}[{\BOM} vs.\ $\BMR$]
\label{thm:BOM_BMR:quasi-mhr}
\begin{flushleft}
For asymmetric quasi-MHR buyers in a downward-closed setting, {\BayesianMonopolyReserves}  achieves a tight $\calC_{\BOM}^{\BMR} = \frac{1}{e+1}$-approximation to {\BayesianOptimalMechanism}.
\end{flushleft}
\end{theorem}

\Cref{thm:BOM_BMR:quasi-mhr} holds for both eager monopoly reserves and lazy monopoly reserves. Moreover, the same approximation guarantee holds even if we compare the expected revenue of {\BayesianMonopolyReserves} (using lazy monopoly reserves) with the optimal expected social surplus, i.e., the first best benchmark. 

The lower bound part of \Cref{thm:BOM_BMR:quasi-mhr} mainly utilizes our structural lemma \Cref{lem:quasi-mhr:structural results} established in \Cref{sec:structural}. The upper bound part utilizes \Cref{exp:BOM_BMR:quasi-mhr}.

\begin{comment}

\yj{The proof of the upper-bound part of \Cref{thm:BOM_BMR:quasi-mhr}.

\url{https://www.desmos.com/calculator/0gz0orkkzf}}

\begin{example}[Asymmetric Quasi-MHR Instances for $\BOM$ vs.\ $\BMR$]
    \label{exp:BOM_BMR:quasi-mhr}
    Fix any integers $k, m\geq 1$ and real number $\varepsilon > 0$. There are $n = \sum_{\ell \in[0:k]} m^\ell$ buyers. All buyers form a tree $T$ with $k + 1$ layers indexed by $[0:k]$ (i.e., tree $T$ has height $k + 1$). See Figure XYZ for a graphical illustration. Each non-leaf buyer (i.e., buyers on the first $k$ layers) has exactly $m$ children. The downward-closed feasibility is constructed such that each buyer is allocated only if none of her ancestors is allocated. Let $a^{(\ell)} \triangleq ((1 + e^{\frac{e-1}{e}})m)^{-\ell} + \varepsilon m^{-\ell}$ for every $\ell\in[0:k]$. The root buyer on layer 0 has deterministic value equal to $a^{(0)}$.
    For each layer $\ell\in[k]$, all buyers on layer $\ell$ have the same distribution~$\prior^{(\ell)}$ with cumulative density function $\cdf^{(\ell)}(\val) = 1 - e^{-\frac{\val}{e\cdot a^{(\ell)}}}$ and support $\supp(\prior^{(\ell)}) = [a^{(\ell)}, \infty)$.\footnote{Namely, distribution $\prior^{(\ell)}$ has probability mass $1-e^{-\frac{1}{e}}$ at value $a^{(\ell)}$.} 
    % where parameter $a^{(\ell)} = ((1 + e^{\frac{e-1}{e}})n)^{1-\ell}$. 
    In this instance, $\frac{\BMR}{\BOM} = \frac{1}{e + 1} - o_{m,k,\varepsilon}(1)$.
\end{example}
\end{comment}

\begin{example}[Asymmetric Quasi-MHR Instances for $\BOM$ vs.\ $\BMR$]
    \label{exp:BOM_BMR:quasi-mhr}
    Fix sufficiently large $m\in\naturals$ and let $\eps = \sqrt{\log(m)/m}$. There are $n = m + 1$ buyers. The first $m$ buyers have identical quasi-MHR distribution $\prior_i$ with support $\supp(\prior_i) = [1, \infty)$ and cumulative density function $\cdf_i(\val) = 1 - e^{-\frac{\val}{e}}$ for $\val\in[1,\infty)$. The last buyer $m + 1$ has a deterministic value equal to $((e+1)e^{-1/e} - \eps)m$. We refer to the first $m$ buyers as the ``small'' buyers and the last buyer $m + 1$ as the ``big'' buyer.
    We say a buyer subset is feasible if it (i) only contains the big buyer; or (ii) contains at most $\lceil(e^{-1/e}+\eps)m\rceil$ small buyers and does not contain the big buyer.
\end{example}

\begin{proof}[Proof of \Cref{thm:BOM_BMR:quasi-mhr}] 
    We first show the lower bound part utilizing \Cref{lem:quasi-mhr:structural results}, and then show the upper bound part by analyzing \Cref{exp:BOM_BMR:quasi-mhr}.
    
    \xhdr{Lower bounding $\calC_{\BOM}^{\BMR}$.} Here we present two arguments for lazy monopoly reserves and eager monopoly reserves separately. 
    
    For {\BayesianLazyMonopolyReserves}, we prove a stronger guarantee that its expected revenue is at least a $\frac{1}{e + 1}$-approximation to the optimal expected social welfare (i.e., the first best benchmark). Since we are in the downward-closed setting, the optimal expected social welfare is achieved by {\sf Vickrey-Clarke-Groves (VCG) Auction} with no reserves. Fix an arbitrary buyer $i$ and all other buyers' valuation profile $\vals_{-i}$. There exists a threshold $\threshold_i$ depending on $\vals_{-i}$ such that buyer $i$ becomes a winner (i.e., receives an item) in {\sf VCG Auction} if and only if her value $\vali$ weakly exceeds $\threshold_i$. Meanwhile, buyer $i$ becomes a winner and pays $\max\{\threshold_i,\optreserve_i\}$ if and only if her value $\vali$ weakly exceeds $\max\{\threshold_i,\optreserve_i\}$ in {\BayesianLazyMonopolyReserves}. Invoking Property~\ref{property:quasi-mhr:revenue approximate welfare} in \Cref{lem:quasi-mhr:structural results}, we know that conditioning on other buyers' valuation profile $\vals_{-i}$, the expected revenue contributed from buyer $i$ in {\BayesianLazyMonopolyReserves} is a $\frac{1}{e + 1}$-approximation to the expected social welfare contributed from buyer $i$ in {\sf VCG Auction}. Thus, taking the expectation over $\vals_{-i}$ and summing over all buyers finish the argument.

    For {\BayesianEagerMonopolyReserves}, we use a similar but slightly different argument. Similar to the above argument for lazy monopoly reserves, invoking Property~\ref{property:quasi-mhr:revenue approximate welfare} in \Cref{lem:quasi-mhr:structural results}, we know that the expected revenue in {\BayesianEagerMonopolyReserves} is a $\frac{1}{e + 1}$-approximation to its own expected social welfare. Thus, it suffices to prove that the expected social welfare in {\BayesianEagerMonopolyReserves} is at least the expected social welfare in {\BayesianOptimalMechanism}. To see this, note that {\BayesianEagerMonopolyReserves} maximizes the (ex post) social welfare among all mechanisms (including {\BayesianOptimalMechanism}) that allocate items to an buyer $i$ only if her value are at least her monopoly reserve, i.e., $\vali \geq \optreserve_i$.

    \xhdr{Upper bounding $\calC_{\BOM}^{\BMR}$.}
    We now analyze \Cref{exp:BOM_BMR:quasi-mhr} which shows the upper bound part of the approximation ratio, i.e., $\calC_{\BOM}^{\BMR} \leq \frac{1}{e+1}$. The quasi-MHR condition of the constructed distribution can be easily checked by algebra.

    We first lower bound the expected revenue of {\BayesianOptimalMechanism}. Since the big buyer $m + 1$ has deterministic value $((e+1)e^{-1/e} - \eps)m$, the optimal revenue is at least $((e+1)e^{-1/e} - \eps)m$, which can be achieved by posting a take-it-or-leave-it price to the big buyer $m + 1$, 
    \begin{align*}
        \BOM \geq ((e+1)e^{-1/e} - \eps)m
    \end{align*} 
    We next analyze the expected revenue of {\BayesianMonopolyReserves}. Note that the construction of valuation distributions ensures that the monopoly reserve of each buyer $i\in[n]$ is equal to the smallest possible value from distribution $\prior_i$. Thus, there is no difference between eager reserve and lazy reserve in this instance. At a high level, we argue that with high probability, $(e^{-1/e}+\eps)m$ small buyers receive items, and each of them pays $1$. Hence, we can upper bound the expected revenue of  {\BayesianMonopolyReserves} by 
    \begin{align}
    \label{eq:BMR UB example analysis}
        \BMR \leq \lceil(e^{-1/e}+\eps)m\rceil + O_m(1)
    \end{align}
    Putting the two pieces together, the upper bound part of $\calC_{\BOM}^{\BMR} = \frac{1}{e + 1}$ is shown as $m$ approaches infinity. In the remainder of the proof, we formally show inequality~\eqref{eq:BMR UB example analysis}.

    Let $X_{(i)}$ be the $i$-th highest value among $m$ small buyers for every $i \in[m]$ and let $k \triangleq \lceil(e^{-1/e}+\eps)m\rceil$. Consider the following two events:
    \begin{align*}
        \event_a \triangleq \left(\sum_{i\in[k]}X_{(i)} > ((e+1)e^{-1/e} - \eps)m\right)
        \;\;
        \mbox{and}
        \;\;
        \event_b \triangleq \left(X_{(k + 1)} = 1\right)
    \end{align*}
    Namely, event $\event_a$ ensures that the items are allocated to small buyers; while event $\event_b$ ensures that the VCG payment of each small buyer is at most $1$. Therefore, the expected revenue of {\BayesianMonopolyReserves} can be decomposed and upper bounded as 
    \begin{align*}
        \BMR &{}= \expect{\BMR\cdot \indicator{\event_a \cap \event_b}}
        +
        \expect{\BMR\cdot \indicator{\neg(\event_a \cap \event_b)}}
        \\
        &{}\leq 
        \expect{\BMR\cdot \indicator{\event_a \cap \event_b}}
        +
        \expect{\BMR\cdot \indicator{\neg\event_a}}
        +
        \expect{\BMR\cdot \indicator{\neg\event_b}}
    \end{align*}
    Below we bound the three terms separately.

    First, note that when both events $\event_a$ and $\event_b$ occur, the items are allocated to the $k$ highest small buyers, and each of them pays the VCG payment equal to $X_{(k + 1)} = 1$. Thus,
    \begin{align*}
        \expect{\BMR\cdot \indicator{\event_a \cap \event_b}}
        \leq 
        \expect{\BMR \given {\event_a , \event_b}}
        =
        k
    \end{align*}

    To upper bound the remaining two terms $\expect{\BMR\cdot \indicator{\neg\event_a}}$ and $\expect{\BMR\cdot \indicator{\neg\event_b}}$, we need to introduce the following auxiliary distribution and random variables $(\auxprior, \{Y_{(i)}\}_{i\in[m]}, \{Z_i\}_{i\in[m]})$: 
    \begin{itemize}
        \item Consider (auxiliary) exponential distribution $\auxprior$ with support $\supp(\auxprior) = [0, \infty)$ and cumulative density function $\auxprior(\val) = 1 - e^{-\val/e}$. By construction, auxiliary distribution $\auxprior$ is first-order stochastically dominated by the original distribution $\prior$ (truncated from below at value 1).
        \item Let random variables $\{Y_{(i)}\}_{i\in[m]}$ be the order statistics over $m$ i.i.d.\ samples from exponential distribution $\auxprior$. By the construction and \Cref{fact:order statistic dominance}, every auxiliary order statistic $Y_{(i)}$ is first-order stochastically dominated by original order statistic $X_{(i)}$. Moreover, for every index $i\in[m]$, $\prob{Y_{(i)} > \threshold} = \prob{X_{(i)} > \threshold}$ for all thresholds $\threshold \geq 1$.
        \item Let random variables $\{Z_{i}\}_{i\in[m]}$ be $m$ i.i.d.\ samples from exponential distribution $\auxprior$. Invoking \Cref{fact:exp order statistic reformulation} below, we can couple $\{Y_{(i)}\}_{i\in[m]}$ and $\{Z_{i}\}_{i\in[m]}$ such that random variables $Y_{(i)} = \sum_{j\in[i:m]}{Z_j}/{j}$ for every $i\in[m]$.
    \end{itemize}
    
    \begin{fact}[Section~2.5 in \cite{DN04}]
    \label{fact:exp order statistic reformulation}
        Given $m \geq 1$ many i.i.d.\ values from exponential distribution~$\auxprior$, every order statistic $Y_{(i:m)}$ can be expressed as $Y_{(i:m)} = \sum_{j\in[i:m]}{Z_j}/{j}$ for $i\in[m]$, where random variables $\{Z_i\}_{i\in[m]}$ are $m$ i.i.d.\ samples from the same exponential distribution $\auxprior$.
    \end{fact}

    When event $\event_a$ fails, the item is allocated to the big buyer with payment equal to his deterministic value (as his monopoly reserve). Thus,
    \begin{align*}
        \expect{\BMR\cdot \indicator{\neg\event_a}} &{} =
        ((e+1)e^{-1/e} - \eps)m \cdot \prob{\neg\event_a} 
    \end{align*}
    where probability $\prob{\neg\event_a}$ can be further upper bounded as 
    \begin{align*}
        \prob{\neg\event_a} & {} = 
        \prob{\sum_{i\in[k]}X_{(i)} < ((e+1)e^{-1/e} - \eps)m} 
        \\
        &{} \overset{(a)}{\leq} 
        \prob{\sum_{i\in[k]}Y_{(i)} < ((e+1)e^{-1/e} - \eps)m} 
        \\
        &{} \overset{(b)}{=} 
        \prob{\sum_{i\in[k]}\sum_{j\in[i:m]}\frac{Z_j}{j}
        < ((e+1)e^{-1/e} - \eps)m
        }
        \\
        &{} = 
        \prob{\sum_{j\in[k]}Z_j + \sum_{j\in[k+1:m]}\frac{k}{j}\cdot Z_j
        < ((e+1)e^{-1/e} - \eps)m
        }
        \\
        &{} \overset{(c)}{\leq}
        O_m\left(\frac{1}{m}\right)
    \end{align*}
    where inequality~(a) and equality~(b) holds due to the construction of random variables $\{Y_{(i)}\}_{i\in[m]}$ and $\{Z_i\}_{i\in[m]}$; and inequality~(c) holds due to the Bernstein inequality. Therefore, we obtain 
    \begin{align*}
        \expect{\BMR\cdot \indicator{\neg\event_a}} = O_{m}(1)
    \end{align*}

    When event $\event_b$ fails, we upper bound the expected revenue by the social welfare. In particular,
    \begin{align*}
        \expect{\BMR\cdot \indicator{\neg\event_b}}
        &\leq 
        ((e+1)e^{-1/e} - \eps)m \cdot \prob{\neg\event_b}
        +
        \expect{\sum_{i\in[k]}X_{(i)}\cdot \indicator{\neg\event_b}}
        \\
        &=
        ((e+1)e^{-1/e} - \eps)m \cdot \prob{X_{(k + 1)} > 1}
        +
        \expect{\sum_{i\in[k]}X_{(i)}\cdot \indicator{X_{(k + 1)} > 1}}
    \end{align*}
    where the two terms on the right-hand side correspond to the maximum social welfare that can be achieved from the big buyer and small buyers, respectively. Note that probability $\prob{X_{(k + 1)} > 1}$ can be upper bounded (similar to $\prob{\neg\event_a}$) as follows:
    \begin{align*} 
        \prob{X_{(k+1)} > 1} 
        % \\
        % &{} 
        \overset{(a)}{=} 
        \prob{Y_{(k+1)} > 1} 
        % \\
        % &{} 
        \overset{(b)}{=} 
        \prob{\sum_{j\in[k+1:m]}\frac{Z_j}{j}
        > 1
        }
        % &{} 
        \overset{(c)}{\leq}
        O_m\left(\frac{1}{m}\right)
    \end{align*}
    where inequality~(a) and equality~(b) holds due to the construction of random variables $\{Y_{(i)}\}_{i\in[m]}$ and $\{Z_i\}_{i\in[m]}$; and inequality~(c) holds due to the Bernstein inequality. On the other side, note that
    \begin{align*}
        & \expect{\sum_{i\in[k]}X_{(i)}\cdot \indicator{X_{(k + 1)} > 1}}
        \\
        {}\overset{(a)}{=}{} &
        \expect{\sum_{i\in[k]}Y_{(i)}\cdot \indicator{Y_{(k + 1)} > 1}}
        \\
        {}\overset{(b)}{=}{} &
        \expect{\sum_{i\in[k]}\sum_{j\in[i:m]}\frac{Z_j}{j}\cdot \indicator{\sum_{j\in[k+1:m]}\frac{Z_j}{j}
        > 1}}
        \\
         {} ={} &
        \expect{\sum_{j\in[k]}{Z_j}\cdot \indicator{\sum_{j\in[k+1:m]}\frac{Z_j}{j}
        > 1
        }}
        +
        k\cdot \expect{\sum_{j\in[k+1:m]}\frac{Z_j}{j}\cdot \indicator{\sum_{j\in[k+1:m]}\frac{Z_j}{j}
        > 1}}
        \\
         \overset{(c)}{=}{} &
        \expect{\sum_{j\in[k]}{Z_j}}\cdot \prob{\sum_{j\in[k+1:m]}\frac{k\cdot Z_j}{j}
        > 1
        }
        +
        k\cdot \expect{\sum_{j\in[k+1:m]}\frac{Z_j}{j}\cdot \indicator{\sum_{j\in[k+1:m]}\frac{Z_j}{j}
        > 1}}
        \\
        \overset{(d)}{\leq}{} &
        k\cdot e \cdot O_m\left(\frac{1}{m}\right)
        +
        O_m(1)
        =
        O_m(1)
    \end{align*}
    where inequality~(a) and equality~(b) holds due to the construction of random variables $\{Y_{(i)}\}_{i\in[m]}$ and $\{Z_i\}_{i\in[m]}$; inequality~(c) holds due to the independence of $\{Z_j\}_{j\in[m]}$; and inequality~(d) holds due to the Bernstein inequality and the construction of $\{Z_j\}_{j\in[m]}$. Putting all the upper bounds together, we obtain
    \begin{align*}
        \expect{\BMR\cdot \indicator{\neg\event_b}} = O_m(1)
    \end{align*}
    which completes the verification of inequality~\eqref{eq:BMR UB example analysis} and concludes the proof of \Cref{thm:BOM_BMR:quasi-mhr}.
\end{proof}

For the special case of MHR distributions, the revenue tight approximation ratios of {\BayesianMonopolyReserves} is $\frac{1}{2}$ for eager reserves \cite[Theorem~3.2]{HR09} and $\frac{1}{e}$ for lazy reserves \cite[Theorem~3.11]{DRY15}, respectively. However, when we relax MHR to quasi-MHR distributions, one reason behind the degeneracy of the approximation ratio to $\frac{1}{e + 1}$ comes from the fact that the revenue curve of a quasi-MHR distribution is not necessarily concave (aka., the distribution is not necessarily regular -- see \Cref{prop:hierarchy} for an example). Due to the non-concavity of the revenue curve, in the VCG-type auctions, for any given buyer $i$,  when the VCG payment (aka., winning threshold) $\threshold_i$ (induced by all other buyers' valuation profile) is greater than the monopoly reserve $\optreserve$, the optimal price $\price^*_i(\threshold_i) \triangleq \argmax_{p\geq \threshold_i} \price(1 - \cdf_i(\price))$ may not be the threshold $\threshold_i$. Motivated by this, we can consider a natural variant of {\BayesianLazyMonopolyReserves} defined as follows:
\begin{definition}
    In the {\sf Bayesian Adaptive Lazy Monopoly Reserves} ($\BMR\primed$ for short) works as follows: (i) compute the VCG payment $\threshold_i$ and adaptive price $\price^*_i(\threshold_i) \triangleq \argmax_{p\geq \threshold_i} \price(1 - \cdf_i(\price))$ for each buyer $i\in[n]$; (ii) post take-it-or-leave-it price $\price^*_i$ to each buyer $i\in[n]$.
\end{definition}

Since the price posted to each buyer does not depend on the buyer's value, {\sf Bayesian Adaptive Lazy Monopoly Reserves} is truthful (i.e., DSIC). Moreover, when buyers have regular distribution, the concavity of revenue curves ensures that $\price^*_i(\threshold_i) = \threshold_i$ for every $\threshold_i \geq \optreserve_i$. Hence, {\sf Bayesian Adaptive Lazy Monopoly Reserves} recovers {\BayesianLazyMonopolyReserves} for regular buyers.

\begin{observation}
    {\sf Bayesian Adaptive Lazy Monopoly Reserves} is dominant strategy incentive compatible (DSIC). For regular buyers, it is equivalent to {\BayesianLazyMonopolyReserves}.
\end{observation}

The improved revenue approximation for this new variant is as follows. The revenue approximation holds even if we compare with the optimal expected social surplus, i.e., the first best benchmark. The analysis is similar to the one for \Cref{thm:BOM_BMR:quasi-mhr}.

\begin{theorem}[{\BOM} vs.\ $\BMR\primed$]
\label{thm:BOM_BMRPlus:quasi-mhr}
\begin{flushleft}
For asymmetric quasi-MHR buyers in a downward-closed setting, {\sf Bayesian Adaptive Lazy Monopoly Reserves} achieves a tight $\calC_{\BOM}^{\BMR\primed} = \frac{1}{3}$-approximation to {\BayesianOptimalMechanism}.
\end{flushleft}
\end{theorem}

\begin{proof}
    We first show the lower bound part utilizing \Cref{lem:quasi-mhr:structural results}, and then show the upper bound part by analyzing \Cref{exp:BOM_BMRPlus:quasi-mhr}.
    
    \xhdr{Lower bounding $\calC_{\BOM}^{\BMR\primed}$.} Here we prove a stronger guarantee that its expected revenue is at least a $\frac{1}{3}$-approximation to the optimal expected social welfare (i.e., the first best benchmark). Since we are in the downward-closed setting, the optimal expected social welfare is achieved by {\sf Vickrey-Clarke-Groves (VCG) Auction} with no reserves. Fix an arbitrary buyer $i$ and all other buyers' valuation profile $\vals_{-i}$. There exists a threshold $\threshold_i$ (aka., VCG payment) depending on $\vals_{-i}$ such that buyer $i$ becomes a winner (i.e., receives an item) in {\sf VCG Auction} if and only if her value $\vali$ weakly exceeds $\threshold_i$. Meanwhile, buyer $i$ receives the item and pays $\price_i^*(\threshold_i)$ if and only if her value $\vali$ weakly exceeds $\price_i^*(\threshold_i)$ in {\sf Bayesian Adaptive Lazy Monopoly Reserves}. Invoking Property~\ref{property:quasi-mhr:revenue approximate welfare} in \Cref{lem:quasi-mhr:structural results}, we know that conditioning on other buyers' valuation profile $\vals_{-i}$, the expected revenue contributed from buyer $i$ in {\sf Bayesian Adaptive Lazy Monopoly Reserves} is a $\frac{1}{3}$-approximation to the expected social welfare contributed from buyer $i$ in {\sf VCG Auction}. Thus, taking the expectation over $\vals_{-i}$ and summing over all buyers finish the argument.

    \xhdr{Upper bounding $\calC_{\BOM}^{\BMR\primed}$.}
    We now analyze \Cref{exp:BOM_BMRPlus:quasi-mhr} which shows the upper bound part of the approximation ratio, i.e., $\calC_{\BOM}^{\BMR\primed} \leq \frac{1}{3}$. The quasi-MHR condition of the constructed distribution can be easily checked by algebra. 

    \begin{example}[Asymmetric Quasi-MHR Instances for $\BOM$ vs.\ $\BMR\primed$]
        \label{exp:BOM_BMRPlus:quasi-mhr}
        Fix sufficiently large $m\in\naturals$ and let $\eps = \sqrt{\log(m)/m}$. There are $n = m + 1$ buyers. The first $m$ buyers have identical quasi-MHR distribution $\prior_i$ with support $\supp(\prior_i) = [1, \infty)$ and cumulative density function $\cdf_i(\val) = 1 - \frac{1}{v}$ for $\val\in[1, e]$ and $\cdf_i(\val) = 1 - e^{-\frac{\val}{e}}$ for $\val\in[e,\infty)$. The last buyer $m + 1$ has a deterministic value equal to $(3 - \eps)m$. We refer to the first $m$ buyers as the ``small'' buyers and the last buyer $m + 1$ as the ``big'' buyer.
        We say a buyer subset is feasible if it does not contain both the big and smaller buyers.
    \end{example}

    We first lower bound the expected revenue of {\BayesianOptimalMechanism}. Since the big buyer $m + 1$ has deterministic value $(3 - \eps)m$, the optimal revenue is at least $(3 - \eps)m$, which can be achieved by posting a take-it-or-leave-it price to the big buyer $m + 1$, 
    \begin{align*}
        \BOM \geq (3 - \eps)m
    \end{align*} 
    We next analyze the expected revenue of {\sf Bayesian Adaptive Lazy Monopoly Reserves}. Note that the construction of valuation distributions ensures that the monopoly reserve of each buyer $i\in[n]$ is equal to the smallest possible value from distribution $\prior_i$. Note that
    \begin{align*}
        \BMR\primed &= \expect{\BMR\primed \cdot \indicator{\sum_{i\in[m]}\val_i \geq \val_{m + 1}}}
        +
        \expect{\BMR\primed \cdot \indicator{\sum_{i\in[m]}\val_i < \val_{m + 1}}}
        \\
        &\leq 
        m\cdot \prob{{\sum_{i\in[m]}\val_i \geq \val_{m + 1}}}
        +
        \val_{m + 1}
        \cdot \prob{{\sum_{i\in[m]}\val_i < \val_{m + 1}}}
        \\
        &=
        m + O_m(1)
    \end{align*}
    where the last equality holds due to the Bernstein inequality. This concludes the proof of \Cref{thm:BOM_BMRPlus:quasi-mhr}.
\end{proof}

%% file: table/table_simple_mechanism.tex
\begin{table}[ht]
    {\centering\small
    \begin{tabular}{|l|>{\raggedright\arraybackslash}p{6.85cm}|>{\arraybackslash}p{6.7cm}|}
        \hline
        \rule{0pt}{13pt} & asymmetric buyers & symmetric buyers \\ [2pt]
        \hline
        \rule{0pt}{13pt}$\regulardistspace$ & \multirow{2}{*}{\rule{0pt}{17pt}\parbox{6.85cm}{$\mathrm{TB}^{\ast} \approx 0.3817$ \cite{AHNPY19,JLTX20,JLQTX19} {\bf [\Cref{thm:BOSP_BOUP,thm:BOM_BOUP}]}}} & $\mathrm{TB} = 1 - 1 / e \approx 0.6321$ \cite{CHMS10,DFK16} \\ [2pt]
        \cline{1-1}\cline{3-3}
        % \hline
        \rule{0pt}{13pt}$\quasiregulardistspace$ & & \multirow{2}{*}{\rule{0pt}{15pt}$\mathrm{TB}^{\star} = \frac{1}{2}$ \cite{CHMS10,DFK16,har16}} \\ [2pt]
        \cline{1-2}
        \rule{0pt}{13pt}$\generaldistspace$ & no revenue guarantee & \\ [2pt]
        \hline
        \multicolumn{3}{l}{\parbox{14.35cm}{\rule{0pt}{13pt}$\ast$ For asymmetric quasi-regular buyers, \Cref{conj:BOM_BOUP} asserts that $\calC_{\BOM}^{\BOUP} = \calC_{\BOSP}^{\BOUP} \approx 0.3817$; regardless of this conjecture, $\calC_{\BOM}^{\BOUP} \in [0.2770, 0.3817]$ (\Cref{thm:BOM_BOUP}).}} \\ [2pt]
        \multicolumn{3}{l}{\parbox{14.35cm}{\rule{0pt}{13pt}$\star$ We can easily check that the $\frac{1}{2}$-approximation instances in \cite{CHMS10,DFK16,har16} satisfies the quasi-regularity condition.}} 
    \end{tabular}
    \par}
    \caption{{\BayesianOptimalMechanism} (resp.\ {\BayesianOptimalSequentialPricing}) vs.\ {\BayesianOptimalUniformPricing}, assuming the correctness of \Cref{conj:BOM_BOUP}. We mark our results in {\bf bold}.}
    \label{tab:simple-mechanism:BOUP}
    \vspace{.1in}
    {\centering\small
    \begin{tabular}{|l|>{\raggedright\arraybackslash}p{6.85cm}|>{\arraybackslash}p{6.7cm}|}
        \hline
        \rule{0pt}{13pt} & asymmetric buyers & symmetric buyers \\ [2pt]
        \hline
        \rule{0pt}{13pt}$\regulardistspace$ & \multirow{2}{*}{\rule{0pt}{17pt}\parbox{6.85cm}{$\mathrm{TB}^{\ast} \in [0.3817, 0.4630]$ \parbox{3.35cm}{\white{.}} \cite{HR09,JLTX20,JLQTX19} {\bf [\Cref{thm:BOM_BOUR}]}}} & $\mathrm{TB} = 1$ \cite{M81} \\ [2pt]
        \cline{1-1}\cline{3-3}
        % \hline
        \rule{0pt}{13pt}$\quasiregulardistspace$ & & $\mathrm{TB}^{\star} \approx 0.6822$ {\bf [\Cref{exp:BOM_BOUR:iid}]} \\ [2pt]
        \hline
        \rule{0pt}{13pt}$\generaldistspace$ & no revenue guarantee & $\mathrm{TB} = \frac{1}{2}$ \cite{CHMS10,DFK16,har16} \\ [2pt]
        \hline
        \multicolumn{3}{l}{\parbox{14.35cm}{\rule{0pt}{13pt}$\ast$ For asymmetric quasi-regular vs.\ regular buyers, \Cref{conj:BOM_BOUR} asserts that both revenue gaps are the same; regardless of this conjecture, they are $\calC_{\BOM}^{\BOUR} \in [0.2770, 0.4630]$ (\Cref{thm:BOM_BOUR}) vs.\ $\calC_{\BOM}^{\BOUR} \in [0.3817, 0.4630]$ \cite{JLTX20,JLQTX19}, respectively.}} \\ [2pt]
        \multicolumn{3}{l}{\parbox{14.35cm}{\rule{0pt}{13pt}$\star$ For i.i.d.\ quasi-regular, \Cref{conj:BOM_BOUR} asserts that the tight revenue gap $\calC_{\BOM}^{\BOUR} \approx 0.6822$ is the unique positive solution to the equation $\frac{x}{1 - x} + \ln(1 - x) = 1$; regardless of this conjecture, this is $\calC_{\BOM}^{\BOUR} \in [\frac{1}{2}, 0.6822]$ by \Cref{exp:BOM_BOUR:iid} and implication of \cite{CHMS10,DFK16,har16}.}}
    \end{tabular}
    \par}
    \caption{{\BayesianOptimalMechanism} vs.\ {\BayesianOptimalUniformReserve}, assuming the correctness of \Cref{conj:BOM_BOUR}. We mark our results in {\bf bold}.}
    \label{tab:simple-mechanism:BOUR}
    \ignore{\vspace{.1in}
    {\centering\small
    \begin{tabular}{|l|>{\raggedright\arraybackslash}p{6.85cm}|>{\arraybackslash}p{6.7cm}|}
        \hline
        \rule{0pt}{13pt} & asymmetric buyers & symmetric buyers \\ [2pt]
        \hline
        \rule{0pt}{13pt}$\regulardistspace$ & & $\mathrm{TB} = 1 - 1 / e \approx 0.6321$ \cite{CHMS10,DFK16} \\ [2pt]
        \cline{1-1}\cline{3-3}
        % \hline
        \rule{0pt}{13pt}$\quasiregulardistspace$ & \multicolumn{2}{l|}{$\mathrm{TB}^{\ast} = 6 / \pi^{2} \approx 0.6079$ \cite{JLTX20}} \\ [2pt]
        \cline{1-1}
        \rule{0pt}{13pt}$\generaldistspace$ & \multicolumn{2}{l|}{} \\ [2pt]
        \hline
        \multicolumn{3}{l}{\parbox{14.35cm}{\rule{0pt}{13pt}$\ast$ To show a matching upper bound $\calC_{\BOUR}^{\BOUP} \leq 6 / \pi^{2} \approx 0.6079$, \cite{JLTX20} provided an asymmetric regular instance \cite[Example~3]{JLTX20} and a symmetric instance \cite[Example~2]{JLTX20}; we can easily check that the latter one satisfies the quasi-regularity condition.}}
    \end{tabular}
    \par}
    \caption{{\BayesianOptimalUniformReserve} vs.\ {\BayesianOptimalUniformPricing}, the revenue guarantees in various single-item settings. (All results in this table were established in previous works.)}
    \vspace{.1in}
    {\centering\small
    \begin{tabular}{|l|>{\raggedright\arraybackslash}p{6.85cm}|>{\arraybackslash}p{6.7cm}|}
        \hline
        \rule{0pt}{13pt} & asymmetric buyers & symmetric buyers \\ [2pt]
        \hline
        \rule{0pt}{13pt}$\regulardistspace$ & \multirow{3}{*}{\rule{0pt}{26pt}\parbox{4.2cm}{$\mathrm{LB} \gtrsim 0.7258$ \cite{PT22,BC23}
        $\mathrm{UB} \lesssim 0.7454$ [implication]}} & \\ [2pt]
        \cline{1-1}
        \rule{0pt}{13pt}$\quasiregulardistspace$ & &  $\mathrm{TB} \approx 0.7454$ \cite{K86,CFHOV21} \\ [2pt]
        \cline{1-1}
        \rule{0pt}{13pt}$\generaldistspace$ & & \\ [2pt]
        \hline
    \end{tabular}
    \par}
    \caption{{\BayesianOptimalMechanism} vs.\ {\BayesianOptimalSequentialPricing}, the revenue guarantees in various single-item settings; there are reductions from general buyers $\generaldistspace$ to regular buyers $\regulardistspace$ \cite{CHMS10,Y11,CFPV19}. (All results in this table were established in previous works.)}}
\end{table}

%% file: figure/BOM_BOUR_iid_example.tex
\begin{tikzpicture}[scale=1, transform shape]
\begin{axis}[
axis line style=gray,
axis lines=middle,
xlabel = $\quant$,
ylabel = $\revcurve(\quant)$,
xtick={0, 1},
ytick={0, 0.2, 0.8725},
xticklabels={0, 1},
yticklabels={0, $\frac{1}{n}$, $a$},
xmin=0,xmax=1.1,ymin=-0.0,ymax=1.,
width=0.5\textwidth,
height=0.4\textwidth,
samples=1000]

\addplot[domain=0.0001:0.18648, black!100!white, line width=0.5mm] (x, {(1 - x)/5});

\addplot[domain=0.18648:1, black!100!white, line width=0.5mm] (x, {0.8725*x});

\addplot[dotted, gray, line width=0.3mm] (1, 0) -- (1, 0.8725) -- (0, 0.8725);

\addplot[domain=0.0001:1, black!100!white, dashed, gray, line width=0.3mm] (x, {(1 - x)/5});

\end{axis}

\end{tikzpicture}

%% file: figure/BOM_BOUP_iid_example.tex
\begin{tikzpicture}[scale=1, transform shape]
\begin{axis}[
axis line style=gray,
axis lines=middle,
xlabel = $\quant$,
ylabel = $\revcurve(\quant)$,
xtick={0, 1},
ytick={0, 0.2, 1},
xticklabels={0, 1},
yticklabels={0, $\frac{1}{n}$, 1},
xmin=0,xmax=1.1,ymin=-0.0,ymax=1.2,
width=0.5\textwidth,
height=0.4\textwidth,
samples=500]

\addplot[domain=0.0001:1, black!100!white, line width=0.5mm] (x, {x/(1-(1-x)^5)});

\addplot[dotted, gray, line width=0.3mm] (1, 0) -- (1, 1) -- (0, 1);

\end{axis}

\end{tikzpicture}

%% file: figure/BOM_IP_iid_example.tex
\begin{tikzpicture}[scale=1, transform shape]
\begin{axis}[
axis line style=gray,
axis lines=middle,
xlabel = $\quant$,
ylabel = $\revcurve(\quant)$,
xtick={0, 1},
ytick={0, 0.12032, 1},
xticklabels={0, 1},
yticklabels={0, $\frac{a}{n}$, $1$},
xmin=0,xmax=1.1,ymin=-0.0,ymax=1.2,
width=0.5\textwidth,
height=0.4\textwidth,
samples=500]

\addplot[domain=0.0001:0.106548, black!100!white, line width=0.5mm] (x, {x * (1-x) * 0.6016 / 5 / x});

\addplot[domain=0.106548:1, black!100!white, line width=0.5mm] (x, {x * (1+0.01 * (1-x))});

\addplot[dotted, gray, line width=0.3mm] (1, 0) -- (1, 1) -- (0, 1);

\addplot[domain=0.0001:1, black!100!white, dashed, gray, line width=0.3mm] (x, {x * (1-x) * 0.6016 / 5 / x});

\end{axis}

\end{tikzpicture}

%% file: source/duplicate.tex
\section{Revenue Approximation by Prior-Independent Mechanisms}
% \section{Generalized Bulow-Klemperer Theorems}
\label{sec:duplicate}

\input{table/table_duplicate}

The celebrated paper \cite{BK96} shows that in the single-item setting, for any symmetric regular buyers instance, the expected revenue of {\SecondPriceAuction} (which requires no information about buyers' prior and hence is prior-independent) after adding one additional buyer (hereafter, {\oneDuplicateSecondPriceAuction}), is at least the expected revenue of {\BayesianOptimalMechanism} (which may heavily rely on the knowledge of buyers' prior and hence is prior-dependent) for the original instance (without duplicates). 

In this section, we study the extensions of \cite{BK96} for more general settings. In \Cref{subsec:duplicate:single_item}, we study the revenue approximation of {\oneDuplicateSecondPriceAuction} for (possibly asymmetric) quasi-regular buyers in the single-item environment. In \Cref{subsec:duplicate:downward-closed}, we study the revenue approximation of {\nDuplicateVCGAuction} for (possibly asymmetric) quasi-MHR buyers in the downward-closed settings.

\subsection{The Single-Item Setting}
\label{subsec:duplicate:single_item}

In this subsection, we consider the extensions of \cite{BK96} to asymmetric quasi-regular buyers. Specifically, we are interested in the revenue approximation guarantee of the prior-independent {\SecondPriceAuction} after adding one duplicate buyer (i.e., {\oneDuplicateSecondPriceAuction}) against some revenue benchmark of the original instance (without duplicates). Here we allow the seller to decide which buyer to duplicate based on their valuation distributions.

Recall that when all buyers are symmetric (i.e., have the same valuation distribution) and regular, {\BayesianOptimalMechanism} and {\BayesianOptimalUniformReserve} are equivalent. However, for more general settings considered in this subsection, these two mechanisms are no longer equivalent, but have a bounded gap as we shown in \Cref{sec:simple-mechanism:single-item}. Because of this, we present revenue approximations against both benchmarks. All results in this subsection are summarized in \Cref{tab:duplicate:BOUR} and \Cref{tab:duplicate:BOM}.

We first compare {\oneDuplicateSecondPriceAuction} and {\BayesianOptimalUniformReserve}. \Cref{thm:BOUR_SPAone} gives a (nontrivial) {\em approximation-preserving} generalization of \cite{BK96} from symmetric regular buyers to asymmetric quasi-regular buyers.

\begin{theorem}[$\BOUR$ vs.\ $\SPAone$]
\label{thm:BOUR_SPAone}
\begin{flushleft}
For asymmetric quasi-regular buyers, {\oneDuplicateSecondPriceAuction} revenue-surpasses {\BayesianOptimalUniformReserve} (without duplicates).\footnote{The upper bound holds even for symmetric regular buyers.}
\end{flushleft}
\end{theorem}

As mentioned in \Cref{sec:intro:contribution}, the special case $n = 1$ (essentially symmetric buyers) of \Cref{thm:BOUR_SPAone} has been established in \cite[Lemma~D.1]{HR14}. Nonetheless, the proof of the general case $n \geq 1$ is technically nontrivial.

\begin{proof}[Proof of \Cref{thm:BOUR_SPAone}]
Fix an arbitrary $n\geq 1$ and an arbitrary group of $n$ asymmetric buyers with quasi-regular distributions $\priors = \{\prior_i\}_{i\in[n]}$. Let $\reserve$ denote the uniform reserve used in {\BayesianOptimalUniformReserve} (without duplicates). 

Define critical buyer index $k\triangleq \argmin_{i\in[n]} \cdf_i(\reserve)$. In the remainder of the analysis, we prove the theorem statement by arguing that {\sf Second Price Auction} given one additional duplicate buyer with prior $\prior_k$ revenue-surpasses {\BayesianOptimalUniformReserve}, i.e.,
\begin{align*}
    \SPAone(\priors)\geq \SPA(\priors\otimes\prior_k) \geq \BOUR(\priors)
\end{align*}
We start our argument by presenting the closed-form characterizations of $\BOUR(\priors)$ as follows:
\begin{align*}
    \BOUR(\priors) &{} =
    \reserve\cdot \prob[\vals\sim\priors]{\val_{(1:n)}\geq \reserve}
    +
    \expect[\vals\sim\priors]{\val_{(2:n)} - \reserve\given \val_{(2:n)} \geq \reserve}
    \cdot 
     \prob[\vals\sim\priors]{\val_{(2:n)}\geq \reserve}
     \\
     &{}=
     \reserve\cdot \left(1 - \cdf_{(1:n)}(\reserve)\right)
     +
     \displaystyle\int_{\reserve}^{\infty}1 - \cdf_{(2:n)}(\val)\cdot\d \val
     % \\
     % &{}
     % =
     % \reserve\cdot\left(1 - \prod_{i\in[n]}\cdf_i(\reserve)\right)
     % +
     % \displaystyle
     % \int_{\reserve}^{\infty}
     % \left(
     % 1 - \prod_{i\in[n]}\cdf_i(\val)
     % \left(1 + \sum_{j\in[n]}\frac{1-\cdf_j(\val)}{\cdf_j(\val)}\right)
     % \right)\cdot\d\val
\end{align*}
where $\cdf_{(1:n)}$ and $\cdf_{(2:n)}$ are the first and second order statistics distributions induced by $n$ distributions $\priors$, and all three equalities hold by the definition. Similarly, we can characterize $\SPA(\priors\otimes\prior_k)$ as follows:
\begin{align*}
    \SPA(\priors\otimes\prior_k) 
    & {} = 
    \expect[\vals\sim\priors\otimes\prior_k]{\val_{(2:n+1)}}
    =
    \displaystyle\int_0^{\infty}1 - \cdf_{(2:n + 1)}(\val)\cdot\d\val
    % =
    %  \displaystyle
    %  \int_{0}^{\infty}
    %  \left(
    %  1 - \cdf_k(\val)\prod_{i\in[n]}\cdf_i(\val)
    %  \left(\frac{1}{\cdf_k(\val)}+\sum_{j\in[n]}\frac{1-\cdf_j(\val)}{\cdf_j(\val)}\right)
    %  \right)\cdot\d\val
\end{align*}
where $\cdf_{(2:n+1)}$ is the second order statistics distributions induced by $n+1$ distributions $\priors\otimes\prior_k$.
Note that 
\begin{align*}
    &\SPA(\priors\otimes\prior_k)  - \BOUR(\priors)
    \\
    ={}&
    % \left(
    \displaystyle\int_0^{\reserve}1 - \cdf_{(2:n + 1)}(\val)\cdot\d\val
    -
    \reserve\cdot \left(1 - \cdf_{(1:n)}(\reserve)\right)
    % \right)
    +
    % \left(
    \displaystyle\int_{\reserve}^{\infty}1 - \cdf_{(2:n + 1)}(\val)\cdot\d\val
    -
    \displaystyle\int_{\reserve}^{\infty}1 - \cdf_{(2:n)}(\val)\cdot\d \val
    % \right)
\end{align*}
Since order statistics distribution $\prior_{(2:n+1)}$ first order stochastically dominates order statistics distribution $\prior_{(2:n)}$, the term $
\int_{\reserve}^{\infty}1 - \cdf_{(2:n + 1)}(\val)\cdot\d\val
-
\int_{\reserve}^{\infty}1 - \cdf_{(2:n)}(\val)\cdot\d \val
$ in the right-hand side is non-negative. 
Therefore, it suffices to verify that
\begin{align*}
    \displaystyle\int_0^{\reserve}1 - \cdf_{(2:n + 1)}(\val)\cdot\d\val
    -
    \reserve\cdot \left(1 - \cdf_{(1:n)}(\reserve)\right)
    \geq 0
\end{align*}
To show this, we construct the following auxiliary distributions $\auxpriors=\{\auxprior_i\}_{i\in[n]}$: for each buyer $i\in[n]$, define distribution $\auxprior_i$ with support $\supp(\auxprior_i) = [0, \reserve]$ and following cumulative density function
\begin{align*}
    \auxcdf_i(\val) = 
    \left\{
    \begin{array}{ll}
     \frac{\val}{\val + \frac{1 - \cdf_i(\reserve)}{\cdf_i(\reserve)}\reserve}    & \text{~ if $\val \geq \reserve$} \\
    1   & \text{~ if $\val > \reserve$} 
    \end{array}
    \right.
\end{align*}
As a sanity check, by construction, auxiliary distribution $\auxprior_i$ is regular and $\auxcdf_i(\reserve) = \cdf_i(\reserve)$.
Moreover, since distribution $\prior_i$ is quasi-regular, invoking Condition~\ref{condition:quasi-regular:cdf}, the constructed auxiliary distribution $\auxprior_i$ is first order stochastically dominated by distribution $\prior_i$. Therefore, we have 
\begin{align*}
    &\displaystyle\int_0^{\reserve}1 - \cdf_{(2:n + 1)}(\val)\cdot\d\val
    \geq
    \displaystyle\int_0^{\reserve}1 - \auxcdf_{(2:n + 1)}(\val)\cdot\d\val
    \;\;\mbox{and}\;\;
    \reserve\cdot \left(1 - \cdf_{(1:n)}(\reserve)\right)
    =
    \reserve\cdot \left(1 - \auxcdf_{(1:n)}(\reserve)\right)
\end{align*}
Note that by construction, 
\begin{align*}
    \displaystyle\int_0^{\reserve}1 - \cdf_{(2:n + 1)}(\val)\cdot\d\val
    =
    \expect[\vals\sim\auxpriors\otimes \auxprior_k]{\val_{(2:n+1)}}
    =
    \SPA(\auxpriors\otimes \auxprior_k)
\end{align*}
Let $\virtualval_i\primed(\cdot)$ be the virtual value function of auxiliary distribution $\auxprior_i$. By construction, for every buyer index $i\in[n]$, $\virtualval_i\primed(\reserve) = \reserve$. Moreover, since critical buyer index $k = \argmin_{i\in[n]}\cdf_i(\reserve) = \argmin_{i\in[n]}\auxcdf_i(\reserve)$, we have $\virtualval_i\primed(\val) \geq - \frac{1 - \auxcdf_k(\reserve)}{\auxcdf_k(\reserve)}\reserve$ for every buyer index $i\in[n]$ and value $\val \in [0,\reserve)$. Thus, invoking \Cref{prop:revenue_equivalence}, we can lower bound $\SPA(\auxpriors\otimes \auxprior_k)$ as follows:
\begin{align*}
    \SPA(\auxpriors\otimes \auxprior_k) &{}\geq 
    \reserve\cdot \left(1  - \auxcdf_k(\reserve) \prod_{i\in[n]}\auxcdf_i(\reserve)\right)
    -
    \frac{1 - \auxcdf_k(\reserve)}{\auxcdf_k(\reserve)}\reserve
    \cdot
    \auxcdf_k(\reserve)\prod_{i\in[n]}\auxcdf_i(\reserve)
    \\
    &{}=
    \reserve\cdot \left(1  - \prod_{i\in[n]}\auxprior_i(\reserve)\right)
    =
    \reserve\cdot \left(1 - \auxcdf_{(1:n)}(\reserve)\right)
    =
    \reserve\cdot \left(1 - \cdf_{(1:n)}(\reserve)\right)
\end{align*}
which finishes the proof of \Cref{thm:BOUR_SPAone} as desired.
\end{proof}

Combining \Cref{thm:BOUR_SPAone} with results established in \Cref{sec:simple-mechanism:single-item}, we obtain revenue approximation guarantees of {\oneDuplicateSecondPriceAuction} against {\BayesianOptimalMechanism} in \Cref{thm:BOM_SPAone} and \Cref{thm:BOM_SPAone iid} for asymmetric buyers and symmetric buyers, respectively.

\begin{theorem}[{\BOM} vs.\ $\SPAone$ -- Asymmetric Buyers]
\label{thm:BOM_SPAone}
% \begin{flushleft}
For asymmetric quasi-regular or regular buyers,
{\oneDuplicateSecondPriceAuction} achieves a tight $\calC_{\BOM}^{\SPAone} \in [0.2770, 0.6365]$- or $\calC_{\BOM}^{\SPAone} \in [0.3817, 0.6365]$- approximation to {\BayesianOptimalMechanism} (without duplicates), respectively.
% \end{flushleft}
\end{theorem}

We note that for the special case of asymmetric regular buyers, \Cref{thm:BOM_SPAone} improves the state-or-art: the lower bound is increased from 0.108 \cite[Theorem~3.1]{FLR19} to 0.3817 and the upper bound is decreased from $\ln(2)\approx 0.6931$ \cite[Remark~5.7]{FLR19} to 0.6365. In particular, the upper bound is analyzed using the following example.

\begin{example}[Asymmetric Regular Instances for $\BOM$ vs.\ $\SPAone$] 
\label{exp:BOM_SPAone:asymmetric}
Fix any $a \geq 0$. There are two asymmetric regular buyers. Buyer 1 has deterministic value $a$. Buyer 2 has valuation distribution $\prior$ with cumulative density function $\cdf(\val) = \frac{\val}{\val + 1}$ and support $\supp(\prior) = [0,\infty)$. 
In this instance, $\frac{\SPAone}{\BOM} = {\max\{a, 2\ln(1+a)-1+\frac{2}{a+1}\}}/({1+a})$, which is minimized at approximately 0.6365 by setting $a \approx 1.75088$.\ignore{\yfnote{\url{https://www.desmos.com/calculator/5v6rpaxfmi}}}
\end{example}

\begin{proof}[Proof of \Cref{thm:BOM_SPAone}]
    For the lower bound part, note that $\calC_{\BOM}^{\SPAone}\geq \calC_{\BOUR}^{\SPAone} \cdot \calC_{\BOM}^{\BOUR}$. Therefore, invoking the results that $\calC_{\BOUR}^{\SPAone} = 1$ (\Cref{thm:BOUR_SPAone}) and $\calC_{\BOM}^{\BOUR} \geq 0.2770$ (\Cref{thm:BOM_BOUR}) or $\calC_{\BOM}^{\BOUR} \geq 0.3817$ for asymmetric regular buyers \cite[Theorem~1]{JLQTX19} finishes the lower bound analysis.

    For the upper bound part, we analyze \Cref{exp:BOM_SPAone:asymmetric}. The regularity of the constructed distributions can be easily checked by algebra. Below we verify the approximation ratio stated in this example.

    We first lower bound the expected revenue in {\BayesianOptimalMechanism}. Note that the seller can first sell the item to Buyer 2 with a take-it-or-leave-it price $H$, and get expected revenue $\frac{H}{H + 1}$. With probability $\frac{1}{H + 1}$, Buyer 2 does not purchase and then the item is sold to Buyer 1 with price~$a$. Letting price $H$ approach infinite, we obtain a lower bound of $1 + a$ for the expected revenue in {\BayesianOptimalMechanism}.

    Next we analyze the expected revenue in {\oneDuplicateSecondPriceAuction}. By duplicating Buyer~1, the ex post revenue is always equal to $a$. Next, we compute the expected revenue by duplicating Buyer 2. In this case, the expected revenue can be expressed as 
    \begin{align*}
        &\displaystyle\int_0^a \val \cdot 2\cdf(\val) \pdf(\val)\cdot\d\val
        +
        2a \cdf(a)(1 - \cdf(a)) 
        +
        \displaystyle\int_a^\infty \val\cdot 2(1 - \cdf(\val))  \pdf(\val)\cdot\d\val
        % \\
        =
        2\ln(1 + a) - 1 + \frac{2}{a + 1}
    \end{align*}
    where three terms on the left-hand side correspond to the revenue contribution when (i) both Buyer~2 and her duplicate have values at most $a$, (ii) one of Buyer 2 and her duplicate has a value at most $a$, (iii) both Buyer 2 and her duplicate have values at least $a$, respectively.

    Putting all the pieces together, we obtain the approximation ratio stated in \Cref{exp:BOM_SPAone:asymmetric} as desired.
\end{proof}

For symmetric quasi-regular buyers, we formalize the revenue approximation in the following theorem. Its upper bound utilizes \Cref{exp:BOM_SPAone:iid}, which is the same as \Cref{exp:BOM_BOUR:iid} in \Cref{sec:simple-mechanism:single-item}.

\begin{theorem}[{\BOM} vs.\ $\SPAone$ -- Symmetric Buyers]
\label{thm:BOM_SPAone iid}
For symmetric quasi-regular buyers,
{\oneDuplicateSecondPriceAuction} achieves a tight $\calC_{\BOM}^{\SPAone} \in [0.2770, 1 + (\LambertFunc(-\frac{1}{e^2}))^{-1}]$-approximation to {\BayesianOptimalMechanism}, where the upper bound of $1 + (\LambertFunc(-\frac{1}{e^2}))^{-1}\approx 0.6822$.
\end{theorem}

\begin{example}[Symmetric Quasi-Regular Instances for $\BOM$ vs.\ $\SPAone$]
\label{exp:BOM_SPAone:iid}
Fix any $a \geq 0$.
There are $n$ symmetric buyers with independent and identical valuation distributions
$\priors = \{\prior\}^{\otimes n}$, where distribution $\prior$ has cumulative density function $\cdf(\val) = \frac{\val}{\val + {1}/{n}}$ and support $\supp(\prior) = [a, \infty)$. 
In this instance, $\frac{\SPAone}{\BOUR} = \frac{ae^{-{1}/{a}} + 1}{a +1} + o_n(1)$, which is minimized at $1 + (\LambertFunc(-\frac{1}{e^2}))^{-1}\approx 0.6822$ by setting $a  = -(\LambertFunc(-\frac{1}{e^2}) + 2)^{-1} \approx 0.8725$ and letting $n$ approach infinity.
\end{example}

\begin{proof}[Proof of \Cref{thm:BOM_SPAone iid}]
    The lower bound part is a direct implication of \Cref{thm:BOM_SPAone}.

    For the upper bound part, we analyze \Cref{exp:BOM_SPAone:iid}. Since \Cref{exp:BOM_SPAone:iid} is the same as \Cref{exp:BOM_BOUR:iid}, using the analysis of \Cref{prop:BOM_BOUR:iid}, we know the constructed distribution is quasi-regular, and the expected revenue of {\BayesianOptimalMechanism} is equal to 
    $$1 + a - \frac{1}{n}$$
    Moreover, the expected revenue of {\oneDuplicateSecondPriceAuction} is at most {\BayesianOptimalUniformReserve} with one duplicate buyer, which is 
    $$a - \frac{((n+1)a + 1)^{n+1}((n+1)a-n)-((n+1)a)^{n+1}((n+1)a + 1)}{(n+1)((n+1)a + 1)^{n+1}} = ae^{\frac{1}{a}} + 1 + o_n(1)$$
    Putting two pieces together, we obtain
    \begin{align*}
        \frac{\SPAone}{\BOM} = 
        \frac{ae^{-\frac{1}{a}} + 1}{a +1} + o_n(1)
    \end{align*}
    By considering the first-order condition of $\frac{ae^{-\frac{1}{a}} + 1}{a + 1}$, we know that the ratio is minimized at $1 + (\LambertFunc(-\frac{1}{e^2}))^{-1}\approx 0.6822$ by setting $a  = -(\LambertFunc(-\frac{1}{e^2}) + 2)^{-1} \approx 0.8725$.
\end{proof}

\begin{comment}

\subsection{The Matroid Settings}
\label{subsec:duplicate:matroid}

\begin{theorem}[{\BOM} vs.\ $\VCGn$]
% [{\BayesianOptimalMechanism} vs.\ {\nDuplicateVCGAuction}]
\label{thm:BOM_VCGn:quasi-regular}
\begin{flushleft}
% In the matroid setting with independent (but possibly asymmetric) quasi-regular buyers,
For quasi-regular buyers in a matroid setting,
{\nDuplicateVCGAuction} achieves a tight $\calC_{\BOM}^{\VCGn} \in [\frac{1}{2}, \frac{3}{4}]$-approximation to {\BayesianOptimalMechanism} (without duplicates).\textsuperscript{\textnormal{\ref{footnote:upper_bound:regular}}}
\end{flushleft}
\end{theorem}

\begin{proof}
\green{proved}

This finishes the proof of \Cref{thm:BOM_VCGn:quasi-regular}.
\end{proof}

\begin{remark}[Tightness of ]
the best known impossibility result (i.e., the upper bound) is $\frac{3}{4}$ \cite{HR09}
\end{remark}

\end{comment}

\subsection{The Downward-Closed Settings}
\label{subsec:duplicate:downward-closed}

In this subsection, we consider the downward-closed settings, which generalize the single-item settings studied in \Cref{subsec:duplicate:single_item}. Instead of duplicating a single buyer, we study the revenue approximation of {\SecondPriceAuction} which has a duplicate for each of $n$ buyers. Specifically, in the duplicated environment, each buyer is replaced by a pair of buyers with the same valuation distribution. The feasible sets of the duplicated environment are defined such that (i) at most one buyer from each pair can be allocated, and (ii) the set of winners (who receive items) when naturally interpreted as a set of buyers from the original instance, is a feasible set in that instance.

The following \Cref{thm:BOM_VCGn:quasi-mhr} gives a (nontrivial) {\em approximation-preserving} generalization of \cite[Theorem~4.2]{HR09} from MHR buyers to quasi-MHR buyers.

\begin{theorem}[{\BOM} vs.\ $\VCGn$]
\label{thm:BOM_VCGn:quasi-mhr}
% \begin{flushleft}
For quasi-MHR buyers in a downward-closed setting,
{\nDuplicateVCGAuction} achieves a tight $\calC_{\BOM}^{\VCGn} = \frac{1}{3}$-approximation to {\BayesianOptimalMechanism} (without duplicates).
% \footnote{
(The upper bound holds even for asymmetric MHR buyers.)
% }
% \end{flushleft}
\end{theorem}

The approximation guarantee in \Cref{thm:BOM_VCGn:quasi-mhr} holds even if we compare the expected revenue of {\nDuplicateVCGAuction} with the optimal expected social surplus, i.e., the first best benchmark.
Since it is an approximation preserving generalization of \cite{HR09}, the upper bound for MHR buyers naturally implies the upper bound for quasi-MHR buyers. Meanwhile, similar to the high-level proof idea in \cite{HR09}, the lower bound argument heavily relies on Property~\ref{property:quasi-mhr:revenue approximate welfare duplicating} of \Cref{lem:quasi-mhr:structural results} established in \Cref{sec:structural}. 

\begin{proof}[Proof of \Cref{thm:BOM_VCGn:quasi-mhr}]
The upper-bound part follows from \cite[Example~4.3]{HR09}. For the lower bound part, we prove a stronger claim that the expected revenue of {\nDuplicateVCGAuction} is a $\frac{1}{3}$ approximation to the optimal expected social surplus, i.e., the first best benchmark. To show this claim, consider an arbitrary buyer $i$ and her duplicate. Denote the maximum value between buyer $i$ and her duplicate by $\val_{(1:2)}^{(i)}$. Fix any valuation profile of all other buyers and their duplicates. In {\nDuplicateVCGAuction}, there exists a threshold $\threshold$ such that the buyer $i$ or her duplicate receives an item if and only if $\val_{(1:2)}^{(i)} \geq \threshold$. Invoking \Cref{prop:revenue_equivalence} and Property~\ref{property:quasi-mhr:revenue approximate welfare duplicating}, the conditional expected revenue contribution from buyer $i$ and her duplicate is a $\frac{1}{3}$-approximation to the conditional expected social welfare contribution from buyer $i$ and her duplicate. Taking expectation over threshold $\threshold$ (i.e., other buyers and their duplicates' valuation profile) and summing over all buyers and their duplicates, we know that the expected revenue of {\nDuplicateVCGAuction} is a $\frac{1}{3}$-approximation to the expected social welfare of {\nDuplicateVCGAuction}. Since {\nDuplicateVCGAuction} maximizes the expected social welfare in the duplicated environment, its expected social welfare is at least the optimal social welfare for the original instance (without duplicates).
\end{proof}

%% file: table/table_duplicate.tex
\begin{table}[t]
    {\centering\small
    \begin{tabular}{|l|>{\arraybackslash}p{4.75cm}|>{\arraybackslash}p{4.45cm}|>{\arraybackslash}p{3.9cm}|}
        \hline
        \rule{0pt}{13pt} & asymmetric multi-buyer & symmetric multi-buyer & single-buyer \\ [2pt]
        \hline
        \rule{0pt}{13pt}$\regulardistspace$ & \multicolumn{3}{c|}{\multirow{2}{*}{\rule{0pt}{15pt}$\mathrm{TB} = 1$ \cite[Thm~1]{BK96} {\bf [\Cref{thm:BOUR_SPAone}]}}} \\ [2pt]
        \cline{1-1}
        \rule{0pt}{13pt}$\quasiregulardistspace$ & \multicolumn{3}{c|}{} \\ [2pt]
        \hline
        \rule{0pt}{13pt}$\generaldistspace$ & \multicolumn{3}{c|}{no revenue guarantee} \\ [2pt]
        \hline
    \end{tabular}
    \par}
    \caption{{\oneDuplicateSecondPriceAuction} vs.\ {\BayesianOptimalUniformReserve}, the revenue guarantees in various single-item settings. We mark our results in {\bf bold}.}
    \label{tab:duplicate:BOUR}
    \vspace{.1in}
    {\centering\small
    \begin{tabular}{|l|>{\arraybackslash}p{4.75cm}|>{\arraybackslash}p{4.45cm}|>{\arraybackslash}p{3.9cm}|}
        \hline
        \rule{0pt}{13pt} & asymmetric multi-buyer & symmetric multi-buyer & single-buyer \\ [2pt]
        \hline
        \rule{0pt}{13pt}\multirow{2}{*}{\rule{0pt}{15pt}$\regulardistspace$} & $\mathrm{LB} \gtrsim 0.3817$ {\bf [\Cref{thm:BOUR_SPAone}]} & \multicolumn{2}{c|}{\multirow{2}{*}{\rule{0pt}{15pt}$\mathrm{TB} = 1$ \cite[Thm~1]{BK96}}} \\ [2pt]
        \rule{0pt}{13pt} & $\mathrm{UB} \lesssim 0.6365$ {\bf [\Cref{exp:BOM_SPAone:asymmetric}]} & \multicolumn{2}{c|}{} \\ [2pt]
        \hline
        \rule{0pt}{13pt}\multirow{2}{*}{\rule{0pt}{15pt}$\quasiregulardistspace$} & $\mathrm{LB} \gtrsim 0.2770$ {\bf [\Cref{thm:BOUR_SPAone}]} & $\mathrm{LB} \gtrsim 0.2770$ {\bf [implication]} & \multirow{2}{*}{\rule{0pt}{15pt}$\mathrm{TB} = 1$ \cite[Lem~D.1]{HR14}} \\ [2pt]
        \rule{0pt}{13pt} & $\mathrm{UB} \lesssim 0.6365$ {\bf [implication]} & $\mathrm{UB} \lesssim 0.6822$ {\bf [\Cref{exp:BOM_SPAone:iid}]} & \\ [2pt]
        \hline
        \rule{0pt}{13pt}$\generaldistspace$ & \multicolumn{3}{c|}{no revenue guarantee} \\ [2pt]
        \hline
    \end{tabular}
    \par}
    \caption{{\oneDuplicateSecondPriceAuction} vs.\ {\BayesianOptimalMechanism}, the revenue guarantees in various single-item settings. We mark our results in {\bf bold}.}
    \label{tab:duplicate:BOM}
\end{table}

%% file: source/single_sample.tex
\section{Revenue Approximation by a Single Sample}
\label{sec:single-sample}

In this section, we study the revenue maximization problem with a single sample access. A rapidly growing literature \cite{DRY15,FILS15,HMR18,BGMM18,DZ20,FHL21,ABB22} has been established in this direction. An important motivation in this area is to identify and characterize the ``minimum information'' needed for achieving a meaningful revenue approximation guarantee.

We focus on the single-item setting for multiple buyers. In this model, the seller does not know the buyers' valuation distributions $\priors$, but has a single sample access of the first order statistic of the valuation profile.
% \footnote{Recall that given a realized valuation profile $\vals = \{\vali\}_{i\in[n]}$ drawn from joint distribution $\priors=\{\prior_i\}_{i\in[n]}$ independently, the first order statistic $\val_{(1:n)} \triangleq \max_{i\in[n]} \vali$ is defined as the maximum value, and follows distribution $\prior_{(1:n)}$ with cumulative density function $\cdf_{(1:n)}(\val) \triangleq \prod_{i\in[n]} \cdf_i(\val)$.} 
Specifically, while neither the valuation profile $\vals = \{\vali\}_{i\in[n]}$ nor the valuation distributions $\priors = \{\priori\}_{i\in[n]}$ is known by the seller, she has a single sample $\firstsample\sim\firstprior$ drawn from the (unknown) first order statistic (i.e., highest value) distribution $\firstprior$ and is independent with the realized valuation profiles $\{\vali\}_{i\in[n]}$. 

% \begin{remark}
We remark that the single sample access of the first order statistic can be viewed as the ``minimal prior information'', since no bounded approximation guarantee can be achieved when the seller has no information about distributions $\priors$, or has a sample access of $k$-th order statistic with $k\geq 2$. 
% \end{remark}

Motivated by the approximate optimality of {\BayesianOptimalUniformPricing} in the Bayesian mechanism design setting (\Cref{thm:BOM_BOUP}), we are interested in \emph{\IdentityPricing} that posts the realized sample $\firstsample\sim\firstprior$ from the first order statistic distribution as a take-it-or-leave-it price to all buyers. We present revenue approximation of {\IdentityPricing} against both {\BayesianOptimalUniformReserve} and {\BayesianOptimalMechanism}, both of which heavily rely on the knowledge of buyers distributions. All results in this section are summarized in \Cref{tab:IP vs BOUR} and \Cref{tab:IP vs BOM}.

\input{table/table_single_sample}

\xhdr{Revenue Approximation against $\BOUR$.}
We first compare {\IdentityPricing} and {\BayesianOptimalUniformReserve} in \Cref{thm:BOUR_IP}. We note that the revenue approximation $\calC_{\BOUR}^{\IP} = \frac{1}{2}$ of {\IdentityPricing} stated in this theorem is optimal among all deterministic mechanisms even for instances with a single regular buyer \cite[Table~1]{ABB22}. 

\begin{theorem}[$\BOUR$ vs.\ {\IP}]
\label{thm:BOUR_IP}
\begin{flushleft}
For asymmetric quasi-regular buyers,
% Given $n \geq 1$ many independent (but possibly asymmetric) quasi-regular buyers,
{\IdentityPricing} achieves a tight $\calC_{\BOUR}^{\IP} = \frac{1}{2}$-approximation to {\BayesianOptimalUniformReserve}.
% (This guarantee is the best possible among deterministic $1$-sample mechanisms and is tight even for a single regular buyer.)
\end{flushleft}
\end{theorem}

Here we sketch the high-level proof idea behind \Cref{thm:BOUR_IP}. In the first step, we characterize the expected revenue of both  {\BayesianOptimalUniformReserve} and {\IdentityPricing} as functions of the first order statistic distribution $\prior_{(1:n)}$ (\Cref{lem:BOUR revenue expression,fact:second order statistic lower bound,lem:IP revneue expression}). Equipped with those characterizations, we then divide the expected revenue of each mechanism into two pieces. Utilizing the Condition~\ref{condition:quasi-regular:revenue curve} of quasi-regularity established in \Cref{prop:q-regular equivalent definition} and the quasi-regularity guarantee of order statistic in \Cref{thm:order:quasi-regular}, we argue that for each piece, {\IdentityPricing} is a 2-approximation to {\BayesianOptimalUniformReserve}. Below we first list the technical lemmas used in the first step and then provide the formal proof of \Cref{thm:BOUR_IP}.

\begin{lemma}[\cite{CGM15}]
\label{lem:BOUR revenue expression}
    For $n\geq1$ asymmetric buyers and any price $\price\geq 0$, the expected revenue of {\SecondPriceAuction} with uniform reserve $\reserve = \price$ is 
    \begin{align*}
        \SPA_{\reserve = \price} = \price \cdot (1 - \cdf_{(1:n)}(\price)) + \displaystyle\int_{\price}^\infty (1 - \cdf_{(2:n)}(\val))\cdot \d\val
    \end{align*}
    where $\cdf_{(1:n)}$ and $\cdf_{(2:n)}$ are the distributions of the first and second order statistic, respectively.
\end{lemma}

\begin{fact}[\cite{JLTX20}]
\label{fact:second order statistic lower bound}
    Given $n\geq 1$ many independent (but possibly asymmetric) distributions~$\priors = \{\prior_i\}_{i\in[n]}$, for every value $\val\in\reals_+$,
    \begin{align*}
        \cdf_{(2:n)}(\val) \geq \cdf_{(1:n)}(\val) \cdot (1 - \ln \cdf_{(1:n)}(\val))
    \end{align*}
    where $\cdf_{(1:n)}$ and $\cdf_{(2:n)}$ are the distributions of the first and second order statistic, respectively.
\end{fact}

\begin{lemma}
\label{lem:IP revneue expression}
    For $n\geq1$ asymmetric buyers, the expected revenue of {\IdentityPricing} is 
    \begin{align*}
        \IP \geq \displaystyle\int_{0}^1 \revcurve_{(1:n)}(\quant)\cdot \d\quant
    \end{align*}
    where $\revcurve_{(1:n)}$ is the revenue curve induced by the first order statistic distribution $\prior_{(1:n)}$, i.e., $\revcurve_{(1:n)}(\quant) \triangleq \quant\cdot \cdf_{(1:n)}^{-1}(\quant)$ for every quantile $\quant\in[0, 1]$. When there is no point mass in distribution $\prior_{(1:n)}$, the above inequality holds with equality.
\end{lemma}
\begin{proof}
    Consider {\IdentityPricing} in the quantile space. Recall that drawing a sample from the a given distribution (e.g., the first order statistic distribution $\prior_{(1:n)}$) is equivalent to drawn a uniform quantile. Moreover, in {\IdentityPricing}, the purchase happens if and only if the first order statistic from the realized valuation profile is at least than the realized sample. Thus, by the definition of revenue curve $\revcurve_{(1:n)}$, the expected revenue of {\IdentityPricing} is at least (resp.\ exactly equal to) $\expect[\quant]{\revcurve_{(1:n)}(\quant)}=\int_{0}^1 \revcurve_{(1:n)}(\quant)\cdot \d\quant$ (resp.\ if there is no point mass in distribution $\prior_{(1:n)}$).
\end{proof}

Now we are ready to prove \Cref{thm:BOUR_IP}.

\begin{proof}[Proof of \Cref{thm:BOUR_IP}]
    In this analysis, with slight abuse of notations, we shorthand $\prior_{(1)}$, $\revcurve_{(1)}$ for $\prior_{(1:n)}$, $\revcurve_{(1:n)}$, respectively.
    Let $\reserve$ be the uniform reserve used in {\BayesianOptimalUniformReserve}. Invoking \Cref{lem:BOUR revenue expression} and \Cref{fact:second order statistic lower bound}, the expected revenue of {\BayesianOptimalUniformReserve} can be lower bounded by
    \begin{align*}
        \BOUR \leq \reserve \cdot (1 - \cdf_{(1)}(\reserve)) + \displaystyle\int_{\reserve}^\infty \left(1 - \cdf_{(1)}(\val) \cdot (1 - \ln \cdf_{(1)}(\val))\right)\cdot \d\val
    \end{align*}
    Let $\quant_{\reserve} = 1 - \cdf_{(1)}(\reserve)$ be the quantile of reserve $\reserve$ in the first order statistic distribution $\prior_{(1)}$. Invoking \Cref{lem:IP revneue expression}, the expected revenue of {\IdentityPricing} can be expressed as
    \begin{align*}
        \IP
        &{} \geq \displaystyle\int_{0}^1 \revcurve_{(1)}(\quant)\cdot \d\quant
        \\
        &{} =
        \int_0^{\quant_{\reserve}} \left(\revcurve_{(1)}(\quant) - \frac{\quant}{\quant_{\reserve}}\cdot \revcurve_{(1)}(\quant_{\reserve})
        \right)
        \cdot \d\quant 
        +
        \int_{\quant_{\reserve}}^1 \left(\revcurve_{(1)}(\quant) - \frac{1-\quant}{1-\quant_{\reserve}}\cdot \revcurve_{(1)}(\quant_{\reserve})
        \right) \cdot \d\quant
        +
        \frac{1}{2} \revcurve_{(1)}(\quant_{\reserve})
        \\
        &{} =
        \int_0^{\quant_{\reserve}} \left(\revcurve_{(1)}(\quant) - \quant \reserve
        \right)
        \cdot \d\quant 
        +
        \int_{\quant_{\reserve}}^1 \left(\revcurve_{(1)}(\quant) - \frac{1-\quant}{1-\quant_{\reserve}}\cdot \revcurve_{(1)}(\quant_{\reserve})
        \right) \cdot \d\quant
        +
        \frac{1}{2} \revcurve_{(1)}(\quant_{\reserve})
    \end{align*}
    where the second equality holds by algebra and the third equality holds since $\reserve=\frac{\revcurve_{(1)}(\quant_\reserve)}{\quant_{\reserve}}$ by construction. Putting two pieces together, it suffices to bound the approximation of {\IdentityPricing} against {\BayesianOptimalUniformReserve} by pieces as follows
    \begin{align}
    \label{eq:IP approximates BOUR part I}
        &\displaystyle\int_{\quant_{\reserve}}^1 \left(\revcurve_{(1)}(\quant) - \frac{1-\quant}{1-\quant_{\reserve}}\cdot \revcurve_{(1)}(\quant_{\reserve})
        \right) \cdot \d\quant
        +
        \frac{1}{2} \revcurve_{(1)}(\quant_{\reserve}) 
        \geq 
        \frac{1}{2} \cdot \reserve \cdot (1 - \cdf_{(1)}(\reserve))
        \\
    \label{eq:IP approximates BOUR part II}
        &\int_0^{\quant_{\reserve}} \left(\revcurve_{(1)}(\quant) - \quant\reserve
        \right)
        \cdot \d\quant 
        \geq 
        \frac{1}{2}\cdot \displaystyle\int_{\reserve}^\infty \left(1 - \cdf_{(1)}(\val) \cdot (1 - \ln \cdf_{(1)}(\val))\right)\cdot \d\val
    \end{align}
    
    \xhdr{Verify Inequality~\eqref{eq:IP approximates BOUR part I}.} Since all buyers are quasi-regular, \Cref{thm:order:quasi-regular} ensures the first order statistic distribution $\prior_{(1)}$ is also quasi-regular. Combining with Condition~\ref{condition:quasi-regular:revenue curve} for quasi-regularity in \Cref{prop:q-regular equivalent definition}, we have $\revcurve_{(1)}(\quant) - \frac{1-\quant}{1-\quant_{\reserve}}\cdot \revcurve_{(1)}(\quant_{\reserve}) \geq 0$ for every quantile $\quant\in[\quant_{\reserve}, 1]$. Thus, inequality~\eqref{eq:IP approximates BOUR part I} holds as
    \begin{align*}
        \displaystyle\int_{\quant_{\reserve}}^1 \left(\revcurve_{(1)}(\quant) - \frac{1-\quant}{1-\quant_{\reserve}}\cdot \revcurve_{(1)}(\quant_{\reserve})
        \right) \cdot \d\quant
        +
        \frac{1}{2} \revcurve_{(1)}(\quant_{\reserve}) 
        \geq 
        \frac{1}{2} \revcurve_{(1)}(\quant_{\reserve}) 
        =
        \frac{1}{2} \cdot \reserve \cdot (1 - \cdf_{(1)}(\reserve))
    \end{align*}
    where the equality is due to the definition of quantile $\quant_{\reserve}$ and revenue curve $\revcurve_{(1)}$.
    
    \xhdr{Verify Inequality~\eqref{eq:IP approximates BOUR part II}.} We prove this inequality for any $\reserve \geq 0$ (not necessarily the uniform reserve used in {\BayesianOptimalUniformReserve}) and any distribution (not necessarily quasi-regular). Moreover, to simplify the presentation, we use a differential element method argument. 

    Fix an infinitesimal $\detquant \geq 0$. 
    Let $\mathcal{Q}=\{\quant_\reserve - \ell\detquant \given \ell\in[0:\frac{\quant_\reserve}{\detquant}]\}$.
    For any quantile $\quant\in[0, 1]$,  let quantile $\quant\primed = \quant - \detquant$. Additionally, we denote by $\val = \frac{\revcurve_{(1)}(\quant)}{\quant}$ and $\val\primed = \frac{\revcurve_{(1)}(\quant\primed)}{\quant\primed} $ be their corresponding value (with respect to the first order statistic distribution $\prior_{(1)}$), respectively. Define $A(\quant)$ as 
    \begin{align*}
        A(\quant) \triangleq \displaystyle\int_{\val}^{\val\primed} \left(1 - \cdf_{(1)}(\val) \cdot (1 - \ln \cdf_{(1)}(\val))\right)\cdot \d\val
    \end{align*} 
    By definition, we have
    \begin{align*}
        \frac{1}{2}\cdot \sum_{\quant\in\calQ} A(\quant) = \text{right-hand side of inequality~\eqref{eq:IP approximates BOUR part II}}
    \end{align*}
    Meanwhile, note that the left-hand side of inequality~\eqref{eq:IP approximates BOUR part II} is essentially the area of the region between the area under revenue curve $\revcurve_{(1)}$ and the area above straight line cross origin with slope $\reserve$. We can partition this area into $|\calQ|$ strips. Specifically, for each quantile $\quant\in \calQ$, consider a strip that is the area enclosed by revenue curve $\revcurve_{(1)}$ and two straight lines cross origin with slopes $\val$ and $\val\primed$, respectively. (Recall $\val = {\revcurve_{(1)}(\quant)}/{\quant}$, $\val\primed = {\revcurve_{(1)}(\quant\primed)}/{\quant\primed}$ and $\quant\primed = \quant - \detquant$.) See \Cref{fig:BOUR_IP_analysis} for a graphical illustration. Let $B(\quant)$ be the area of this strip. By definition, we have
    \begin{align*}
        \sum_{\quant\in\calQ} B(\quant) = \text{left-hand side of inequality~\eqref{eq:IP approximates BOUR part II}}
    \end{align*}
    Putting two pieces together, to verify inequality~\eqref{eq:IP approximates BOUR part II}, it suffices to show for every quantile $\quant\in\calQ$:
    \begin{align*}
        \frac{B(\quant)}{A(\quant)} \geq \frac{1}{2} - o_{\detquant}(1)
    \end{align*}
    To see this, note that since $\detquant$ is sufficiently small, we can express $A(\quant)$ as 
    \begin{align*}
       A(\quant) 
       &{}= 
       (\val - \val\primed)\cdot  \left(1 - \cdf_{(1)}(\val) \cdot (1 - \ln \cdf_{(1)}(\val))\right)\cdot \left(1 \pm o_{\detquant}(1)\right)
       \\
       &{}= 
       \left(\frac{\revcurve_{(1)}(\quant) - \detquant \revcurve_{(1)}'(\quant)}{\quant - \detquant} - \frac{\revcurve_{(1)}(\quant)}{\quant}\right)\cdot  \left(1 - (1-\quant) \cdot (1 - \ln (1-\quant))\right)\cdot \left(1 \pm o_{\detquant}(1)\right)
       \\
       &{}=
       \frac{\revcurve_{(1)}(\quant) - \revcurve_{(1)}'(\quant)\cdot \quant}{\quant^2}\cdot \detquant
       \cdot  \left(1 - (1-\quant) \cdot (1 - \ln (1-\quant))\right)\cdot \left(1 \pm o_{\detquant}(1)\right)
    \end{align*}
    where the second equality holds by construction and the third equality holds by algebra.
    Similarly, we can express $B(\quant)$ as 
    \begin{align*}
        B(\quant) 
        &{}= 
        \frac{1}{2}\cdot \quant\primed \cdot \left(
        \revcurve_{(1)}(\quant\primed) - \val\quant\primed
        \right)
        +
        \frac{1}{2}\cdot \left(\quant - \quant\primed\right)
        \cdot \left(
        \revcurve(\quant\primed) - \val\quant\primed
        \right)\cdot (1 \pm o_{\detquant}(1))
        \\
        &{}=
        \frac{1}{2}\cdot \left(\revcurve_{(1)}(\quant) - \revcurve_{(1)}'(\quant)\cdot \quant\right)
        \cdot \detquant \cdot (1 \pm o_{\detquant}(1))
    \end{align*}
    where the first equality holds by definition (see \Cref{fig:BOUR_IP_analysis} for a graphical illustration), and the second equality holds by construction.
    Combining the above two expressions, we have
    \begin{align*}
        \frac{B(\quant)}{A(\quant)} 
        &{} =
        \frac{\frac{1}{2}\cdot \left(\revcurve_{(1)}(\quant) - \revcurve_{(1)}'(\quant)\cdot \quant\right)
        \cdot \detquant \cdot (1 \pm o_{\detquant}(1))}{\frac{\revcurve_{(1)}(\quant) - \revcurve_{(1)}'(\quant)\cdot \quant}{\quant^2}\cdot \detquant
       \cdot  \left(1 - (1-\quant) \cdot (1 - \ln (1-\quant))\right)\cdot \left(1 \pm o_{\detquant}(1)\right)}
       \\
       &{} = 
       \frac{\quant^2}{2(1 - \ln (1-\quant))}\cdot \left(1 \pm o_{\detquant}(1)\right)
       \\
       &{} \geq 
       \frac{1}{2}\cdot \left(1 \pm o_{\detquant}(1)\right)
    \end{align*}
    which finishes the argument for inequality~\eqref{eq:IP approximates BOUR part II}. Here the inequality holds since term $\frac{\quant^2}{2(1 - \ln (1-\quant))}$ is decreasing in $\quant$ and is equal to $\frac{1}{2}$ for $\quant = 1$.

    Finally, to see this tightness of the revenue approximation ratio, consider a single buyer with regular distribution $\prior(\val) = \frac{\val}{\val + 1}$. For this instance, the expected revenue in {\BayesianOptimalUniformReserve} is 1, while the expected revenue in {\IdentityPricing} is $\frac{1}{2}$. 
\end{proof}

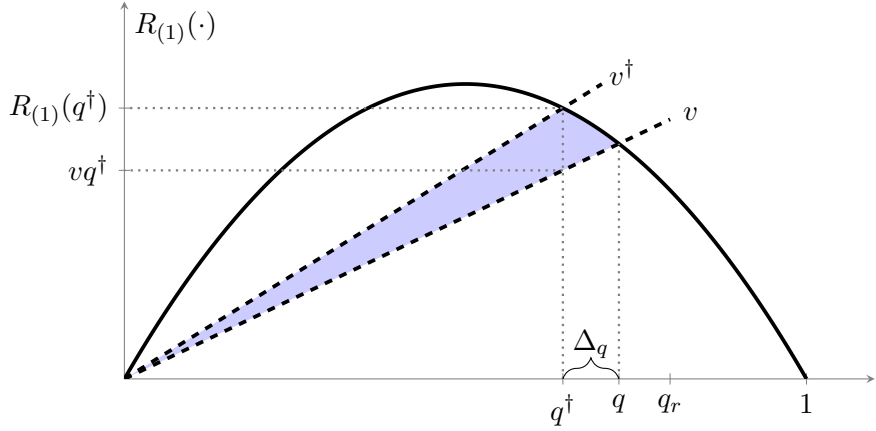
\begin{figure}
    \centering
    \input{figure/BOUR_IP_analysis}
    \caption{Graphical illustration of the analysis in \Cref{thm:BOUR_IP}. The blue area represents $B(\quant)$, which is used to characterize the left-hand side of inequality~\eqref{eq:IP approximates BOUR part II}.}
    \label{fig:BOUR_IP_analysis}
\end{figure}

For general instances with irregular buyers, the revenue gap between {\IdentityPricing} and {\BayesianOptimalUniformReserve} (or even {\BayesianOptimalUniformPricing}) can be unbounded.\footnote{To see this, consider a single-buyer instance, where the buyer has value $1$ with probability $\varepsilon$ for sufficiently small~$\varepsilon$.} However, the analysis of \Cref{thm:BOUR_IP} can be reused to obtain the following theorem which compares {\IdentityPricing} with {\SecondPriceAuction} for asymmetric irregular buyers.

\begin{theorem}[{\SPA} vs.\ {\IP}]
\label{thm:SPA_IP}
\begin{flushleft}
For arbitrary (independent) buyers,
{\IdentityPricing} achieves a tight $\calC_{\SPA}^{\IP} = \frac{1}{2}$-approximation to {\SecondPriceAuction}.
% (This guarantee is the best possible among deterministic $1$-sample mechanisms and is tight even for two i.i.d.\ regular buyers.)
\end{flushleft}
\end{theorem}
\begin{proof}
For the lower bound part, note that the argument for verifying inequality~\eqref{eq:IP approximates BOUR part II} (in the proof of \Cref{thm:BOUR_IP}) works for {\SecondPriceAuction} (with $\reserve = 0$) and general valuation distributions. For the upper bound, consider the following instance with two symmetric regular buyers. Let $\varepsilon$ be a sufficiently small constant. Both buyers have valuation distribution $\prior$ with cumulative density function $\cdf(\val) = \frac{\val - 1}{\varepsilon}$ and support $\supp(\prior) = [1, 1 + \varepsilon]$. Invoking \Cref{lem:BOUR revenue expression} and \Cref{lem:IP revneue expression}, the expected revenues of {\SecondPriceAuction} and {\IdentityPricing} are $1 + o_\varepsilon(1)$ and $\frac{1}{2} + o_{\varepsilon}(1)$, respectively. Thus, letting $\varepsilon$ approach zero finishes the proof.
\end{proof}

\xhdr{Revenue Approximation against $\BOM$.} Combining \Cref{thm:BOUR_IP} with the revenue approximations obtained in \Cref{sec:simple-mechanism:single-item}, we obtain the revenue approximations of {\IdentityPricing} against {\BayesianOptimalMechanism} for asymmetric quasi-regular/regular buyers (\Cref{thm:BOM_IP:asymmetric}), symmetric regular buyers (\Cref{thm:BOM_IP:iid:regular}), and symmetric quasi-regular buyers (\Cref{thm:BOM_IP:iid:quasi-regular}). 
It is worthwhile to highlight that \Cref{thm:BOM_IP:iid:regular} is a (nontrivial) {\em approximation-preserving} generalization of \cite{DRY15} from a single regular buyer to multiple symmetric regular buyers.

\begin{theorem}[{\BOM} vs.\ {\IP}]
\label{thm:BOM_IP:iid:regular}
\begin{flushleft}
For symmetric regular buyers,
{\IdentityPricing} achieves a tight $\calC_{\BOM}^{\IP} = \frac{1}{2}$-approximation to {\BayesianOptimalMechanism}.
(This guarantee is the best possible among deterministic sample-based mechanisms and is tight even for a single regular buyer.)
\end{flushleft}
\end{theorem}
\begin{proof}
    For symmetric regular buyers, {\BayesianOptimalUniformReserve} and {\BayesianOptimalMechanism} are the same. Thus, invoking \Cref{thm:BOUR_IP} finishes the lower bound part. For the upper bound part, note that the instance constructed in the proof of \Cref{thm:BOUR_IP} is also symmetric regular.
\end{proof}

\begin{theorem}[{\BOM} vs.\ {\IP}]
\label{thm:BOM_IP:asymmetric}
\begin{flushleft}
For asymmetric quasi-regular or regular buyers,
{\IdentityPricing} achieves a tight $\calC_{\BOM}^{\IP} \in [0.1385, 0.3750]$- or $\calC_{\BOM}^{\IP} \in [0.1908, 0.3750]$-approximation to {\BayesianOptimalMechanism}, respectively. 
\end{flushleft}
\end{theorem}

\begin{example}[Asymmetric Regular Instances for $\BOM$ vs.\ $\IP$]
\label{exp:single_sample:asymmetric:MA}
Fix any $a\geq 0$ and $\varepsilon > 0$. 
There are two regular asymmetric buyers. Buyer 1 has distribution $\prior_1$ with cumulative density function $\cdf_1(\val) = \frac{\val - 1}{\varepsilon}$ and support $\supp(\prior_1) = [1, 1 + \varepsilon]$. Buyer 2 has distribution $\prior_2$ with cumulative density function $\cdf_2(\val) = \frac{\val}{\val + a}$ and support $\supp(\prior_2) = [0, \infty)$. In this instance, $\frac{\IP}{\BOM} = \frac{a^2 + a + 1}{2(a+1)^2} + o_\varepsilon(1)$, which is minimized at $\frac{3}{8}= 0.375$ by setting $a = 1$ and letting $\varepsilon$ approach zero.\footnote{A slightly improved upper bound can be obtained by a carefully designed three-buyer instance with triangular distributions. We omit it since the improvement is negligible.}
\end{example}

\begin{proof}[Proof of \Cref{thm:BOM_IP:asymmetric}]
    For the lower bound part, note that $\calC_{\BOM}^{\IP}\geq \calC_{\BOUR}^{\IP} \cdot \calC_{\BOM}^{\BOUR}$. Therefore, invoking the results that $\calC_{\BOUR}^{\IP} = \frac{1}{2}$ (\Cref{thm:BOUR_IP}) and $\calC_{\BOM}^{\BOUR} \geq 0.2770$ (\Cref{thm:BOM_BOUR}) or $\calC_{\BOM}^{\BOUR} \geq 0.3817$ for asymmetric regular buyers \cite[Theorem~1]{JLQTX19} finishes the lower bound analysis.

    For the upper bound part, we analyze \Cref{exp:single_sample:asymmetric:MA}. The regularity of the constructed distributions can be easily checked by algebra. Below we verify the approximation ratio stated in this example.

    We first lower bound the expected revenue in {\BayesianOptimalMechanism}. Note that the seller can first sell the item to Buyer 2 with a take-it-or-leave-it price $H$, and get expected revenue $\frac{aH}{H + 1}$. With probability $\frac{1}{H + 1}$, Buyer 2 does not purchase and then the item can be sold to Buyer 1 with price~$1$. Letting price $H$ approach infinite, we obtain a lower bound of $1 + a$ for the expected revenue in {\BayesianOptimalMechanism}.

    For {\IdentityPricing}, by straightforward calculation, its expected revenue can be expressed as 
    \begin{align*}
        \IP = \frac{1}{2}\frac{a}{1+a}\left(\frac{a}{1 + a} + a\right)
        +
        \frac{1}{1+a}\left(\frac{1}{2} + \frac{1}{2}\frac{a}{1+a}\right)
        +
        o_\varepsilon(1)
    \end{align*}
    where the first term (up to $o_\varepsilon(1)$ additive difference) is the expected revenue contribution from the event that the first order statistic sample is from Buyer 2's distribution $\prior_2$, and the second term (up to $o_\varepsilon(1)$ additive difference) is the expected revenue contribution from the event that the first order statistic sample is from Buyer 1's distribution $\prior_1$.

    Putting all the pieces together, we obtain the approximation ratio stated in \Cref{exp:single_sample:asymmetric:MA}.
\end{proof}

\begin{theorem}[{\BOM} vs.\ {\IP}]
\label{thm:BOM_IP:iid:quasi-regular}
\begin{flushleft}
For symmetric quasi-regular buyers,
{\IdentityPricing} achieves a tight $\calC_{\BOM}^{\IP} \in[0.1385,0.3978]$-approximation to {\BayesianOptimalMechanism}.
(This guarantee is the best possible among deterministic $1$-sample mechanisms and is tight even for a single regular buyer.)
\end{flushleft}
\end{theorem}

\begin{example}[Symmetric Quasi-Regular Instances for $\BOM$ vs.\ $\IP$]
\label{exp:single_sample:iid:MA}
Fix any $a \geq 0$ and $\varepsilon > 0$. There are $n$ symmetric regular buyer with independent and identical valuation distributions
$\priors = \{\prior\}^{\otimes n}$, where distribution $\prior$ has support $\supp(\prior) = [1, \infty)$ and cumulative density function 
\begin{align*}
\cdf(\val) = \left\{
\begin{array}{ll}
 \frac{\val - 1}{\varepsilon}    &  \text{~~if $\val \leq \val\primed$}\\
 \frac{n\val}{n\val + a}    & \text{~~if $\val \geq \val\primed$}
\end{array}
\right.
\end{align*}
where $\val\primed$ is the unique positive solution of $\frac{\val - 1}{\varepsilon} = \frac{n\val}{n\val + a}$. In this instance, $\frac{\IP}{\BOM} = \auxfunc(a) + o_{n,\varepsilon}(1)$. Here auxiliary function $\auxfunc$ is defined as
\begin{align*}
    \auxfunc(a) = \frac{\displaystyle\int_0^{1 - e^{-a}}-\frac{a\quant}{\ln(1-\quant)}\cdot \d\quant + \frac{1}{2}\left(1 - \left(1 - e^{-a}\right)^2\right)}{a + 1}
\end{align*}
which attains its unique minimum $\min_{a\geq 0}\auxfunc(a) \approx 0.3978$ at $a\approx 0.6016$.
\ignore{\yfnote{\url{https://www.desmos.com/calculator/lbrzxi79cb}}}
\end{example}

\begin{proof}[Proof of \Cref{thm:BOM_IP:iid:quasi-regular}]
    The lower bound part is a direct implication of \Cref{thm:BOM_IP:asymmetric}. For the upper bound part, we analyze \Cref{exp:single_sample:iid:MA}. The regularity of the constructed distributions can be easily checked by algebra. Also see \Cref{fig:BOM_IP:iid} for a graphical illustration. Below we verify the approximation ratio stated in this example.

    We first compute the expected revenue of {\BayesianOptimalMechanism}. Let $\auxprior_{(n:n)}$ be the $n$-th order statistic of $n$ i.i.d.\ uniform distribution with support $[0, 1]$. As a sanity check, the quantile of largest value $\val_{(1:n)} \sim \prior$ drawn i.i.d.\ from distribution $\prior$ is exactly $\quant_{(n:n)}\sim \auxprior_{(n:n)}$. Let $\revcurve$ and $\ironrevcurve$ be the revenue curve and ironed revenue curve induced by distribution $\prior$. Invoking \Cref{prop:revenue_equivalence}, the expected revenue of {\BayesianOptimalMechanism} can be expressed as 
    \begin{align*}
        \BOM &{}= \expect
        {\ironrevcurve'(\quant_{(n:n)})}
        % \\
        % &{}
        =
        \expect{\ironrevcurve'(0)\cdot \indicator{\quant_{(0:0)} = 0}}
        +
        \expect{\ironrevcurve'(\quant_{(n:n)})\cdot \indicator{\quant_{(0:0)}>0}}
        % \\
        % &{}
        % \overset{(*)}{=}
        =
        a + 1 - \frac{a}{n}
    \end{align*}
    where all expectations are taken over $\quant_{(n:n)}\sim \auxprior_{(n:n)}$. The last equality holds by the facts that $\expect{\ironrevcurve'(0)\cdot \indicator{\quant_{(0:0)}=0}} = a$ (implied by $\revcurve(0) = \frac{a}{n}$) and $\ironrevcurve'(\quant) = 1 - \frac{a}{n}$ for every $\quant\in(0, 1)$.

    We next compute the expected revenue of {\IdentityPricing}. Recall that $\val\primed$ is the unique positive solution of $\frac{\val - 1}{\varepsilon} = \frac{n\val}{n\val + a}$. Let $\quant\primed \triangleq 1 - \cdf(\val\primed)$. By the construction of distribution $\prior$, revenue curve $\revcurve_{(1:n)}$ induced by the first order statistic distribution $\prior_{(1:n)}$ is
    \begin{align*}
        \forall \quant \in [0, 1 - (1 - \quant\primed)^n]:&
        \qquad
        \revcurve_{(1:n)}(\quant) = \frac{a}{n}\frac{(1-\quant)^{\frac{1}{n}}}{1 - (1-\quant)^{\frac{1}{n}}}\cdot \quant + o_\varepsilon(1)
        \\
        \forall \quant \in [1 - (1 - \quant\primed)^n, 1]:&
        \qquad
        \revcurve_{(1:n)}(\quant) = \quant + o_\varepsilon(1)
    \end{align*}
    Invoking \Cref{lem:IP revneue expression}, since the first order statistic distribution $\prior_{(1:n)}$ has no point mass, the expected revenue of {\IdentityPricing} is 
    \begin{align*}
        \IP 
        &{} = \displaystyle\int_{0}^{1}\revcurve_{(1:n)}(\quant)\cdot\d\quant
        \\
        &{} =
        \displaystyle\int_{0}^{1 - (1 - \quant\primed)^n}
        \frac{a}{n}\frac{(1-\quant)^{\frac{1}{n}}}{1 - (1-\quant)^{\frac{1}{n}}}\cdot \quant\cdot \d\quant 
        +
        \displaystyle\int_{1 - (1 - \quant\primed)^n}^1
        \quant\cdot\d\quant 
        +
        o_\varepsilon(1)
        \\
        &{} = 
        {\displaystyle\int_0^{1 - e^{-a}}-\frac{a\quant}{\ln(1-\quant)}\cdot \d\quant + \frac{1}{2}\left(1 - \left(1 - e^{-a}\right)^2\right)}
        +
        o_{n,\varepsilon}(1)
    \end{align*}
    where the last equality holds since $(1 - \quant\primed)^n = e^{-a} + o_{n,\varepsilon}(1)$ and $\frac{1}{n}\frac{(1-\quant)^{\frac{1}{n}}}{1 - (1-\quant)^{\frac{1}{n}}} = \frac{1}{-\ln(1-\quant)} + o_{n,\varepsilon}(1)$.

    Putting all the pieces together, we obtain the approximation ratio stated in \Cref{exp:single_sample:iid:MA}.
\end{proof}

\xhdr{Connection to Sample-Based Prophet Inequality.} We conclude this section by discussing the connection between our results and the prophet inequality literature \cite{Luc17,CFHOV21}. In the classic prophet inequality problem where the decision maker has full knowledge about the valuation distributions, there is a simple reduction based on virtual value interpretation that converts the revenue maximization objective into the welfare maximization objective \cite{Luc17}. 

In recent years, there has been a rapidly growing literature studying the sample-based prophet inequality problem, e.g., \cite{CDFS19,RWW20,GLT22,CCES24,DKLRS24,FLTWWZ24}. To the best of our knowledge, all prior works in this direction focus on the welfare maximization problem (except for the single period (agent) problem studied in \cite{DRY15}). Note that the aforementioned reduction from revenue maximization to welfare maximization cannot be applied in this model, since the virtual value calculation requires full knowledge about the valuation distributions. In this sense, our results (\Cref{thm:BOM_IP:asymmetric,thm:BOM_IP:iid:regular,thm:BOM_IP:iid:quasi-regular})
can be viewed as the first nontrivial competitive ratio guarantee for the revenue-maximization sample-based prophet inequality. 

In the single-item welfare-maximization sample-based prophet inequality, posting the first order statistic sample achieves the optimal competitive ratio of $\frac{1}{2}$ \cite{RWW20}. However, \Cref{exp:single_sample:asymmetric:MA} (with two asymmetric regular buyers) gives a competitive ratio upper bound of $\frac{3}{8}$ for posting  the first order statistic sample in the single-item revenue-maximization sample-based prophet inequality. This indicates that though revenue maximization and welfare maximization are essentially equivalent in the classic prophet inequality, there might exist a separation between these two objectives for the sample-based prophet inequality, and the revenue maximization is strictly harder (in terms of the competitive ratio).

%% file: table/table_single_sample.tex
\begin{table}[ht]
    {\centering\small
    \begin{tabular}{|l|>{\arraybackslash}p{4.45cm}|>{\arraybackslash}p{4.45cm}|>{\arraybackslash}p{4.15cm}|}
        \hline
        \rule{0pt}{13pt} & asymmetric multi-buyer & symmetric multi-buyer & single-buyer \\ [2pt]
        \hline
        \rule{0pt}{13pt}$\regulardistspace$ & \multicolumn{3}{c|}{\multirow{2}{*}{\rule{0pt}{15pt}$\mathrm{TB} = \frac{1}{2}$ \cite[Thm~3.2]{DRY15} {\bf [\Cref{thm:BOUR_IP}]}}} \\ [2pt]
        \cline{1-1}
        \rule{0pt}{13pt}$\quasiregulardistspace$ & \multicolumn{3}{c|}{} \\ [2pt]
        \hline
        \rule{0pt}{13pt}$\generaldistspace$ & \multicolumn{3}{c|}{no revenue guarantee} \\ [2pt]
        \hline
    \end{tabular}
    \par}
    \caption{{\IdentityPricing} vs.\ {\BayesianOptimalUniformReserve}, the revenue guarantees in various single-item settings. We mark our results in {\bf bold}.}
    \label{tab:IP vs BOUR}
    \vspace{.1in}
    {\centering\small
    \begin{tabular}{|l|>{\arraybackslash}p{4.45cm}|>{\arraybackslash}p{4.45cm}|>{\arraybackslash}p{4.15cm}|}
        \hline
        \rule{0pt}{13pt} & asymmetric multi-buyer & symmetric multi-buyer & single-buyer \\ [2pt]
        \hline
        \rule{0pt}{13pt}\multirow{2}{*}{\rule{0pt}{15pt}$\regulardistspace$} & $\mathrm{LB} \gtrsim 0.1908$ {\bf [\Cref{thm:BOM_IP:asymmetric}]} & \multirow{2}{*}{\rule{0pt}{15pt}$\mathrm{TB} = \frac{1}{2}$ {\bf [\Cref{thm:BOM_IP:iid:regular}]}} & \multirow{2}{*}{\rule{0pt}{15pt}$\mathrm{TB} = \frac{1}{2}$ \cite[Thm~3.2]{DRY15}} \\ [2pt]
        \rule{0pt}{13pt} & $\mathrm{UB} \lesssim 0.3750$ {\bf [\Cref{exp:single_sample:asymmetric:MA}]} & & \\ [2pt]
        \hline
        \rule{0pt}{13pt}\multirow{2}{*}{\rule{0pt}{15pt}$\quasiregulardistspace$} & $\mathrm{LB} \gtrsim 0.1385$ {\bf [\Cref{thm:BOM_IP:asymmetric}]} & $\mathrm{LB} \gtrsim 0.1385$ {\bf [implication]} & \multirow{2}{*}{\rule{0pt}{15pt}$\mathrm{TB} = \frac{1}{2}$ {\bf [implication]}} \\ [2pt]
        \rule{0pt}{13pt} & $\mathrm{UB} \lesssim 0.3750$ {\bf [implication]} & $\mathrm{UB} \lesssim 0.3978$ {\bf [\Cref{exp:single_sample:iid:MA}]} & \\ [2pt]
        \hline
        \rule{0pt}{13pt}$\generaldistspace$ & \multicolumn{3}{c|}{no revenue guarantee} \\ [2pt]
        \hline
    \end{tabular}
    \par}
    \caption{{\IdentityPricing} vs.\ {\BayesianOptimalMechanism}, the revenue guarantees in various single-item settings. We mark our results in {\bf bold}.}
    \label{tab:IP vs BOM}
\end{table}

%% file: figure/BOUR_IP_analysis.tex
\begin{tikzpicture}[scale=1, transform shape]
\begin{axis}[
axis line style=gray,
axis lines=middle,
% xlabel = $\quant$,
ylabel = $\revcurve_{(1)}(\cdot)$,
xtick={0, 0.642857, 0.725, 0.8, 1},
ytick={0, 0.176785675, 0.229591877551, 1},
xticklabels={0, $\quant\primed$, $\quant$, $\quant_\reserve$, 1},
yticklabels={0, $\val\quant\primed$, $\revcurve_{(1)}(\quant\primed)$, $1$},
xmin=0,xmax=1.1,ymin=-0.0,ymax=0.32,
width=0.7\textwidth,
height=0.4\textwidth,
samples=500]

\fill[blue!20] (0,0) -- (0.642857, 0.229591877551) -- (0.725, 0.199375) -- cycle;

\addplot[domain=0:1, black!100!white, line width=0.5mm] (x, {x * (1-x)});

\addplot[line width=0.5mm, dashed] (0, 0) -- (0.8, 0.22);

\addplot[line width=0.5mm, dashed] (0, 0) -- (0.7, 0.25);

\addplot[dotted, gray, line width=0.3mm] (0.725, 0) -- (0.725, 0.199375);% -- (0, 0.199375);

\addplot[dotted, gray, line width=0.3mm] (0.642857, 0.176785675) -- (0, 0.176785675);

\addplot[dotted, gray, line width=0.3mm] (0.642857, 0) -- (0.642857, 0.229591877551) -- (0, 0.229591877551);

\draw (0.73, 0.26) node {$\val\primed$};
\draw (0.83, 0.225) node {$\val$};

\draw[decorate,decoration={brace,amplitude=8pt}] (0.642857,0) -- (0.725,0) node[midway,above=4pt] {$\Delta_\quant$};

\end{axis}

\end{tikzpicture}

%% file: source/sample_complexity.tex
\section{The Sample Complexity of Revenue Maximization}
\label{sec:sample-complexity}

In this section, we revisit the sample complexity of revenue maximization. A rapid growing literature \cite{CR14, MR15, CGM15, MM16, RS16, DHP16,  GN17, S17, GHZ19, HT19, CHMY23, LSTW23, JLX23} has been established in this direction. We focus on two important mechanisms -- {\EmpiricalUniformReserve} and {\EmpiricalMyersonAuction} and study the number of value samples needed to obtain an $(1-\varepsilon)$-approximation under quasi-MHR and quasi-regular distributions.

Unlike the single-sample setting (\Cref{sec:single-sample}) where the goal is to design mechanisms with constant approximation, we address the more difficult task of {\em PAC-learnability of Bayesian mechanisms}.
To this end, the first-order-statistic sample access $\sample_{(1:n)} \sim \prior_{(1:n)}$ is insufficient and, thus, must be replaced with the {\em full sample access} $\samples \sim \priors$.

\subsection{Learning Empirical Uniform Reserve}
\label{sec:sample-complexity:EUR}

In this subsection, we present the sample complexity of {\EmpiricalUniformReserve}. Specifically, the goal is to learn a $(1 - \eps)$-approximation to {\BayesianOptimalUniformReserve} (that fully leverages the Bayesian information) using a finite number of full samples $\samples \sim \priors$. 

\begin{theorem}[The Sample Complexity of {\sf Empirical Uniform Reserve}]
\label{thm:sample-complexity:EUR}
\begin{flushleft}
% Given any $\eps \in (0, 1)$ and $n \geq 1$ many independent (but possibly asymmetric) buyers,
Given any $\eps \in (0, 1)$ and $n \geq 1$ many asymmetric buyers,
there exists an algorithm that accesses $m$ samples and, with probability at least $(1 - \delta)$, returns an $(1 - \eps)$-approximate {\EmpiricalUniformReserve}, when:
\begin{enumerate}[label=(\roman*)]
    \item the buyers are quasi-regular and $m = O(\eps^{-3} \log(\delta^{-1} \eps^{-1})) = \Tilde{O}(\eps^{-3})$; or
    \item the buyers are quasi-MHR and $m = O(\eps^{-2} \log(\delta^{-1} \eps^{-1})) = \Tilde{O}(\eps^{-2})$.
\end{enumerate}
\end{flushleft}
\end{theorem}

\Cref{thm:sample-complexity:EUR} is a \emph{sample-complexity-preserving} generalization of \cite{JLX23} from MHR (resp.\ regular) distributions to quasi-MHR (resp.\ quasi-regular) distributions. Moreover, the optimal sample complexity bound of quasi-regular is tight (up to logarithm factor), where the lower bound is achieved by regular distributions \cite{HMR18,GHZ19}.
Both the algorithm and the analysis follow \cite{JLX23} with the observation that the distributional properties (i.e., probability tail bounds) of MHR and regular distributions used in the prior work also hold for quasi-MHR and quasi-regular distributions. 
\begin{proof}[Proof (Sketched) of \Cref{thm:sample-complexity:EUR}]
The algorithm in \cite[Section~3]{JLX23} works for general buyers;\footnote{Indeed, the algorithm in \cite[Section~3]{JLX23} works even for {\em correlated} buyers.}
different distributional assumptions are utilized only to establish different sample complexity bounds.

To obtain the $\Tilde{O}(\eps^{-3})$ sample complexity bound for quasi-regular buyers, consider the analysis in \cite{JLX23} of the same bound for regular buyers. In their analysis, the regularity condition is only used in \cite[Fact~4.13]{JLX23}, which can be replaced with Condition~\ref{condition:quasi-regular:cdf} (in combination with \Cref{thm:order:quasi-regular}).

To obtain the $\Tilde{O}(\eps^{-2})$ sample complexity bound for quasi-MHR buyers, consider the analysis in \cite{JLX23} of the same bound for MHR buyers. In their analysis, the MHR condition is only used in \cite[Lemma~4.16, Item~1]{JLX23}, which can be replaced with Condition~\ref{condition:quasi-MHR:cdf} (in combination with \Cref{thm:order:quasi-mhr}).
% 
% This finishes the proof of \Cref{thm:sample-complexity:EUR}.
\end{proof}

\subsection{Learning Empirical Myerson Auction}
\label{sec:sample-complexity:BOM}

In this subsection, we present the sample complexity of {\EmpiricalMyersonAuction}. Specifically, the goal is to learn a $(1 - \eps)$-approximation to {\BayesianOptimalMechanism} (that fully leverages the Bayesian information) using a finite number of full samples $\samples \sim \priors$. We first discuss the result for quasi-MHR distributions and then discuss the result for quasi-regular distributions.

\xhdr{Quasi-MHR Buyers.} We first present the result for quasi-MHR buyers.
\begin{theorem}[The Sample Complexity of {\EmpiricalMyersonAuction} for Quasi-MHR Buyers]
\label{thm:sample-complexity:EMA:q-MHR}
\begin{flushleft}
% Given any $\eps \in (0, 1)$ and $n \geq 1$ many independent (but possibly asymmetric) buyers,
Given any $\eps \in (0, 1)$ and $n \geq 1$ many asymmetric quasi-MHR buyers,
there exists an algorithm that accesses $m = O(n \eps^{-2} \log(\eps^{-1}) \log(n \eps^{-1}) \log(n \delta^{-1} \eps^{-1})) = \Tilde{O}(n \eps^{-2})$ samples and, with probability at least $(1 - \delta)$, returns an $(1 - \eps)$-approximate {\EmpiricalMyersonAuction}.
% , when:
% \begin{enumerate}[label=(\roman*)]
%     \item the buyers are quasi-regular and $m = O(n \eps^{-3} \log(n \eps^{-1}) \log(n \delta^{-1} \eps^{-1})) = \Tilde{O}(n \eps^{-3})$; or
%     \item the buyers are quasi-MHR and $m = O(n \eps^{-2} \log(\eps^{-1}) \log(n \eps^{-1}) \log(n \delta^{-1} \eps^{-1})) = \Tilde{O}(n \eps^{-2})$.
% \end{enumerate}
\end{flushleft}
\end{theorem}

\Cref{thm:sample-complexity:EMA:q-MHR} is a \emph{sample-complexity-preserving} generalization of \cite{GHZ19} from MHR distributions to quasi-MHR distributions. 
% Moreover, this sample complexity bound of quasi-MHR is tight (up to logarithm factor), where the lower bound are achieved by MHR distributions \cite{GHZ19}.
Both the algorithm and the analysis follow \cite{GHZ19} with the observation that the following distributional property of MHR distributions used in the prior work also hold for quasi-MHR distributions. 

\begin{lemma}[Truncations of Quasi-MHR Distributions]
\label{lem:truncation:quasi-mhr}
% \begin{flushleft}
% Given any $\eps \in (0, 1)$ and $n \geq 1$ many independent (but possibly asymmetric) quasi-MHR distributions $\priors = \{\prior_{i}\}_{i \in [n]}$,
Given any $\eps \in (0, 1)$ and $n \geq 1$ many asymmetric quasi-MHR distributions $\priors = \{\prior_{i}\}_{i \in [n]}$, 
consider distributions $\Tilde{\priors} = \{\Tilde{\prior}_{i}\}_{i \in [n]}$ derived from truncating $\priors = \{\prior_{i}\}_{i \in [n]}$ to the bounded interval $[0, \threshold]$, where the truncation point $\threshold \triangleq 9\ln(2 / \eps) \cdot {\sf BOM}(\priors)$. The expected revenue of {\BayesianOptimalMechanism} under truncated distributions $\Tilde{\priors}$ is an $(1-\varepsilon)$-approximation to the expected revenue under original distributions $\priors$, i.e., $\BOM(\Tilde{\priors}) \geq (1 - \eps) \cdot \BOM(\priors)$.
% \end{flushleft}
\end{lemma}

\begin{proof}[Proof of \Cref{lem:truncation:quasi-mhr}]
Without loss of generality, we normalize the distributions $\priors = \{\prior_{i}\}_{i \in [n]}$ such that $\BOM(\priors) = 1$. Note the expected revenue of {\BayesianOptimalMechanism} under distributions $\priors$ can be decomposed by 
\begin{align*}
    \BOM(\priors) = \expect[\vals\sim\priors]{\BOM(\vals)\cdot \indicator{\vals\in[0, \threshold]^n}}
    +
    \expect[\vals\sim\priors]{\BOM(\vals)\cdot \indicator{\vals\not\in[0, \threshold]^n}}
\end{align*}
Due to the construction of truncated distribution $\tilde{\priors}$, we have
\begin{align*}
    \expect[\vals\sim\priors]{\BOM(\vals)\cdot \indicator{\vals\in[0, \threshold]^n}}
    \leq 
    \BOM(\tilde\priors)
\end{align*}
Meanwhile, we have
\begin{align*}
    \expect[\vals\sim\priors]{\BOM(\vals)\cdot \indicator{\vals\not\in[0, \threshold]^n}}
    &{}\leq 
    \expect[\vals\sim\priors]{\val_{(1:n)}\cdot \indicator{\vals\not\in[0, \threshold]^n}}
    =
    \displaystyle\int_\threshold^\infty 
    \val\cdot\d\cdf_{(1:n)}(\val)
\end{align*}
We next upper bound $\cdf_{(1:n)}(\val)$ for all value $\val \geq \threshold$. Since the normalize that $\BOM(\priors) = 1$, the expected revenue of posting uniform price $\price = e$ is at most 1, i.e., $e \cdot (1-\cdf_{(1:n)}(e)) \leq \BOM(\priors) \leq 1$, and thus $\cdf_{(1:n)}(e) \geq 1 - \frac{1}{e}$. Since all buyers are quasi-regular, by \Cref{thm:order:quasi-mhr}, the first order statistic distribution $\prior_{(1:n)}$ is also quasi-regular. Invoking Condition~\ref{condition:quasi-MHR:cdf} in \Cref{prop:q-MHR equivalent definition} with anchor value $\val = e$, we have $\cdf_{(1:n)}(\val\primed) \geq 1 - e^{-{\val\primed}/{e}}$ for every value $\val\primed \geq \threshold = 9\ln(2/\varepsilon) > e$. Therefore,
\begin{align*}
    \expect[\vals\sim\priors]{\BOM(\vals)\cdot \indicator{\vals\not\in[0, \threshold]^n}} \leq 
    \displaystyle\int_\threshold^\infty 
    \val\cdot\d\cdf_{(1:n)}(\val)
    \leq 
    \displaystyle\int_\threshold^\infty 
    \val\cdot\d(1-e^{-\val/e})
    \leq \varepsilon
    = 
    \varepsilon \cdot 
    \expect[\vals\sim\priors]{\BOM(\vals)}
\end{align*}
where the last inequality holds by algebra, and the last equality holds since the normalization that $\expect[\vals\sim\priors]{\BOM(\vals)} = 1$.

Combining all results above, we obtain
\begin{align*}
    \BOM(\tilde\priors) \geq (1 - \varepsilon) \cdot \expect[\vals\sim\priors]{\BOM(\vals)}
\end{align*}
which finishes the lemma analysis as desired.
\end{proof}

Now we sketch the proof of \Cref{thm:sample-complexity:EMA:q-MHR}.
\begin{proof}[Proof (Sketched) of \Cref{thm:sample-complexity:EMA:q-MHR}]
The algorithm in \cite[Section~3.1]{GHZ19} works for general distributions, and 
different distributional assumptions are utilized only to establish different sample complexity bounds. Specifically, in \cite{GHZ19}'s proof of the bound $\Tilde{O}(n \eps^{-2})$ for MHR buyers, the MHR condition is only used in \cite[Lemma~17]{GHZ19} (which summarizes the extreme value theorems by \cite{CD15,DHP16,MR15}); that lemma considers certain truncations of MHR distributions and upper bounds the revenue loss in {\sf Bayesian Myerson Auction}.
\Cref{lem:truncation:quasi-mhr} generalizes it to quasi-MHR distributions and is a surrogate to obtain the same bound $\Tilde{O}(n \eps^{-2})$ for quasi-MHR buyers.
\end{proof}

\xhdr{Quasi-Regular Buyers.} As we mentioned in the previous part (for quasi-MHR buyers), the algorithm developed in \cite[Section~3.1]{GHZ19} works for general distributions. To obtain the sample complexity bound, the authors consider certain truncations of the original distributions and bound the revenue losses due to this truncation utilizing different distributional properties of MHR or regular distributions. For quasi-regular buyers, we conjecture that the revenue loss of truncating quasi-regular distributions has the same order as the loss of truncating regular distributions \cite[Lemma~6.6]{DHP16}. We formalize this conjecture in \Cref{lem:truncation:quasi-regular} below. We prove this conjecture for single-buyer instances (i.e., $n = 1$) and leave the proof for the multi-buyer instances as a future direction.

\begin{conjecture}[Truncations of Quasi-Regular Distributions]
\label{lem:truncation:quasi-regular}
% \begin{flushleft}
% Given any $\eps \in (0, 1)$ and $n \geq 1$ many independent (but possibly asymmetric) quasi-regular distributions $\priors = \{\prior_{i}\}_{i \in [n]}$,
Given any $\eps \in (0, 1)$ and $n \geq 1$ many asymmetric quasi-regular distributions $\priors = \{\prior_{i}\}_{i \in [n]}$, 
consider distributions $\Tilde{\priors} = \{\Tilde{\prior}_{i}\}_{i \in [n]}$ derived from truncating $\priors = \{\prior_{i}\}_{i \in [n]}$ to the bounded interval $[0, \threshold]$, where the truncation point $\threshold \triangleq (1 / \eps) \cdot \BOM(\priors)$. The expected revenue of {\BayesianOptimalMechanism} under truncated distributions $\Tilde{\priors}$ is an $(1-\varepsilon)$-approximation to the expected revenue under original distributions $\priors$, i.e., $\BOM(\Tilde{\priors}) \geq (1 - O(\eps)) \cdot \BOM(\priors)$.
% \end{flushleft}
\end{conjecture}

\begin{proof}[Proof of \Cref{lem:truncation:quasi-regular} for $n = 1$.]
Since we are considering the single-buyer case, we drop the buyer index $i$ for all notations in this argument. Without loss of generality, we normalize the distribution $\prior$ such that $\max_{\price} \price(1 - \cdf(\price)) = \BOM(\prior) = 1$. 

Let $\optreserve$ be the monopoly reserve of distribution $\prior$. If truncation point $\threshold \geq \optreserve$, $\BOM(\Tilde{\prior}) = \BOM(\prior)$ and thus the theorem statement holds. In the remainder of the analysis, we assume truncation point $\threshold < \optreserve$. Since $\BOM(\prior) = 1$, $\cdf(\optreserve) = 1 - 1/\optreserve$. Since distribution $\prior$ is quasi-regular, invoking Condition~\ref{condition:quasi-regular:cdf} in \Cref{prop:q-regular equivalent definition} with anchor value $\val = \optreserve$, we have $\cdf(\threshold) \leq \frac{\threshold}{\threshold + \optreserve/(\optreserve - 1)}$ since truncation point $\threshold< \optreserve$. Thus, the expected revenue under the truncated distribution can be lower bounded as 
\begin{align*}
    \BOM(\tilde\prior)
    \geq
    \threshold\left(1 - \tilde\cdf(\threshold)\right)
    =
    \threshold(1 - \cdf(\threshold))
    \geq \threshold \left(
    1 - \frac{\threshold}{\threshold + \frac{\optreserve}{\optreserve - 1}}
    \right)
    \geq 1 - \frac{\eps}{1 - \eps}
    =
    (1 - O(\eps))\cdot \BOM(\prior)
\end{align*}
where the last inequality holds by algebra and the construction that $\threshold = \frac{1}{\eps}$, and the last equality holds since the normalization that $\expect[\vals\sim\priors]{\BOM(\vals)} = 1$.
\end{proof}

The following sample complexity bound for quasi-regular buyers holds under \Cref{lem:truncation:quasi-regular}. 

\begin{theorem}[The Sample Complexity of {\EmpiricalMyersonAuction} for Quasi-Regular Buyers]
\label{thm:sample-complexity:EMA:q-regular}
\begin{flushleft}
Assuming \Cref{lem:truncation:quasi-regular} holds, given any $\eps \in (0, 1)$ and $n \geq 1$ many asymmetric quasi-regular buyers,
there exists an algorithm that accesses $m = O(n \eps^{-3} \log(n \eps^{-1}) \log(n \delta^{-1} \eps^{-1})) = \Tilde{O}(n \eps^{-3})$ samples and, with probability at least $(1 - \delta)$, returns an $(1 - \eps)$-approximate {\EmpiricalMyersonAuction}.
\end{flushleft}
\end{theorem}
\begin{proof}[Proof (Sketched) of \Cref{thm:sample-complexity:EMA:q-regular}]
The algorithm in \cite[Section~3.1]{GHZ19} works for general distributions, and 
different distributional assumptions are utilized only to establish different sample complexity bounds. Specifically, in \cite{GHZ19}'s proof of the bound $\Tilde{O}(n \eps^{-3})$ for regular buyers, the regularity condition is only used in \cite[Lemma~16]{GHZ19} (which directly invokes \cite[Lemma~6.6]{DHP16}); that lemma considers certain truncations of regular distributions and upper bounds the revenue loss in {\sf Bayesian Myerson Auction}.
\Cref{lem:truncation:quasi-regular} generalizes it to quasi-regular distributions and is a surrogate to obtain the same bound $\Tilde{O}(n \eps^{-3})$ for quasi-regular buyers.
\end{proof}

\begin{comment}

In more details, in their proof of the bound $\Tilde{O}(n \eps^{-3})$ for regular buyers, the regularity condition is only used in \cite[Lemma~16]{GHZ19} (which just invokes \cite[Lemma~6.6]{DHP16}); that lemma considers certain truncations of regular distributions and upper bounds the revenue loss in {\sf Bayesian Myerson Auction}.
Soon after, we will generalize it to the following \Cref{lem:truncation:quasi-regular}.
Hence, using our \Cref{lem:truncation:quasi-regular} as a surrogate gives the same bound $\Tilde{O}(n \eps^{-3})$ for quasi-regular buyers.

Likewise, in their proof of the bound $\Tilde{O}(n \eps^{-2})$ for MHR buyers, the MHR condition is only used in \cite[Lemma~17]{GHZ19} (which summarizes the extreme value theorems by \cite{CD15,DHP16,MR15}); that lemma considers certain truncations of MHR distributions and upper bounds the revenue loss in {\sf Bayesian Myerson Auction}.
Soon after, we can generalize it to the following \Cref{lem:truncation:quasi-mhr}.
Thus, using our \Cref{lem:truncation:quasi-mhr} as a surrogate gives the same bound $\Tilde{O}(n \eps^{-2})$ for quasi-regular buyers.

This finishes the proof of \Cref{thm:sample-complexity:EMA}.

In the rest of this subsection, we accomplish the proofs of \Cref{lem:truncation:quasi-regular,lem:truncation:quasi-mhr}, one by one.

\begin{proof}[Proof of \Cref{lem:truncation:quasi-regular}]
\yj{I will write this proof.}

This finishes the proof of \Cref{lem:truncation:quasi-regular}.
\end{proof}

\end{comment}

%% file: source/conclusion.tex
\section{Conclusions and Future Directions}
\label{sec:conclusion}

In this work, we introduce the families of quasi-regular distributions $\quasiregulardistspace$, and quasi-MHR distributions $\quasimhrdistspace$, as generalizations of the classic regular $\regulardistspace$ and MHR distributions $\mhrdistspace$, respectively. Many structural properties of the original distribution families $\regulardistspace$, $\mhrdistspace$, such as probability tail bounds, extend naturally to these new families $\quasiregulardistspace$, $\quasimhrdistspace$. Additionally, we demonstrate that while the original families are closed under the operation of the ``i.i.d.\ order statistics'' operation, the new families maintain closure under the ``order statistic'' operation.

Leveraging these structural properties, we extend various important and seminal results in single-parameter revenue maximization---revenue approximation by Bayesian simple mechanisms, by prior-independent mechanisms, by a single sample, and the sample-complexity of revenue maximization---to the new distribution families. Many of these results are approximation-preserving generalizations, requiring technically challenging analysis, and some also advance the state-of-the-art even for the original distribution family $\regulardistspace$.

There are many interesting directions for future research. First, an important next step would be to rigorously prove all conjectures posed in this work, particularly those that specify the tight bounds for various approximations. See \Cref{conj:BOM_BOUP,conj:BOM_BOUR,lem:truncation:quasi-regular} and \Cref{tab:intro:simple-mechanism,tab:intro:duplicate,tab:intro:single-sample,tab:intro:sample-complexity} for details. Second, it would also be worthwhile to investigate whether other important revenue approximation results in the literature can be extended to these new distribution families. For instance, it would be valuable to extend the 2-approximation results of {\nDuplicateVCGAuction} and {\BayesianMonopolyReserves} from regular buyers, as established in \cite[Theorems~3.7 and 4.4]{HR09}, to quasi-regular buyers within matroid environments, and to adapt the randomization techniques developed in \cite{ABB22} to improve the revenue approximation by a single sample for quasi-regular buyers. Third, leveraging the structural results on order statistics developed in this work, it may be possible to achieve improved (or even tight) competition complexity bounds \cite{EFFTW17a} for multi-item auctions with unit-demand buyers. Finally, it would be worthwhile to explore whether these new distribution families exhibit additional desirable economic and mathematical properties. For example, are the new distribution families close under the ``sum'' or ``i.i.d.\ sum" operations?

%% file: source/prophet_inequality.tex
\newcommand{\cost}{c}

\section{Cost Prophet Inequalities}
\label{sec:prophet}

In addition to the Bayesian mechanism design problem, distributional assumptions such as regularity or the MHR condition are also crucial in other Bayesian models. In this section, we revisit the \emph{cost prophet inequality problem} \cite{LM24} and demonstrate how an approximation-preserving extension can be made from MHR distributions to quasi-MHR distributions.

\xhdr{Model.} In the stationary cost prophet inequality problem, there are $n$ rounds. In each round $i \in [n]$, a random cost $\cost_i$ is drawn i.i.d.\ from distribution $\prior$. Upon observing the realized cost $\cost_i$, the decision maker must immediately and irrevocably decide whether to accept this cost and stop, or skip it and move to the next round. If it is the final round $n$, the decision maker must accept the realized cost $\cost_n$. The objective is to design an online algorithm that minimizes the competitive ratio, which is defined as the ratio between the expected cost accepted by the online algorithm and the expected minimum cost (i.e., the prophet’s expected cost).

\xhdr{Results and Analysis.}
Before presenting the result, we first revisit the meta-theorem developed in \cite{LM24}.\footnote{The work \cite{LM24} claimed their Theorem~3.3 for a large enough $n \in \naturals$. However, after checking their proofs step by step, we have confirmed that their Theorem~3.3 actually holds for every $n \in \naturals$.}

\begin{theorem}[{\cite[Theorem~3.3]{LM24}}]
\label{thm:prophet:meta-theorem}
%     In the stationary prophet inequality problem, for any entire\footnote{See \cite[Definitions 2.3 and 2.4]{LM24} for the formal definition of the Puiseux series expansion and the entire property. We omit their formal definitions as we use this meta-theorem in a black-box fashion.} distribution $\prior$ with cumulative hazard rate function $\cumhazardrate$, which has a Puiseux series expansion $\cumhazardrate(\cost) = \sum_{\ell\in\naturals}a_\ell\cdot \cost^{d_{\ell}}$,
%     for large enough $n\in\naturals$,
%     there exists a $\lambda(d_1)$-competitive online algorithm, where
%     \begin{align*}
%         \lambda(d_1) \triangleq \frac{(1 + 1 / d_1)^{1 / d_1}}{\int_0^{\infty} y^{1/d_1}\cdot e^{-y}\cdot \d y}
%     \end{align*}
% \end{theorem}
In the stationary prophet inequality problem, for any entire\footnote{See \cite[Definitions 2.3 and 2.4]{LM24} for the formal definition of the Puiseux series expansion and the entire property. We omit their formal definitions here as we apply this meta-theorem in a black-box fashion.} distribution $\prior$ with cumulative hazard rate function $\cumhazardrate$, which has a Puiseux series expansion $\cumhazardrate(\cost) = \sum_{\ell\in\naturals}a_\ell \cdot \cost^{d_{\ell}}$,\ignore{ for large enough $n \in \naturals$,} there exists a $\lambda(d_1)$-competitive online algorithm, where 
    \begin{align*}
        \lambda(d_1) \triangleq \frac{(1 + 1 / d_1)^{1 / d_1}}{\int_0^{\infty} y^{1/d_1}\cdot e^{-y}\cdot \d y}
    \end{align*}
\end{theorem}

The main result of this section is as follows.
\begin{theorem}
\label{thm:prophet:quasi-mhr}
    In the stationary prophet inequality problem, for quasi-MHR and entire distribution~$\prior$,\ignore{ for large enough $n \in \naturals$,} there exists a 2-competitive online algorithm.
\end{theorem}
\cite{LM24} proves the special case of MHR and entire distributions using \Cref{thm:prophet:meta-theorem}. They also provide an asymptotically matching lower bound by considering exponential distributions, which are MHR and therefore quasi-MHR.

\begin{proof}[Proof of \Cref{thm:prophet:quasi-mhr}]
    Since distribution $\prior$ is quasi-MHR, for every cost $\cost\in[0, \infty)$
    \begin{align*}
        \left(\frac{\cumhazardrate(\cost)}{\cost}\right)'
        \geq 0
    \end{align*}
    Replacing the cumulative hazard rate function $\cumhazardrate$ with its Puiseux series expansion, we obtain
    \begin{align*}
        \left(\frac{\sum\nolimits_{\ell\in\naturals}a_\ell \cdot \cost^{d_{\ell}}}{\cost}\right)'
        =
        \sum\limits_{\ell\in\naturals}
        a_\ell\cdot (d_{\ell} - 1)\cdot\cost ^{d_{\ell}-2}
        \geq 0
    \end{align*}
    By \cite[Observation 2.3]{LM24}, since the distribution $\prior$ is entire, we have that $a_1 > 0$ and $d_\ell$'s are all positive and increasing in $\ell$. If $d_1 = 1$, note that $\lambda(1) = 2$, and thus invoking \Cref{thm:prophet:meta-theorem} directly proves the theorem statement. Therefore, it remains to analyze the case where $d_1 \neq 1$. In this case, to ensure that $\sum_{\ell \in \naturals} a_\ell \cdot (d_\ell - 1) \cdot \cost^{d_\ell - 2} \geq 0$ for sufficiently small cost $\cost > 0$ such that the first term, $a_1(d_1 - 1) \cdot \cost^{d_1 - 2}$, dominates the sum $\sum_{\ell \geq 2} a_\ell \cdot (d_\ell - 1) \cdot \cost^{d_\ell - 2}$, we conclude that $d_1 > 1$, since $a_1 > 0$, as argued earlier. Thus, the competitive ratio is $\lambda(d_1) \leq \lambda(1) = 2$, which completes the analysis as desired.
\end{proof}